\documentclass[11pt,a4paper]{article}
\usepackage{amsfonts}
\usepackage[dvips]{graphicx}
\usepackage{amsmath}
\usepackage{mathrsfs}
\usepackage{makeidx}
\usepackage{xypic}
\usepackage[nottoc]{tocbibind}
\usepackage{pstricks}
\usepackage{bbm}
\usepackage[arrow, matrix, curve]{xy}
\usepackage{amsthm}
\usepackage{amssymb}
\usepackage{epic}
\usepackage{eepic}
\usepackage{times}
\usepackage{epsfig}

\xymatrixcolsep{0.5cm}
\xymatrixrowsep{0.5cm}
\newtheorem{theorem}{Theorem}[section]
\newtheorem{proposition}{Proposition}[section]
\newtheorem{lemma}{Lemma}[section]

\newtheorem{corollary}{Corollary}[section]

\newcommand{\beqa}{\begin{eqnarray}}
\newcommand{\eeqa}{\end{eqnarray}}

\makeindex
\numberwithin{equation}{section}

\begin{document}

\begin{flushright}
YITP-SB-12-18
\end{flushright}

\bigskip \vspace{15pt}

\begin{center}
\textbf{{\Large Non-diagonal open spin-1/2 XXZ quantum chains by separation
of variables:} } \vspace{2pt}

\textbf{{\Large Complete spectrum and matrix elements of some quasi-local
operators} } \vspace{45pt}

{\large G.~Niccoli}\footnote{%
YITP, Stony Brook University, New York 11794-3840, USA,
niccoli@max2.physics.sunysb.edu}

\vspace{50pt}

\vspace{50pt}
\end{center}
\begin{itemize}
\begin{small}
\item[] \textbf{Abstract}\,\,\,The integrable quantum models, associated to the transfer matrices of the 6-vertex reflection algebra for spin 1/2 representations, are studied in this paper. In the framework of Sklyanin's quantum separation of variables (SOV), we provide the complete characterization of the eigenvalues and eigenstates of the transfer matrix and the proof of the simplicity of the transfer matrix spectrum. Moreover, we use these integrable quantum models as further key examples for which to develop a method in the SOV framework to compute matrix elements of local operators. This method has been introduced first in \cite{ADMFKGMN12-SG} and then used also in \cite{ADMFKN12-0}, it is based on the resolution of the quantum inverse problem (i.e. the reconstruction of all local operators in terms of the quantum separate variables) plus the computation of the action of \textit{separate covectors} on \textit{separate vectors}. In particular, for these integrable quantum models, which in the homogeneous limit reproduce the open spin-1/2 XXZ quantum chains with non-diagonal boundary conditions, we have obtained the SOV-reconstructions for a class of quasi-local operators and determinant formulae for the covector-vector actions. As consequence of these findings we provide one determinant formulae for the matrix elements of this class of reconstructed quasi-local operators on transfer matrix eigenstates.
\end{small}
\end{itemize}

\vspace{1cm} \parbox{12cm}{\small }\newpage

\newpage

\tableofcontents
\newpage

\section{Introduction}

In this paper we analyze the lattice quantum integrable models characterized
in the quantum inverse scattering method (QISM) \cite{ADMFKSF78}-\cite{ADMFKIK82} by
(boundary) monodromy matrices which satisfy the reflection algebra \cite{ADMFKGau71}-\cite{ADMFKGZ94} w.r.t. the 6-vertex R-matrix solution of the
Yang-Baxter equation. The prototypical elements in this class of quantum integrable models are the open XXZ spin-1/2 quantum chains, reproduced under the homogeneous limit of the representations of this reflection algebra on the spin-1/2 quantum chains. In the special representations corresponding to the open XXZ chain with diagonal boundary matrices\footnote{Numerical $2\times 2$ matrix solutions of the reflection equations, located at the ending points of the quantum chain.} the system has been
largely analyzed in the framework of the algebraic Bethe ansatz (ABA) \cite{ADMFKSF78}-\cite{ADMFKFST80} with results going from the spectrum\footnote{See \cite{ADMFKRag1+} for the analysis by nested ABA of higher rank open spin chains with diagonal boundary conditions.} \cite{ADMFKGau71}-\cite{ADMFKGZ94} up to the correlation functions\footnote{%
Let us recall that in the half-infinite volume case (with one boundary)
similar multiple integrals formulae have been previously derived in \cite{ADMFKJimKKKM95}-\cite{ADMFKJimKKMW95} by using the q-vertex operator method.} \cite{ADMFKKKMNST07}-\cite{ADMFKKKMNST08}, where the Lyon group method\footnote{Always in the ABA framework, see also \cite{ADMFKK01}-\cite{ADMFKCM07} for the
extension of this method to the higher spin quantum chains.} \cite{ADMFKKitMT99}-\cite{ADMFKKKMST07} has been generalized to the reflection algebra case in the ABA framework.

The situation is more complicated in the case of general non-diagonal
boundary matrices; technical difficulties arise and the ABA method can be
applied only limitedly to non-diagonal boundary matrices which satisfy
special constrains which allow the definition of reference states by gauge
transformations\footnote{See the series of papers \cite{ADMFKRag2+1}-\cite{ADMFKRag2+3} for the analysis of XXZ spin 1/2 open chains with special constraints on the non-diagonal boundaries by a method combing coordinate Bethe ansatz and matrix ansatz methods.} \cite{ADMFKCao03,ADMFKZhang07}. Even in this class of constrained
non-diagonal boundaries the use of algebraic Bethe ansatz is complicated
from the fact that two sets of Bethe ansatz equations are generally required
in order to have some numerical evidence of the completeness of the spectrum
description \cite{ADMFKNR03}. The other problem in the ABA framework is for the
computation of scalar products which is missing as soon as non-diagonal
boundary matrices are considered. Some analysis for solving this last
problem has been addressed in the papers \cite{ADMFKFK10,ADMFKFK10+} where the
partition function for the dynamical diagonal open case was considered, see
also \cite{ADMFKF11}. According to the standard techniques in the ABA framework, this is the first step towards the determination of scalar product formula in a determinant form. The situation for the general
unconstrained non-diagonal boundary matrices is even more fragmented in the
framework of Bethe ansatz analysis\footnote{%
It is worth mentioning that by using a different approach, the
representation theory of the so-called q-Onsager algebra, the spectrum of the spin-1/2 open XXZ quantum chains with the
general non-diagonal boundary conditions has been characterized in \cite{ADMFKBK05}-\cite{ADMFKB07} in terms of the roots of certain characteristic polynomials.};
only the eigenvalue analysis is implemented by the fusion procedure for
special values (roots of unit) of the anisotropy parameter \cite{ADMFKN05,ADMFKMNS06}
while for the other cases (non-roots of unit) the construction of the
Q-operator is mainly at a conjecture level \cite{ADMFKYNZ06}. Once again the
completeness of the spectrum description is verified only by some numerical
analysis and the absence of eigenstates construction is the first
fundamental missing step toward the matrix elements of local operators.

The circumstances that for general non-diagonal boundary matrices the ABA
method does not work while for the constrained ones, for which it works, it is instead missing a scalar
product formula have so far prevented to further
generalize the method described in \cite{ADMFKKKMNST07}-\cite{ADMFKKKMNST08} to
the reflection algebra representations for non-diagonal boundary matrices.
However, we are in the position to develop a different approach based on the
Sklyanin's quantum separation of variables (SOV) method\footnote{Let us comment that in the special case of the spin-1/2 representations of the rational 6-vertex reflection algebra the
construction of the functional version of separation of variables of Sklyanin has been implemented in \cite{ADMFKFSW08,ADMFKFGSW11}; that is a representation of the rational reflection algebra on a space of symmetric functions. However, the explicit construction of the SOV representation and of the transfer matrix eigenstates in the original Hilbert space of the quantum chain are not provided. See also \cite{ADMFKFram+2} for the analysis by the functional SOV of the related but more general spin-boson model introduced and analyzed by ABA in \cite{ADMFKFram+1}.} \cite{ADMFKSk1}-\cite{ADMFKSk3}, which can be used to get the exact characterization of their spectrum (eigenvalues \& eigenstates) and the
computation of their form factors. Indeed, Sklyanin's SOV is a more
efficient method to analyze the
spectral problem w.r.t. other methods\footnote{%
Other important examples are the coordinate Bethe ansatz \cite{ADMFKBe31}, \cite{ADMFKBaxBook} and \cite{ADMFKABBBQ87}, the Baxter Q-operator method \cite{ADMFKBaxBook}
and the analytic Bethe ansatz \cite{ADMFKRe83-1}-\cite{ADMFKRe83-2}.} like the algebraic Bethe ansatz. It works for a large class
of integrable quantum models; it leads to both the eigenvalues and the
eigenstates of the transfer matrix providing a complete characterization of the spectrum under simple requirements\footnote{Note that on the contrary in the ABA framework a proof of completeness has been given only for some models; for example see \cite{ADMFKMTV09} for the XXX Heisenberg model, \cite{ADMFKORR09} for the infinite quantum XXZ spin chain with domain wall boundary conditions, \cite{ADMFKKorff-11} for the nonlinear quantum Schroedinger model and the references contained in these papers.}. Moreover, in all the integrable quantum models analyzed  by SOV in the series of papers \cite{ADMFKNT-10}-\cite{ADMFKN12-3} it was possible to show that the transfer matrix forms a complete\footnote{Note that this is the natural quantum analogue of the definition of classical complete integrability.} set of commuting conserved charges of the quantum model.

The approach used in the present article can be considered as the
generalization to the SOV framework of the Lyon group method. It has been already implemented in \cite{ADMFKN12-0} for the XXZ spin 1/2 quantum chain \cite{ADMFKBe31}, \cite{ADMFKH28}-\cite{ADMFKLM66} with antiperiodic
boundary conditions. The results\footnote{Note that previous results on this model were the Baxter Q-operator \cite{ADMFKBBOY95} and the Sklyanin's functional separation of variables for the XXX chain \cite{ADMFKSk2} extended in \cite{ADMFKNWF09} to the XXZ case.} obtained in \cite{ADMFKN12-0} go from the
complete characterization of the spectrum up to the calculation of the form
factors of the local spin $\sigma _{n}^{a}$ in a determinant
form\footnote{The same type of results are derived in \cite{ADMFKN12-1} for rational 6-vertex Yang-Baxter algebra representations on antiperiodic spin-$s$ quantum chains.}. This approach has been originally
introduced in \cite{ADMFKGMN12-SG} for the lattice quantum sine-Gordon model \cite{ADMFKFST80,ADMFKIK82} and then generalized in \cite{ADMFKGMN12-T2} to the $\tau _{2}$-model\footnote{Note that in \cite{ADMFKGIPS06}-\cite{ADMFKGIPS09} a previous analysis by SOV method of the $\tau _{2}$-model has been implemented.} \cite{ADMFKBa04} and the chiral Potts model \cite{ADMFKBS90}-\cite{ADMFKTarasovSChP}.

Finally, let us comment that the analysis of these systems with non-diagonal
boundary conditions is in particular interesting for the relevant physical
applications to systems in non-equilibrium like the
asymmetric simple exclusion processes (ASEP) as they allow to describe
systems for which the particle number is not conserved\footnote{%
For relevant references on this subject from the point of view of the
connection to quantum integrable models see \cite{ADMFKD98}-\cite{ADMFKGE06} and
reference there in.}. Then, the lack of knowledge on the spectrum in this
general framework, the complications emerging in the use of algebraic Bethe
ansatz plus the physical relevance of these open chains with non-diagonal
boundary conditions makes clear how big can be the impact brought to this
research area from the solution of these systems by our approach in quantum
separation of variables.

\section{Reflection algebra and open spin-1/2 XXZ quantum chain}

In this section we describe a class of quantum integrable models
characterized in the framework of the quantum inverse scattering method by
monodromy matrices $\mathcal{U}(\lambda )$ which are solutions of the
following reflection equation:%
\begin{equation}
R_{12}(\lambda -\mu )\,\mathcal{U}_{1}(\lambda )\,R_{12}(\lambda +\mu )\,%
\mathcal{U}_{2}(\mu )=\mathcal{U}_{2}(\mu )\,R_{12}(\lambda +\mu )\,\mathcal{%
U}_{1}(\lambda )\,R_{12}(\lambda -\mu ),
\end{equation}%
w.r.t. the 6-vertex trigonometric solution of the Yang-Baxter equation, the $%
R$-matrix:%
\begin{equation}
R_{12}(\lambda )\equiv \left( 
\begin{array}{cccc}
\sinh (\lambda +\eta ) & 0 & 0 & 0 \\ 
0 & \sinh \lambda & \sinh \eta & 0 \\ 
0 & \sinh \eta & \sinh \lambda & 0 \\ 
0 & 0 & 0 & \sinh (\lambda +\eta )%
\end{array}%
\right) \in \text{End}(\text{R}_{1}\otimes \text{R}_{2}),
\end{equation}%
where R$_{x}\simeq \mathbb{C}^{2}$ is a 2-dimensional linear space.

\subsection{H.w. representations of 6-vertex reflection algebra on spin-1/2
chains}

Let $K(\lambda ;\zeta ,\delta ,\tau )$ be the following (general
non-diagonal boundary) matrix: 
\begin{equation}
K(\lambda ;\zeta ,\delta ,\tau )=\frac{1}{\sinh \zeta }\left( 
\begin{array}{cc}
\sinh (\lambda +\zeta ) & \kappa e^{\tau }\sinh 2\lambda \\ 
\kappa e^{-\tau }\sinh 2\lambda & \sinh (\zeta -\lambda )%
\end{array}%
\right) ,  \label{ADMFKK}
\end{equation}%
where $\zeta ,$ $\kappa $ and $\tau $ are arbitrary complex parameters,
it is the most general scalar solution of the 6-vertex trigonometric
reflection equation: 
\begin{equation}
R_{12}(\lambda -\mu )K_{1}(\lambda )R_{12}(\lambda +\mu )K_{2}(\mu
)=K_{2}(\mu )R_{12}(\lambda +\mu )K_{1}(\lambda )R_{12}(\lambda -\mu ).
\label{ADMFKbYB}
\end{equation}%
Then following Sklyanin \cite{ADMFKSkly88}, it is possible to construct in the 2$^{\mathsf{N}}$%
-dimensional representation space\footnote{%
The representation space of a spin 1/2 quantum chain of $\mathsf{N}$
local sites each one associated to a 2-dimensional local space R$_{n}$.}:%
\begin{equation}
\mathcal{R}_{\mathsf{N}}\equiv \otimes _{n=1}^{\mathsf{N}}\text{R}_{n},
\end{equation}%
two classes of solutions to the same reflection equation $\left( \ref{ADMFKbYB}%
\right) $. In order to do so, let us define:%
\begin{equation}
K_{-}(\lambda )=K(\lambda -\eta /2;\zeta _{-},\delta _{-},\tau _{-}),\text{
\ \ \ \ }K_{+}(\lambda )=K(\lambda +\eta /2;\zeta _{+},\delta _{+},\tau
_{+}),
\end{equation}%
where $\zeta _{\pm },\delta _{\pm },\tau _{\pm }$ are arbitrary complex
parameters and the (bulk) monodromy matrix\footnote{%
Note that here we have chosen a shifted definition of the inhomogeneity
w.r.t. the one used in the articles \cite{ADMFKKKMNST07,ADMFKKKMNST08}; in this case the homogeneous
limit corresponds to $\xi _{m}=0$ for $m=1,\ldots ,\mathsf{N}$.}:%
\begin{equation}
M_{0}(\lambda )=R_{0\mathsf{N}}(\lambda -\xi _{\mathsf{N}}-\eta /2)\ldots
R_{02}(\lambda -\xi _{2}-\eta /2)\,R_{01}(\lambda -\xi _{1}-\eta /2),\text{
\ \ }\hat{M}(\lambda )=(-1)^{\mathsf{N}}\,\sigma
_{0}^{y}\,M^{t_{0}}(-\lambda )\,\sigma _{0}^{y},  \label{ADMFKT}
\end{equation}%
$M_{0}(\lambda )\in $ End$($R$_{0}\otimes \mathcal{R}_{\mathsf{N}})$,
solution of the 6-vertex Yang-Baxter equation:%
\begin{equation}
R_{12}(\lambda -\mu )M_{1}(\lambda )M_{2}(\mu )=M_{2}(\mu )M_{1}(\lambda
)R_{12}(\lambda -\mu ).  \label{ADMFKYB}
\end{equation}%
Now we can define the following (boundary) monodromy matrices $\mathcal{U}%
_{\pm }(\lambda )\in $ End$($R$_{0}\otimes \mathcal{R}_{\mathsf{N}})$ as it
follows:%
\begin{eqnarray}
\mathcal{U}_{-}(\lambda ) &=&M_{0}(\lambda )K_{-}(\lambda )\hat{M}%
_{0}(\lambda )=\left( 
\begin{array}{cc}
\mathcal{A}_{-}(\lambda ) & \mathcal{B}_{-}(\lambda ) \\ 
\mathcal{C}_{-}(\lambda ) & \mathcal{D}_{-}(\lambda )%
\end{array}%
\right) , \\
\mathcal{U}_{+}^{t_{0}}(\lambda ) &=&M_{0}^{t_{0}}(\lambda
)K_{+}^{t_{0}}(\lambda )\hat{M}_{0}^{t_{0}}(\lambda )=\left( 
\begin{array}{cc}
\mathcal{A}_{+}(\lambda ) & \mathcal{C}_{+}(\lambda ) \\ 
\mathcal{B}_{+}(\lambda ) & \mathcal{D}_{+}(\lambda )%
\end{array}%
\right) ,
\end{eqnarray}%
then $\mathcal{U}_{-}(\lambda )$ and $\mathcal{V}_{+}(\lambda )\equiv 
\mathcal{U}_{+}^{t_{0}}(-\lambda )$ define two classes of solutions of the
reflection equation $\left( \ref{ADMFKbYB}\right) $.

As standard in the quantum inverse method and as it was proven in \cite{ADMFKSkly88}, from these monodromy matrices it is possible to define a commuting
family of transfer matrices $\mathcal{T}(\lambda )\in\, $End$(\mathcal{R}_{%
\mathsf{N}})$ as it follows:%
\begin{equation}
\mathcal{T}(\lambda )\equiv \text{tr}_{0}\{K_{+}(\lambda )\,M(\lambda
)\,K_{-}(\lambda )\hat{M}(\lambda )\}=\text{tr}_{0}\{K_{+}(\lambda )\mathcal{%
U}_{-}(\lambda )\}=\text{tr}_{0}\{K_{-}(\lambda )\mathcal{U}_{+}(\lambda )\}.
\label{ADMFKtransfer}
\end{equation}%
The problems that we address in this paper are the complete characterization
of the spectrum (eigenvalue $\&$ eigenstates) of this transfer matrix and the
computation of some matrix elements of quasi-local operators for two quite
general classes of non-diagonal boundary matrices $K_{\pm }(\lambda )$. It
is then worth recalling that the open spin-1/2 XXZ quantum chain, with the most
general non-diagonal integrable boundary conditions, is characterized by the
following Hamiltonian: 
\begin{align}
H_{\mathsf{N}D}& =\sum_{i=1}^{\mathsf{N}-1}(\sigma _{i}^{x}\sigma
_{i+1}^{x}+\sigma _{i}^{y}\sigma _{i+1}^{y}+\cosh \eta \sigma _{i}^{z}\sigma
_{i+1}^{z})  \notag \\
& +\frac{\sinh \eta }{\sinh \zeta _{-}}\left[ \sigma _{1}^{z}\cosh \zeta
_{-}+2\kappa _{-}(\sigma _{1}^{x}\cosh \tau _{-}+i\sigma _{1}^{y}\sinh \tau
_{-})\right]  \notag \\
& +\frac{\sinh \eta }{\sinh \zeta _{+}}[(\sigma _{\mathsf{N}}^{z}\cosh \zeta
_{+}+2\kappa _{+}(\sigma _{\mathsf{N}}^{x}\cosh \tau _{+}+i\sigma _{\mathsf{N%
}}^{y}\sinh \tau _{+}).  \label{ADMFKH-XXZ-Non-D}
\end{align}%
Hamiltonian which is reproduced in the homogeneous limit by the following derivative of the
transfer matrix $\left( \ref{ADMFKtransfer}\right) $:%
\begin{equation}
H_{\mathsf{N}D}=\frac{2(\sinh \eta )^{1-2\mathsf{N}}}{\text{tr}\{K_{\mathsf{N}%
D}^{+}(\eta /2)\}\,\text{tr}\{K_{\mathsf{N}D}^{-}(\eta /2)\}}\frac{d}{%
d\lambda }\mathcal{T}(\lambda )_{\,\vrule height13ptdepth1pt\>{\lambda =\eta
/2}\!}+\text{constant.}  \label{ADMFKHt}
\end{equation}

\subsection{First fundamental properties}

Here, some important properties about the generators of the reflection
algebra $\mathcal{A}_{\pm }(\lambda ),$ $\mathcal{B}_{\pm }(\lambda ),$ $%
\mathcal{C}_{\pm }(\lambda )$ and $\mathcal{D}_{\pm }(\lambda )$ are given
as they will play a fundamental role in the solution of the transfer matrix $%
\mathcal{T}(\lambda )$ spectral problem.

\begin{proposition}[$\mathcal{U}_{-}$-reflection algebra]
In the reflection algebra generated by the elements of $\,\mathcal{U}%
_{-}(\lambda )$ the quantum determinant:%
\begin{eqnarray}
\det_{q}\mathcal{U}_{-}(\lambda ) &\equiv &\sinh (2\lambda -2\eta )[\mathcal{%
A}_{-}(\lambda +\eta /2)\mathcal{A}_{+}(-\lambda +\eta /2)-\mathcal{B}%
_{-}(\lambda +\eta /2)\mathcal{C}_{-}(-\lambda +\eta /2)]  \label{ADMFKq-detU_1}
\\
&=&\sinh (2\lambda -2\eta )[\mathcal{D}_{-}(\lambda +\eta /2)\mathcal{D}%
_{-}(-\lambda +\eta /2)-\mathcal{C}_{-}(\lambda +\eta /2)\mathcal{B}%
_{-}(-\lambda +\eta /2)],  \label{ADMFKq-detU_2}
\end{eqnarray}%
is central:%
\begin{equation}
\lbrack \det_{q}\mathcal{U}_{-}(\lambda ),\mathcal{U}_{-}(\mu )]=0.
\end{equation}%
Moreover, it admits the following explicit expression:%
\begin{equation}
\det_{q}\mathcal{U}_{-}(\lambda )=\sinh (2\lambda -2\eta )\mathsf{A}%
_{-}(\lambda +\eta /2)\mathsf{A}_{-}(-\lambda +\eta /2),  \label{ADMFKq-detU_-exp}
\end{equation}%
where:%
\begin{equation}
\mathsf{A}_{-}(\lambda )\equiv g_{-}(\lambda )a(\lambda )d(-\lambda ),\text{
\ }d(\lambda )\equiv a(\lambda -\eta ),\text{ \ \ }a(\lambda )\equiv
\prod_{n=1}^{\mathsf{N}}\sinh (\lambda -\xi _{n}+\eta /2),
\end{equation}%
and:%
\begin{equation}
g_{\pm }(\lambda )\equiv \frac{\sinh (\lambda +\alpha _{\pm }\pm \eta
/2)\cosh (\lambda +\beta _{\pm }\pm \eta /2)}{\sinh \alpha _{\pm }\cosh
\beta _{\pm }},  \label{ADMFKg_PM}
\end{equation}%
where $\alpha _{\pm }$\ and $\beta _{\pm }$ are defined in terms of the
boundary parameters by:%
\begin{equation}
\sinh \alpha _{\pm }\cosh \beta _{\pm }\equiv \frac{\sinh \zeta _{\pm }}{%
2\kappa _{\pm }},\text{ \ \ \ \ \ }\cosh \alpha _{\pm }\sinh \beta _{\pm
}\equiv \frac{\cosh \zeta _{\pm }}{2\kappa _{\pm }}.  \label{ADMFKalfa-beta}
\end{equation}%
Moreover, the generator families $\mathcal{A}_{-}(\lambda )$ and $\mathcal{D}%
_{-}(\lambda )$ are related by the following parity relation:%
\begin{equation}
\mathcal{D}_{-}(\lambda )=\frac{\sinh (2\lambda -\eta )}{\sinh 2\lambda }%
\mathcal{A}_{-}(-\lambda )+\frac{\sinh \eta }{\sinh 2\lambda }\mathcal{A}%
_{-}(\lambda ),  \label{ADMFKSym-A-D-}
\end{equation}%
while for the other two families the following parity relations hold:%
\begin{equation}
\mathcal{B}_{-}(-\lambda )=-\frac{\sinh (2\lambda +\eta )}{\sinh (2\lambda
-\eta )}\mathcal{B}_{-}(\lambda )\text{ },\text{ \ }\mathcal{C}_{-}(-\lambda
)=-\frac{\sinh (2\lambda +\eta )}{\sinh (2\lambda -\eta )}\mathcal{C}%
_{-}(\lambda ).  \label{ADMFKSym-B-C-}
\end{equation}
\end{proposition}

\begin{proof}
This proposition is just a rephrasing and a simple extension to the case of
general non-diagonal $K_{-}(\lambda )$ boundary matrix of the results stated
in Propositions 5, 6 and 7 of Sklyanin's article \cite{ADMFKSkly88}. The
quantum determinant as defined in formulae (38)$_{{\small\cite{ADMFKSkly88}}}$-(42)$_{{\small
\cite{ADMFKSkly88}}}$:%
\begin{eqnarray}
\det_{q}\mathcal{U}_{-}(\lambda ) &\equiv &\widetilde{\mathcal{U}}%
_{-}(\lambda -\eta /2)\mathcal{U}_{-}(\lambda +\eta /2)=\mathcal{U}%
_{-}(\lambda +\eta /2)\widetilde{\mathcal{U}}_{-}(\lambda -\eta /2) \\
&=&\mathcal{A}_{-}(\lambda +\eta /2)\widetilde{\mathcal{D}}_{-}(\lambda
-\eta /2)-\mathcal{B}_{-}(\lambda +\eta /2)\widetilde{\mathcal{D}}%
_{-}(\lambda -\eta /2)  \label{ADMFKq-detU_-inter}
\end{eqnarray}%
is central independently from $K_{-}(\lambda )$ being diagonal or
non-diagonal. Here, we have used the definition (43)$_{\small{\cite{ADMFKSkly88}}}$:%
\begin{align}
\widetilde{\mathcal{U}}_{-}(\lambda )& \equiv -\frac{tr_{12}R_{12}(-\eta
)\left( \mathcal{U}_{-}\right) _{2}(\lambda )R_{21}(2\lambda )}{\sinh \eta }%
=\left( 
\begin{array}{cc}
\widetilde{\mathcal{D}}_{-}(\lambda ) & -\widetilde{\mathcal{B}}_{-}(\lambda
) \\ 
-\widetilde{\mathcal{C}}_{-}(\lambda ) & \widetilde{\mathcal{A}}_{-}(\lambda
)%
\end{array}%
\right) \\
& \equiv \left( 
\begin{array}{cc}
\mathcal{D}_{-}(\lambda )\sinh 2\lambda -\mathcal{A}_{-}(\lambda )\sinh \eta
& -\sinh (2\lambda +\eta )\mathcal{B}_{-}(\lambda ) \\ 
-\sinh (2\lambda +\eta )\mathcal{B}_{-}(\lambda ) & \mathcal{A}_{-}(\lambda
)\sinh 2\lambda -\mathcal{D}_{-}(\lambda )\sinh \eta%
\end{array}%
\right).
\end{align}%
So, we can also write it as it follows:%
\begin{equation}
\det_{q}\mathcal{U}_{-}(\lambda )\equiv \det_{q}K_{-}(\lambda
)\det_{q}M_{0}(\lambda )\det_{q}M_{0}(-\lambda ),
\end{equation}%
where $\det_{q}M(\lambda )$ is the (bulk) quantum determinant of the
Yang-Baxter algebra:%
\begin{eqnarray}
\det_{q}M_{0}(\lambda ) &=&A(\lambda +\eta /2)D(\lambda -\eta /2)-B(\lambda
+\eta /2)C(\lambda -\eta /2)  \notag \\
&=&a(\lambda +\eta /2)d(\lambda -\eta /2).
\end{eqnarray}%
Here, we have denoted:%
\begin{equation}
M_{0}(\lambda )=\left( 
\begin{array}{cc}
A(\lambda ) & B(\lambda ) \\ 
C(\lambda ) & D(\lambda )%
\end{array}%
\right) ,
\end{equation}%
and%
\begin{eqnarray}
\det_{q}K_{-}(\lambda ) &=&\left( K_{-}\right) _{1,1}(\lambda +\eta /2)(%
\widetilde{K_{-}})_{2,2}(\lambda -\eta /2)-\left( K_{-}\right)
_{1,2}(\lambda +\eta /2)(\widetilde{K_{-}})_{2,1}(\lambda -\eta /2)  \notag
\\
&=&-\frac{\sinh (2\lambda -2\eta )}{\sinh ^{2}\zeta _{-}}(\sinh (\lambda
+\zeta _{-})\sinh (\lambda -\zeta _{-})+\kappa _{-}^{2}\sinh ^{2}2\lambda ),
\end{eqnarray}%
where:%
\begin{equation}
\widetilde{K_{-}}(\lambda )\equiv \sinh (2\lambda -\eta )\sigma
_{0}^{y}K_{-}^{t_{0}}(-\lambda )\sigma _{0}^{y}.
\end{equation}%
Then the explicit expression (\ref{ADMFKq-detU_-exp})\ follows observing that it
holds:%
\begin{equation}
\det_{q}K_{-}(\lambda )=\sinh (2\lambda -2\eta )g_{-}(\lambda +\eta
/2)g_{-}(-\lambda +\eta /2),
\end{equation}%
when we use the parameters $\alpha _{-}$ and $\beta _{-}$.

Finally, it is simple to remark that formula (46)$_{\small{\cite{ADMFKSkly88}}}$ is
equivalent to the symmetry properties (\ref{ADMFKSym-A-D-})\ and (\ref{ADMFKSym-B-C-}%
), which in turn imply the expressions for the quantum determinant (\ref{ADMFKq-detU_1})\ and (\ref{ADMFKq-detU_2}), when used to rewrite formula (\ref{ADMFKq-detU_-inter}).
\end{proof}

It is worth remarking that similar statements hold for the reflection
algebra generated by $\mathcal{U}_{+}(\lambda )$. In fact, they are simply
consequences of the previous proposition when it is taken into account that $%
\mathcal{U}_{+}^{t_0}(-\lambda )$ satisfies the same reflection equation of $%
\mathcal{U}_{-}(\lambda )$.

\begin{proposition}[$\mathcal{U}_{+}$-reflection algebra]
In the reflection algebra generated by the elements of $\,\mathcal{U}%
_{+}(\lambda )$ the quantum determinant:%
\begin{align}
\det_{q}\mathcal{U}_{+}(\lambda )& =\sinh (2\lambda +2\eta )[\mathcal{A}%
_{+}(\lambda -\eta /2)\mathcal{A}_{+}(-\lambda -\eta /2)-\mathcal{B}%
_{+}(\lambda -\eta /2)\mathcal{C}_{+}(-\lambda -\eta /2)] \\
& =\sinh (2\lambda +2\eta )[\mathcal{D}_{+}(-\lambda -\eta /2)\mathcal{D}%
_{+}(\lambda -\eta /2)-\mathcal{C}_{+}(-\lambda -\eta /2)\mathcal{B}%
_{+}(\lambda -\eta /2)],
\end{align}%
is central:%
\begin{equation}
\lbrack \det_{q}\mathcal{U}_{+}(\lambda ),\mathcal{U}_{+}(\mu )]=0.
\end{equation}%
Moreover, it admits the following explicit expression:%
\begin{equation}
\det_{q}\mathcal{U}_{+}(\lambda )=\sinh (2\lambda +2\eta )\mathsf{D}%
_{+}(\lambda -\eta /2)\mathsf{D}_{+}(-\lambda -\eta /2),
\end{equation}%
where the function $\mathsf{D}_{+}(\lambda )$ is defined by:%
\begin{equation}
\mathsf{D}_{+}(\lambda )=g_{+}(\lambda )a(-\lambda )d(\lambda ),
\end{equation}%
where $g_{+}(\lambda )$ is defined in $(\ref{ADMFKg_PM})$. Moreover, the
generator families $\mathcal{A}_{+}(\lambda )$ and $\mathcal{D}_{+}(\lambda )$
are related by the following parity relation:%
\begin{equation}
\mathcal{D}_{+}(\lambda )=\frac{\sinh (2\lambda +\eta )}{\sinh 2\lambda }%
\mathcal{A}_{+}(-\lambda )-\frac{\sinh \eta }{\sinh 2\lambda }\mathcal{A}%
_{+}(\lambda ),  \label{ADMFKSym-A-D+}
\end{equation}%
while for the other two families the following parity relations hold:%
\begin{equation}
\mathcal{B}_{+}(-\lambda )=-\frac{\sinh (2\lambda -\eta )}{\sinh (2\lambda
+\eta )}\mathcal{B}_{+}(\lambda )\text{ },\text{ \ }\mathcal{C}_{+}(-\lambda
)=-\frac{\sinh (2\lambda -\eta )}{\sinh (2\lambda +\eta )}\mathcal{C}%
_{+}(\lambda ).  \label{ADMFKSym-B-C+}
\end{equation}
\end{proposition}

\begin{proof}
Here, it is only need to notice that:%
\begin{eqnarray}
\det_{q}K_{+}(\lambda ) &=&\left( K_{+}\right) _{1,1}(\lambda -\eta /2)(%
\widetilde{K_{+}})_{2,2}(\lambda +\eta /2)-\left( K_{+}\right)
_{1,2}(\lambda -\eta /2)(\widetilde{K_{+}})_{2,1}(\lambda +\eta /2)  \notag
\\
&=&-\frac{\sinh (2\lambda +2\eta )}{\sinh ^{2}\zeta _{+}}(\sinh (\lambda
+\zeta _{+})\sinh (\lambda -\zeta _{+})+\kappa _{+}^{2}\sinh ^{2}2\lambda ),
\end{eqnarray}%
where:%
\begin{equation}
\widetilde{K_{+}}(\lambda )\equiv \sinh (2\lambda +\eta )\sigma
_{0}^{y}K_{+}^{t_{0}}(-\lambda )\sigma _{0}^{y},
\end{equation}%
can be written in the form:%
\begin{equation}
\det_{q}K_{+}(\lambda )=\sinh (2\lambda +2\eta )g_{+}(\lambda -\eta
/2)g_{+}(-\lambda -\eta /2)),
\end{equation}%
where the parameters $\alpha _{+}$\ and $\beta _{+}$, entering in the
function $g_{+}(\lambda )$, are defined in terms of the parameters of the $K_{+}(\lambda )$ boundary matrix in $(\ref{ADMFKalfa-beta})$.
\end{proof}

Let us introduce the following notations:%
\begin{equation}
K_{\pm }(\lambda )\equiv \frac{1}{\sinh \zeta _{\pm }}\left( 
\begin{array}{cc}
\sinh (\lambda +\zeta _{\pm }\pm \eta /2) & \kappa _{\pm }e^{\tau _{\pm
}}\sinh (2\lambda \pm \eta ) \\ 
\kappa _{\pm }e^{-\tau _{\pm }}\sinh (2\lambda \pm \eta ) & \sinh (\zeta
_{\pm }\mp \eta /2-\lambda )%
\end{array}%
\right) =\left( 
\begin{array}{cc}
a_{\pm }\left( \lambda \right) & b_{\pm }\left( \lambda \right) \\ 
c_{\pm }\left( \lambda \right) & d_{\pm }\left( \lambda \right)%
\end{array}%
\right) ,
\end{equation}%
and by using these notations let us rewrite the transfer matrix $\left( \ref{ADMFKtransfer}\right) $ in the following two equivalent forms:%
\begin{equation}
\mathcal{T}(\lambda )=\mathcal{T}_{\setminus }^{(\pm )}(\lambda )+b_{\mp
}\left( \lambda \right) \mathcal{C}_{\pm }(\lambda )+c_{\mp }\left( \lambda
\right) \mathcal{B}_{\pm }(\lambda ),
\end{equation}%
where: 
\begin{equation}
\mathcal{T}_{\setminus }^{(\pm )}(\lambda )\equiv a_{\mp }\left( \lambda
\right) \mathcal{A}_{\pm }(\lambda )+d_{\mp }\left( \lambda \right) \mathcal{%
D}_{\pm }(\lambda ),  \label{ADMFKTr-pm-Diag}
\end{equation}%
are the transfer matrix\footnote{%
Note that $\mathcal{T}_{\setminus }^{(+)}(\lambda )$ corresponds to $%
K_{-}(\lambda )$ diagonal while $K_{+}(\lambda )$ is left general as well as 
$\mathcal{T}_{\setminus }^{(-)}(\lambda )$ corresponds to $K_{+}(\lambda )$
diagonal while $K_{-}(\lambda )$ is left general.} of the system with
diagonal matrix $K_{\mp }(\lambda )$, respectively, then:

\begin{corollary}
$\mathcal{T}_{\setminus }^{(\pm )}(\lambda )$ admits the following explicitly
even forms w.r.t. the spectral parameter $\lambda $:
\begin{eqnarray}
\mathcal{T}_{\setminus }^{(\pm )}(\lambda ) &\equiv &\mathsf{a}_{\mp
}(\lambda )\mathcal{A}_{\pm }(\lambda )+\mathsf{a}_{\mp }(-\lambda )\mathcal{%
A}_{\pm }(-\lambda )  \label{ADMFKT-diag-pm-A} \\
&=&\mathsf{d}_{\mp }(\lambda )\mathcal{D}_{\pm }(\lambda )+\mathsf{d}_{\mp
}(-\lambda )\mathcal{D}_{\pm }(-\lambda ),  \label{ADMFKT-diag-pm-D}
\end{eqnarray}%
where:%
\begin{eqnarray}
\mathsf{a}_{\pm }(\lambda ) &\equiv &\frac{\sinh (2\lambda \pm \eta )\sinh
(\lambda +\zeta _{\pm }\mp \eta /2)}{\sinh 2\lambda \sinh \zeta _{\pm }}, \\
\mathsf{d}_{\pm }(\lambda ) &\equiv &\frac{\sinh (2\lambda \pm \eta )\sinh
(\zeta _{\pm }-\lambda \pm \eta /2)}{\sinh 2\lambda \sinh \zeta _{\pm }}.
\end{eqnarray}%
Moreover, also the most general transfer matrix is even in the spectral
parameter $\lambda $: 
\begin{equation}
\mathcal{T}(-\lambda )=\mathcal{T}(\lambda ).  \label{ADMFKeven-transfer}
\end{equation}
\end{corollary}

\begin{proof}
By using the formulae $\left( \ref{ADMFKSym-A-D-}\right) $ and $\left( \ref{ADMFKSym-A-D+}\right) $ to rewrite $\mathcal{T}_{\setminus }^{(\pm )}(\lambda )$
only in terms of $\mathcal{A}_{\pm }(\lambda )$\ or only in terms of $%
\mathcal{D}_{\pm }(\lambda )$ after some simple algebra we get our formulae $%
\left( \ref{ADMFKT-diag-pm-A}\right) $ and $\left( \ref{ADMFKT-diag-pm-D}\right) $,
respectively. Then the parity $\left( \ref{ADMFKeven-transfer}\right) $ of the
transfer matrix $\mathcal{T}(\lambda )$ follows remarking that the parity
properties:%
\begin{equation}
b_{\mp }\left( -\lambda \right) \mathcal{C}_{\pm }(-\lambda )=b_{\mp }\left(
\lambda \right) \mathcal{C}_{\pm }(\lambda ),\text{ \ \ }c_{\mp }\left(
-\lambda \right) \mathcal{B}_{\pm }(-\lambda )=c_{\mp }\left( \lambda
\right) \mathcal{B}_{\pm }(\lambda ),
\end{equation}%
are just a rewriting of the known properties $\left( \ref{ADMFKSym-B-C-}\right) $
and $\left( \ref{ADMFKSym-B-C+}\right) $.
\end{proof}

\begin{proposition}\label{ADMFKHerm-conj}
The monodromy matrix \,$\mathcal{U}_{\pm }(\lambda )$ satisfy the following
transformation properties under Hermitian conjugation:\newline
\textsf{I)} Under the condition $\eta \in i\mathbb{R}$ (massless regime), it holds:
\begin{equation}
\mathcal{U}_{\pm }(\lambda )^{\dagger }=\left[ \mathcal{U}_{\pm }(-\lambda
^{\ast })\right] ^{t_{0}}  \label{ADMFKml-Hermitian_U},
\end{equation}%
\ \ \ \ \ for $\{i\tau _{\pm },i\kappa _{\pm },i\zeta _{\pm },\xi
_{1},...,\xi _{\mathsf{N}}\}\in \mathbb{R}^{\mathsf{N}+3}.$\newline
\textsf{II)} Under the condition $\eta \in\mathbb{R}$ (massive regime), it holds:
\begin{equation}
\mathcal{U}_{\pm }(\lambda )^{\dagger }=\left[ \mathcal{U}_{\pm }(\lambda
^{\ast })\right] ^{t_{0}}  \label{ADMFKm-Hermitian_U},
\end{equation}%
\ \ \ \ \ for $\{\tau _{\pm },\kappa _{\pm },\zeta _{\pm },i\xi
_{1},...,i\xi _{\mathsf{N}}\}\in \mathbb{R}^{\mathsf{N}+3}.$\newline
Under the same conditions on the parameters of the representation it holds:
\begin{equation}
\mathcal{T}(\lambda )^{\dagger }=\mathcal{T}(\lambda ^{\ast }),
\end{equation}%
i.e. $\mathcal{T}(\lambda )$ defines a one-parameter\ family of normal
operators which are self-adjoint both for $\lambda $ real and imaginary.
\end{proposition}

\begin{proof}
This proposition gives the generalization to the case of general
non-diagonal boundary conditions of the transformation properties under
Hermitian conjugation proven in \cite{ADMFKKKMNST07} for the diagonal case. The
proof is given by using the following transformation properties under
Hermitian conjugation:

\textsf{I)} For $\{i\eta ,i\tau _{\pm },i\kappa _{\pm },i\zeta _{\pm },\xi
_{1},...,\xi _{\mathsf{N}}\}\in \mathbb{R}^{\mathsf{N}+3}$ then it holds:%
\begin{eqnarray}
R_{0n}(\lambda -\xi _{1}-\eta /2)^{\dagger } &=&\sigma
_{0}^{y}R_{0n}(\lambda ^{\ast }-\xi _{1}-\eta /2)\sigma _{0}^{y}=-\left[
R_{0n}(-\lambda ^{\ast }+\xi _{1}-\eta /2)\right] ^{t_{0}}, \\
K_{\pm }(\lambda )^{\dagger } &=&\left[ K_{\pm }(-\lambda ^{\ast })\right]
^{t_{0}};  \label{ADMFKml-Hermitian_K}
\end{eqnarray}%
\textsf{II)} For $\{\eta ,\tau _{\pm },\kappa _{\pm },\zeta _{\pm },i\xi
_{1},...,i\xi _{\mathsf{N}}\}\in \mathbb{R}^{\mathsf{N}+3}$ then it holds:%
\begin{eqnarray}
R_{0n}(\lambda -\xi _{1}-\eta /2)^{\dagger } &=&-\sigma
_{0}^{y}R_{0n}(\lambda ^{\ast }-\xi _{1}-\eta /2)\sigma _{0}^{y}=\left[
R_{0n}(\lambda ^{\ast }+\xi _{1}-\eta /2)\right] ^{t_{0}} \\
K_{\pm }(\lambda )^{\dagger } &=&\left[ K_{\pm }(\lambda ^{\ast })\right]
^{t_{0}};  \label{ADMFKm-Hermitian_K}
\end{eqnarray}%
which can be verified by direct calculations. From these it follows:%
\begin{equation}
M(\lambda )^{\dagger }=\left( -\frac{\eta }{\eta ^{\ast }}\right) ^{\mathsf{N%
}}\sigma _{0}^{y}M(-\left( \frac{\eta }{\eta ^{\ast }}\right) \lambda ^{\ast
})\sigma _{0}^{y}=\left( \frac{\eta }{\eta ^{\ast }}\right) ^{\mathsf{N}}%
\left[ \hat{M}(\left( \frac{\eta }{\eta ^{\ast }}\right) \lambda ^{\ast })%
\right] ^{t_{0}},
\end{equation}%
for the (bulk) monodromy matrix and so:%
\begin{equation}
\mathcal{U}_{\pm }(\lambda )^{\dagger }=\left. \left\{ \left[ \hat{M}(\left( 
\frac{\eta }{\eta ^{\ast }}\right) \lambda ^{\ast })\right] ^{t_{0}}\left[
K_{\pm }(\left( \frac{\eta }{\eta ^{\ast }}\right) \lambda ^{\ast })\right]
^{t_{0}}\left[ M(\left( \frac{\eta }{\eta ^{\ast }}\right) \lambda ^{\ast })%
\right] ^{t_{0}}\right\} \right\vert _{\text{q-reverse-order}}
\label{ADMFKHermitian-U-0}
\end{equation}%
where the notation q-reverse-order is referred to the reverse order in the
generators of the Yang-Baxter algebras. More in details, the matrix elements
of the $2\times 2$\ matrix $\mathcal{U}_{\pm }(\lambda )^{\dagger }$ are
computed by using the normal matrix products in the auxiliary space $0$, as indicated
inside the brackets at the r.h.s. of (\ref{ADMFKHermitian-U-0}), then in each
element of the matrix $\mathcal{U}_{\pm }(\lambda )^{\dagger }$ the matrix
elements of $\left[ \hat{M}(\left( \frac{\eta }{\eta ^{\ast }}\right)
\lambda ^{\ast })\right] ^{t_{0}}$ are put to the right of those of $\left[
M(\left( \frac{\eta }{\eta ^{\ast }}\right) \lambda ^{\ast })\right]
^{t_{0}} $. It is then simple to verify that the r.h.s. of (\ref{ADMFKHermitian-U-0}) with this prescribed order coincides with:\ 
\begin{equation}
\left[ \mathcal{U}_{\pm }(\left( \frac{\eta }{\eta ^{\ast }}\right) \lambda
^{\ast })\right] ^{t_{0}},
\end{equation}%
which proves both (\ref{ADMFKml-Hermitian_U}) and (\ref{ADMFKm-Hermitian_U}). Then by
using these last two formulae and the transformation properties (\ref{ADMFKml-Hermitian_K}) and (\ref{ADMFKm-Hermitian_K}), we get:%
\begin{equation}
\mathcal{T}(\lambda )^{\dagger }=tr_{0}\left\{ \left[ K_{\mp }(\left( \frac{%
\eta }{\eta ^{\ast }}\right) \lambda ^{\ast })\right] ^{t_{0}}\left[ 
\mathcal{U}_{\pm }(\left( \frac{\eta }{\eta ^{\ast }}\right) \lambda ^{\ast
})\right] ^{t_{0}}\right\} =\mathcal{T}(\left( \frac{\eta }{\eta ^{\ast }}%
\right) \lambda ^{\ast })\underset{(\ref{ADMFKeven-transfer})}{=}\mathcal{T}%
(\lambda ^{\ast }).
\end{equation}
\end{proof}

\section{SOV representations for $\mathcal{T}(\protect\lambda )$-spectral
problem}

\label{ADMFKSOV-Gen}The method to construct quantum separation of variable (SOV)
representations for the spectral problem of the transfer matrices associated
to the representations of the Yang-Baxter algebra has been defined by
Sklyanin in \cite{ADMFKSk1,ADMFKSk2,ADMFKSk3}. Here, we
show that there are quite general representations of the reflection algebra
with non-diagonal boundary matrices for which the quantum SOV
representations can be constructed by adapting Sklyanin's method. More
in details the following theorem holds:

\begin{theorem}
\label{ADMFKTh1}Let the inhomogeneities $\{\xi
_{1},...,\xi _{\mathsf{N}}\}\in \mathbb{C}$ $^{\mathsf{N}}$ satisfy the following conditions:
\begin{equation}
\xi _{a}\neq \xi _{b}+r\eta\text{ \ }\forall a\neq b\in \{1,...,\mathsf{N}\}\,\,\text{and\,\,} r\in\{-1,0,1\},
\label{ADMFKE-SOV}
\end{equation}%
then the commuting families of generators of the reflection algebra $\mathcal{B}_{\epsilon}(\lambda )$\ and $\mathcal{C}_{\epsilon}(\lambda )$ are diagonalizable
and with simple spectrum, respectively, if $b_{\epsilon}\left( \lambda \right) \neq 0$ and $c_{\epsilon}\left( \lambda \right) \neq 0$, where  $\epsilon \in \{+,-\}$.

Moreover, the following statements hold: 
\newline
\vspace{-0.3cm}

\textsf{I)} The representations for which the commuting family $\mathcal{B}%
_{\epsilon }(\lambda )$\ is diagonal define the quantum SOV representations
for the spectral problem of the transfer matrix $\mathcal{T}_\epsilon(\lambda )$\
associated to the boundary matrix $K_{-\epsilon }\left( \lambda \right) $
diagonal or lower triangular while $K_{\epsilon }\left( \lambda \right) $ is
a general non-lower triangular; i.e. the SOV representations
for:%
\begin{equation}
\mathcal{T}_\epsilon(\lambda )\equiv \mathcal{T}_{\setminus }^{(\epsilon )}(\lambda
)+c_{-\epsilon }\left( \lambda \right) \mathcal{B}_{\epsilon }(\lambda ),%
\text{ \ for }b_{-\epsilon }\left( \lambda \right) =0\text{ and }b_{\epsilon
}\left( \lambda \right) \neq 0.
\end{equation}%
\textsf{II)} The representations for which the commuting family $\mathcal{C}%
_{\epsilon }(\lambda )$\ is diagonal define the quantum SOV representations
for the spectral problem of the transfer matrix $\bar{\mathcal{T}}_\epsilon(\lambda )$\
associated to the boundary matrix $K_{-\epsilon }\left( \lambda \right) $
diagonal or upper triangular while $K_{\epsilon }\left( \lambda \right) $ is
a general non-upper triangular one; i.e. the SOV representations for:%
\begin{equation}
\bar{\mathcal{T}}_\epsilon(\lambda )\equiv \mathcal{T}_{\setminus }^{(\epsilon )}(\lambda
)+b_{-\epsilon }\left( \lambda \right) \mathcal{C}_{\epsilon }(\lambda ),%
\text{ \ for }c_{-\epsilon }\left( \lambda \right) =0\text{ and }c_{\epsilon
}\left( \lambda \right) \neq 0.
\end{equation}
\end{theorem}

The proof of the theorem will be given in the following sections by the
explicit construction of $\mathcal{B}_{\pm }(\lambda )$-eigenbasis and the
solution by quantum separation of variables of the spectral problem for the
transfer matrix $\mathcal{T}(\lambda )$ in the class of representations of
the reflection algebra listed in point \textsf{I)} of the theorem. The
Hermitian conjugation properties lead to the construction of the $\mathcal{C}%
_{\pm }(\lambda )$-eigenbasis and then imply that the theorem holds also for
the class of representations listed in point \textsf{II)}\ of the theorem.

\subsection{Left and right representations of the reflection algebras}

Let us give some more details on the space of the representation of the spin
1/2 quantum chain. As this is the same representation space used in \cite{ADMFKN12-0} for the 6-vertex Yang Baxter algebra similar notation will be given here.

Let us introduce the standard spin basis for the $2$-dimensional linear
space R$_{n}$, the quantum space in the site $n$ of the chain, whose
elements are the $\sigma _{n}^{z}$-eigenvectors $|k,n\rangle $,
characterized by:%
\begin{equation}
\sigma _{n}^{z}|k,n\rangle =k|k,n\rangle ,\text{ \ }k\in \{-1,1\}.
\end{equation}%
Similarly, the $\sigma _{n}^{z}$-eigencovectors $\langle k,n|$,
characterized by:%
\begin{equation}
\sigma _{n}^{z}|k,n\rangle =k|k,n\rangle ,\text{ \ }k\in \{-1,1\},
\end{equation}%
define a basis in L$_{n}$ the dual space of R$_{n}$. Then, 2$^{\mathsf{N}}$%
-dimensional representations with $\mathsf{N}+6$ parameters (the
inhomogeneities and the boundary parameters) of the reflection algebra are
defined in the \textit{left} (covectors) and \textit{right} (vectors) linear
spaces:%
\begin{equation}
\mathcal{L}_{\mathsf{N}}\equiv \otimes _{n=1}^{\mathsf{N}}\text{L}_{n},\text{
\ \ \ \ }\mathcal{R}_{\mathsf{N}}\equiv \otimes _{n=1}^{\mathsf{N}}\text{R}%
_{n}.
\end{equation}%
Moreover, $\mathcal{R}_{\mathsf{N}}$ is naturally provided with the
structure of Hilbert space by introducing the scalar product characterized
by the following action on the spin basis:%
\begin{equation}
( \otimes _{n=1}^{\mathsf{N}}|k_{n},n\rangle , \otimes
_{n=1}^{\mathsf{N}}|k_{n}^{\prime },n\rangle  )\equiv \prod_{n=1}^{%
\mathsf{N}}\delta _{k_{n},k_{n}^{\prime }}\text{ \ \ }\forall
k_{n},k_{n}^{\prime }\in \{-1,1\}.
\end{equation}

\subsection{$\mathcal{B}_{-}$-SOV representations of reflection algebra}

In this subsection we construct the left and right SOV-representations of
reflection algebra generated by $\mathcal{U}_{-}(\lambda )$ by constructing
the left and right $\mathcal{B}_{-}(\lambda )$-eigenbasis.

\begin{theorem}
\textsf{I)} \underline{Left $\mathcal{B}_{-}(\lambda )$ SOV-representations}
\ If $\left( \ref{ADMFKE-SOV}\right) $ is satisfied and $b_{-}\left( \lambda
\right) \neq 0$, then the states:%
\begin{equation}
\langle -,h_{1},...,h_{\mathsf{N}}|\equiv \frac{1}{\text{\textsc{n}}_{-}}%
\langle 0|\prod_{n=1}^{\mathsf{N}}\left( \frac{\mathcal{A}_{-}(\eta /2-\xi
_{n})}{\mathsf{A}_{-}(\eta /2-\xi _{n})}\right) ^{h_{n}},
\label{ADMFKD-left-eigenstates}
\end{equation}%
where:%
\begin{equation}
\langle 0|\equiv \otimes _{n=1}^{\mathsf{N}}\langle 1,n|,\text{ \ \textsc{n}}%
_{-}=\left[ \frac{\prod_{1\leq b<a\leq \mathsf{N}}(\eta _{a}^{(1)}-\eta
_{a}^{(1)})}{\langle 0|\left( \prod_{n=1}^{\mathsf{N}}\mathcal{A}_{-}(\eta
/2-\xi _{n})/\mathsf{A}_{-}(\eta /2-\xi _{n})\right) \overline{|0\rangle }}%
\right] ^{1/2},  \label{ADMFKNorm-def}
\end{equation}%
and 
\begin{equation}
\eta _{a}^{(h_{a})}\equiv \cosh 2\left[ (\xi _{n}+(h_{n}-\frac{1}{2})\eta %
\right],
\end{equation}%
\ $h_{n}\in \{0,1\},$ $n\in \{1,...,\mathsf{N}\}$, define a $\mathcal{B}%
_{-}(\lambda )$-eigenbasis of $\mathcal{L}_{\mathsf{N}}$:%
\begin{equation}
\langle -,\text{\textbf{h}}|\mathcal{B}_{-}(\lambda )=\text{\textsc{b}}_{-,%
\text{\textbf{h}}}(\lambda )\langle -,\text{\textbf{h}}|,
\label{ADMFKright-B-eigen-cond}
\end{equation}%
where $\langle -$, \textbf{h}$|\equiv \langle -,h_{1},...,h_{\mathsf{N}}|$
for \textbf{h}$\equiv (h_{1},...,h_{\mathsf{N}})$ and%
\begin{equation}
\text{\textsc{b}}_{-,\text{\textbf{h}}}(\lambda )\equiv \kappa _{-}e^{\tau
_{-}}\frac{\sinh (2\lambda -\eta )}{\sinh \zeta _{-}}a_{\text{\textbf{h}}%
}(\lambda )a_{\text{\textbf{h}}}(-\lambda ),  \label{ADMFKEigenValue-D}
\end{equation}%
with%
\begin{equation}
a_{\text{\textbf{h}}}(\lambda )\equiv \prod_{n=1}^{\mathsf{N}}\sinh (\lambda
-\xi _{n}-(h_{n}-\frac{1}{2})\eta ).
\end{equation} 
On the generic state $\langle -$, \textbf{h}$|$, the action of the remaining reflection algebra generators follows by:
\begin{align}
\langle -,\text{\textbf{h}}|\mathcal{A}_{-}(\lambda )& =\sum_{a=1}^{2\mathsf{%
N}}\frac{\sinh (2\lambda -\eta )\sinh (\lambda +\zeta _{a}^{(h_{a})})}{\sinh
(2\zeta _{a}^{(h_{a})}-\eta )\sinh 2\zeta _{a}^{(h_{a})}}\prod_{\substack{ %
b=1  \\ b\neq a\text{ mod}\mathsf{N}}}^{\mathsf{N}}\frac{\cosh 2\lambda
-\cosh 2\zeta _{b}^{(h_{b})}}{\cosh 2\zeta _{a}^{(h_{a})}-\cosh 2\zeta
_{b}^{(h_{b})}}\mathsf{A}_{-}(\zeta _{a}^{(h_{a})})  \notag \\
& \times \langle -,\text{\textbf{h}}|\text{T}_{a}^{-\varphi
_{a}}+\det_{q}M(0)\cosh (\lambda -\eta /2)\prod_{b=1}^{\mathsf{N}}\frac{%
\cosh 2\lambda -\cosh 2\zeta _{b}^{(h_{b})}}{\cosh \eta -\cosh 2\zeta
_{b}^{(h_{b})}}\langle -,\text{\textbf{h}}|  \notag \\
& +(-1)^{\mathsf{N}}\coth \zeta _{-}\det_{q}M(i\pi /2)\sinh (\lambda -\eta
/2)\prod_{b=1}^{\mathsf{N}}\frac{\cosh 2\lambda -\cosh 2\zeta _{b}^{(h_{b})}%
}{\cosh \eta +\cosh 2\zeta _{b}^{(h_{b})}}\langle -,\text{\textbf{h}}|,
\label{ADMFKL-SOV A-}
\end{align}%
where:%
\begin{eqnarray}
\zeta _{n}^{(h_{n})} &=&\varphi _{n}\left[ \xi _{n}+(h_{n}-\frac{1}{2})\eta %
\right] \text{ \ \ for \ }h_{n}\in \{0,1\}\text{\ \ and\ \ }\forall n\in
\{1,...,2\mathsf{N}\}, \\
\varphi _{a} &=&1-2\theta (a-\mathsf{N})\text{ \ \ with \ }\theta (x)=\{0%
\text{ for }x\leq 0,\text{ }1\text{ for }x>0\},
\end{eqnarray}%
and:%
\begin{equation}
\langle -,h_{1},...,h_{a},...,h_{\mathsf{N}}|\text{T}_{a}^{\pm }=\langle
-,h_{1},...,h_{a}\pm 1,...,h_{\mathsf{N}}|.
\end{equation}%
Indeed, the representation of $\,\mathcal{D}_{-}(\lambda )$ follows from the
identity (\ref{ADMFKSym-A-D-}) while $\mathcal{C}_{-}(\lambda )$ is uniquely
defined by the quantum determinant relation.\smallskip

\textsf{II)} \underline{Right $\mathcal{B}_{-}(\lambda )$ SOV-representations%
} \ If $\left( \ref{ADMFKE-SOV}\right) $ is satisfied and $b_{-}\left( \lambda
\right) \neq 0$, the states:%
\begin{equation}
|-,h_{1},...,h_{\mathsf{N}}\rangle \equiv \frac{1}{\text{\textsc{n}}_{-}}%
\prod_{n=1}^{\mathsf{N}}\left( \frac{\mathcal{D}_{-}(\xi _{n}+\eta /2)}{%
k_{n}^{(-)}\mathsf{A}_{-}(\eta /2-\xi _{n})}\right) ^{(1-h_{n})}\overline{%
|0\rangle },  \label{ADMFKD-right-eigenstates}
\end{equation}%
where:%
\begin{equation}
\overline{|0\rangle }\equiv \otimes _{n=1}^{\mathsf{N}}|-1,n\rangle ,\text{
\ }k_{n}^{(-)}=\frac{\sinh (2\xi _{n}+\eta )}{\sinh (2\xi _{n}-\eta )},
\end{equation}%
\ $h_{n}\in \{0,1\},$ $n\in \{1,...,\mathsf{N}\}$, define a $\mathcal{B}%
_{-}(\lambda )$-eigenbasis of $\mathcal{R}_{\mathsf{N}}$:%
\begin{equation}
\mathcal{B}_{-}(\lambda )|-,\text{\textbf{h}}\rangle =|-,\text{\textbf{h}}%
\rangle \text{\textsc{b}}_{-,\text{\textbf{h}}}(\lambda ).
\label{ADMFKleft-B-eigen-cond}
\end{equation}%
On the generic state $|-$, \textbf{h}$\rangle $, the action of the remaining reflection algebra generators follows by:
\begin{align}
\mathcal{D}_{-}(\lambda )|-,\text{\textbf{h}}\rangle & =\sum_{a=1}^{2\mathsf{%
N}}\text{T}_{a}^{-\varphi _{a}}|-,\text{\textbf{h}}\rangle \frac{\sinh
(2\lambda -\eta )\sinh (\lambda +\zeta _{a}^{(h_{a})})}{\sinh (2\zeta
_{a}^{(h_{a})}-\eta )\sinh 2\zeta _{a}^{(h_{a})}}\prod_{\substack{ b=1  \\ %
b\neq a\text{ mod}\mathsf{N}}}^{\mathsf{N}}\frac{\cosh 2\lambda -\cosh
2\zeta _{b}^{(h_{b})}}{\cosh 2\zeta _{a}^{(h_{a})}-\cosh 2\zeta
_{b}^{(h_{b})}}\mathsf{D}_{-}(\zeta _{a}^{(h_{a})})  \notag \\
& +|-,\text{\textbf{h}}\rangle \det_{q}M(0)\cosh (\lambda -\eta
/2)\prod_{b=1}^{\mathsf{N}}\frac{\cosh 2\lambda -\cosh 2\zeta _{b}^{(h_{b})}%
}{\cosh \eta -\cosh 2\zeta _{b}^{(h_{b})}}  \notag \\
& +(-1)^{\mathsf{N}+1}|-,\text{\textbf{h}}\rangle \coth \zeta
_{-}\det_{q}M(i\pi /2)\sinh (\lambda -\eta /2)\prod_{b=1}^{\mathsf{N}}\frac{%
\cosh 2\lambda -\cosh 2\zeta _{b}^{(h_{b})}}{\cosh \eta +\cosh 2\zeta
_{b}^{(h_{b})}},  \label{ADMFKR-SOV D-}
\end{align}%
where:%
\begin{equation}
\mathsf{D}_{-}(\zeta _{a}^{(h_{a})})=(k_{a}^{(-)})^{\varphi _{a}}\mathsf{A}%
_{-}(\zeta _{a}^{(h_{a})}-2\varphi _{a}\xi _{a}),\text{ \ \ \ \ \ T}%
_{a}^{\pm }|-,h_{1},...,h_{a},...,h_{\mathsf{N}}\rangle
=|-,h_{1},...,h_{a}\pm 1,...,h_{\mathsf{N}}\rangle .
\end{equation}%
Indeed, the representation of $\mathcal{A}_{-}(\lambda )$ follows from the
identity (\ref{ADMFKSym-A-D-}) while $\mathcal{C}_{-}(\lambda )$ is uniquely
defined by the quantum determinant relation.\smallskip
\end{theorem}

\begin{proof}[Proof of \textsf{I)}]
It is worth writing explicitly the (boundary-bulk) decomposition of the
reflection algebra generator:%
\begin{equation}
\mathcal{B}_{-}(\lambda )=-a_{-}(\lambda )A(\lambda )B(-\lambda
)+b_{-}(\lambda )A(\lambda )A(-\lambda )-c_{-}(\lambda )B(\lambda
)B(-\lambda )+d_{-}(\lambda )B(\lambda )A(-\lambda ),
\end{equation}%
in terms of the generators of the Yang-Baxter algebra. Then, the following
well know properties:%
\begin{equation}
\langle 0|A(\lambda )=a(\lambda )\langle 0|,\text{ \ \ \ }\langle
0|D(\lambda )=d(\lambda )\langle 0|,\text{ \ \ \ }\langle 0|B(\lambda )=%
\text{\b{0}},\text{ \ \ \ }\langle 0|C(\lambda )\neq \text{\b{0}},
\label{ADMFKLeft-reference}
\end{equation}%
with:%
\begin{equation}
a(\lambda )\equiv \prod_{n=1}^{\mathsf{N}}\sinh (\lambda -\xi _{n}+\eta /2),%
\text{ \ \ \ }d(\lambda )\equiv \prod_{n=1}^{\mathsf{N}}\sinh (\lambda -\xi
_{n}-\eta /2),
\end{equation}%
imply that $\langle 0|$ is a $\mathcal{B}_{-}(\lambda )$-eigenstate with
non-zero eigenvalue:%
\begin{equation}
\langle 0|\mathcal{B}_{-}(\lambda )\equiv \text{\textsc{b}}_{-,\text{\textbf{%
0}}}(\lambda )\langle 0|
\end{equation}%
where:%
\begin{equation}
\text{\textsc{b}}_{-,\text{\textbf{0}}}(\lambda )\equiv \kappa _{-}e^{\tau
_{-}}\frac{\sinh (2\lambda -\eta )}{\sinh \zeta _{-}}a(\lambda )a(-\lambda ).
\end{equation}%
Now by using the reflection algebra commutation relations:%
\begin{eqnarray}
\mathcal{A}_{-}(\lambda _{2})\mathcal{B}_{-}(\lambda _{1}) &=&\frac{\sinh
(\lambda _{1}-\lambda _{2}+\eta )\sinh (\lambda _{2}+\lambda _{1}-\eta )}{%
\sinh (\lambda _{1}-\lambda _{2})\sinh (\lambda _{1}+\lambda _{2})}\mathcal{B%
}_{-}(\lambda _{1})\mathcal{A}_{-}(\lambda _{2})  \notag \\
&&+\frac{\sinh (2\lambda _{1}-\eta )\sinh \eta }{\sinh (\lambda _{2}-\lambda
_{1})\sinh 2\lambda _{1}}\mathcal{B}_{-}(\lambda _{2})\mathcal{A}%
_{-}(\lambda _{1})  \notag \\
&&-\frac{\sinh \eta }{\sinh (\lambda _{1}+\lambda _{2})\sinh 2\lambda _{1}}%
\mathcal{B}_{-}(\lambda _{2})\mathcal{\tilde{D}}_{-}(\lambda _{1})
\label{ADMFKbYB-AB}
\end{eqnarray}%
we can follow step by step the proof given in \cite{ADMFKN12-0} to prove the
validity of $(\ref{ADMFKright-B-eigen-cond})$ and $(\ref{ADMFKleft-B-eigen-cond})$.
Under the condition $\left( \ref{ADMFKE-SOV}\right) $, these relations also imply
that each set of states $\langle -$, \textbf{h}$|$ and $|-$, \textbf{h}$%
\rangle $ form a set of 2$^{\mathsf{N}}$ independent states, i.e. a $%
\mathcal{B}_{-}(\lambda )$-eigenbasis of $\mathcal{L}_{\mathsf{N}}$ and $%
\mathcal{R}_{\mathsf{N}}$, respectively.

The action of $\mathcal{A}_{-}(\zeta _{b}^{(h_{b})})$ for $b\in \{1,...,2%
\mathsf{N}\}$ follows by the definition of the states $\langle -$, \textbf{h}$%
|$, the reflection algebra commutation relations $\left( \ref{ADMFKbYB-AB}\right) 
$ and the quantum determinant relations. Moreover, by using the identities:%
\begin{equation}
\mathcal{U}_{-}(\eta /2)=\det_{q}M(0)\text{ }I_{0},\text{ \ \ }\mathcal{U}%
_{-}(\eta /2+i\pi /2)=i\coth \zeta _{-}\det_{q}M(i\pi /2)\text{ }\sigma
_{0}^{z},  \label{ADMFKU-identities}
\end{equation}%
and remarking that $\mathcal{A}_{-}(\lambda )$ has the following functional
dependence w.r.t. $\lambda $:%
\begin{equation}
\mathcal{A}_{-}(\lambda )=\sum_{a=0}^{2\mathsf{N}+1}e^{\left( 2a-2\mathsf{N}%
+1\right) \lambda }\mathcal{A}_{-,a}
\end{equation}%
we get the following interpolation formula for the action on $\langle -$, \textbf{h}$|$:%
\begin{align}
\langle -,\text{\textbf{h}}|\mathcal{A}_{-}(\lambda )& =\sum_{a=1}^{2\mathsf{%
N}}\prod_{\substack{ b=1  \\ b\neq a}}^{2\mathsf{N}+2}\frac{\sinh (\lambda
-\zeta _{b}^{(h_{b})})}{\sinh (\zeta _{a}^{(h_{a})}-\zeta _{b}^{(h_{b})})}%
\mathsf{A}_{-}(\zeta _{a}^{(h_{a})})\langle -,\text{\textbf{h}}|\text{T}%
_{a}^{-\varphi _{a}}  \notag \\
& +\det_{q}M(0)\prod_{\substack{ b=1  \\ b\neq 2\mathsf{N}+1}}^{2\mathsf{N}%
+2}\frac{\sinh (\lambda -\zeta _{b}^{(h_{b})})}{\sinh (\zeta _{2\mathsf{N}%
+1}^{(1)}-\zeta _{b}^{(h_{b})})}\langle -,\text{\textbf{h}}|  \notag \\
& +i\coth \zeta _{-}\det_{q}M(i\pi /2)\prod_{b=1}^{2\mathsf{N}+1}\frac{\sinh
(\lambda -\zeta _{b}^{(h_{b})})}{\sinh (\zeta _{2\mathsf{N}+2}^{(1)}-\zeta
_{b}^{(h_{b})})}\langle -,\text{\textbf{h}}|,
\end{align}%
where:%
\begin{equation}
\zeta _{2\mathsf{N}+1}^{(1)}=\eta /2,\text{ \ \ \ \ }\zeta _{2\mathsf{N}%
+2}^{(1)}=\eta /2+i\pi /2,\text{ \ \ }\det_{q}M(\lambda )=a(\lambda +\eta
/2)d(\lambda -\eta /2),
\end{equation}%
and we have denoted $\zeta _{2\mathsf{N}+b}^{(h_{2\mathsf{N}+b})}=\zeta _{2%
\mathsf{N}+b}^{(1)}$ for $b=1,2$ for any $\mathcal{B}_{-}$-eigenstate. Then,
it is a simple exercise to rewrite this in the form $\left( \ref{ADMFKL-SOV A-}%
\right) $.
\end{proof}

\begin{proof}[Proof of \textsf{II)}]
The proof is given along the same line delineated for the point \textsf{I)}
of the theorem, we just need to make the following remarks. First of all
being:%
\begin{equation}
A(\lambda )\overline{|0\rangle }=d(\lambda )\overline{|0\rangle },\text{ \ \
\ }D(\lambda )\overline{|0\rangle }=a(\lambda )\overline{|0\rangle },\text{
\ \ \ }B(\lambda )\overline{|0\rangle }=\text{\b{0}},\text{ \ \ \ }C(\lambda
)\overline{|0\rangle }\neq \text{\b{0}},  \label{ADMFKRight-reference}
\end{equation}%
then $\overline{|0\rangle }$ is a $\mathcal{B}_{-}(\lambda )$-eigenstate
with non-zero eigenvalue:%
\begin{equation}
\mathcal{B}_{-}(\lambda )\overline{|0\rangle }\equiv \text{\textsc{b}}_{-,%
\text{\textbf{1}}}(\lambda )\overline{|0\rangle }
\end{equation}%
where:%
\begin{equation}
\text{\textsc{b}}_{-,\text{\textbf{1}}}(\lambda )\equiv \kappa _{-}e^{\tau
_{-}}\frac{\sinh (2\lambda -\eta )}{\sinh \zeta _{-}}d(\lambda )d(-\lambda ).
\end{equation}%
Now all we need are the following reflection algebra commutation relations:%
\begin{align}
\mathcal{B}_{-}(\lambda _{1})\mathcal{D}_{-}(\lambda _{2})& =\frac{\sinh
(\lambda _{1}-\lambda _{2}+\eta )\sinh (\lambda _{2}+\lambda _{1}-\eta )}{%
\sinh (\lambda _{1}-\lambda _{2})\sinh (\lambda _{1}+\lambda _{2})}\mathcal{D%
}_{-}(\lambda _{2})\mathcal{B}_{-}(\lambda _{1})  \notag \\
& -\frac{\sinh \eta \sinh (\lambda _{2}+\lambda _{1}-\eta )}{\sinh (\lambda
_{1}-\lambda _{2})\sinh (\lambda _{2}+\lambda _{1})}\mathcal{D}_{-}(\lambda
_{1})\mathcal{B}_{-}(\lambda _{2})  \notag \\
& -\frac{\sinh \eta }{\sinh (\lambda _{1}+\lambda _{2})}\mathcal{A}%
_{-}(\lambda _{1})\mathcal{B}_{-}(\lambda _{2}).
\end{align}%
By using them we get the following interpolation formula for the action on $%
|-$, \textbf{h}$\rangle $:%
\begin{eqnarray}
\mathcal{D}_{-}(\lambda )|-,\text{\textbf{h}}\rangle &=&\sum_{a=1}^{2\mathsf{%
N}}\text{T}_{a}^{-\varphi _{a}}|-,\text{\textbf{h}}\rangle \prod _{\substack{
b=1  \\ b\neq a}}^{2\mathsf{N}+2}\frac{\sinh (\lambda -\zeta _{b}^{(h_{b})})%
}{\sinh (\zeta _{a}^{(h_{a})}-\zeta _{b}^{(h_{b})})}\mathsf{D}_{-}(\zeta
_{a}^{(h_{a})})  \notag \\
&&+|-,\text{\textbf{h}}\rangle \det_{q}M(0)\prod_{\substack{ b=1  \\ b\neq 2%
\mathsf{N}+1}}^{2\mathsf{N}+2}\frac{\sinh (\lambda -\zeta _{b}^{(h_{b})})}{%
\sinh (\zeta _{2\mathsf{N}+1}^{(1)}-\zeta _{b}^{(h_{b})})}  \notag \\
&&-i|-,\text{\textbf{h}}\rangle \coth \zeta _{-}\det_{q}M(i\pi
/2)\prod_{b=1}^{2\mathsf{N}+1}\frac{\sinh (\lambda -\zeta _{b}^{(h_{b})})}{%
\sinh (\zeta _{2\mathsf{N}+2}^{(1)}-\zeta _{b}^{(h_{b})})},
\end{eqnarray}%
which can be rewritten in the form $\left( \ref{ADMFKR-SOV D-}\right) $.
\end{proof}

\subsection{$\mathcal{B}_{+}$-SOV representations of reflection algebra}

In this subsection we construct the left and right SOV-representations of
the reflection algebra generated by $\mathcal{U}_{+}(\lambda )$ by constructing
the left and right $\mathcal{B}_{+}(\lambda )$-eigenbasis.

\begin{theorem}
\textsf{I)} \underline{Left $\mathcal{B}_{+}(\lambda )$ SOV-representations}
\ If $\left( \ref{ADMFKE-SOV}\right) $ is satisfied and $b_{+}\left( \lambda
\right) \neq 0$, then the states:%
\begin{equation}
\langle +,h_{1},...,h_{\mathsf{N}}|\equiv \frac{1}{\text{\textsc{n}}_{+}}%
\langle 0|\prod_{n=1}^{\mathsf{N}}\left( \frac{\mathcal{D}_{+}(-\zeta
_{n}^{(1)})}{\mathsf{D}_{+}(-\zeta _{n}^{(1)})}\right) ^{(1-h_{n})},
\label{ADMFKD-left-eigenstates+}
\end{equation}%
where%
\begin{equation}
\text{\textsc{n}}_{+}=\left[ \frac{\prod_{1\leq b<a\leq \mathsf{N}}(\eta
_{a}^{(0)}-\eta _{a}^{(0)})}{\langle 0|\left( \prod_{n=1}^{\mathsf{N}}%
\mathcal{D}_{+}(-\zeta _{n}^{(1)})/\mathsf{D}_{+}(-\zeta _{n}^{(1)})\right) 
\overline{|0\rangle }}\right] ^{1/2},
\end{equation}%
\ $h_{n}\in \{0,1\},$ $n\in \{1,...,\mathsf{N}\}$, define a $\mathcal{B}%
_{+}(\lambda )$-eigenbasis of $\mathcal{L}_{\mathsf{N}}$:%
\begin{equation}
\langle +,\text{\textbf{h}}|\mathcal{B}_{+}(\lambda )=\text{\textsc{b}}_{+,%
\text{\textbf{h}}}(\lambda )\langle +,\text{\textbf{h}}|,
\label{ADMFKD-L-EigenV+}
\end{equation}%
where $\langle +$, \textbf{h}$|\equiv \langle +,h_{1},...,h_{\mathsf{N}}|$
for \textbf{h}$\equiv (h_{1},...,h_{\mathsf{N}})$ and%
\begin{equation}
\text{\textsc{b}}_{+,\text{\textbf{h}}}(\lambda )\equiv \kappa _{+}e^{\tau
_{+}}\frac{\sinh (2\lambda +\eta )}{\sinh \zeta _{+}}a_{\text{\textbf{h}}%
}(\lambda )a_{\text{\textbf{h}}}(-\lambda ).  \label{ADMFKEigenValue-D+}
\end{equation}%
On the generic state $\langle +$, \textbf{h}$|$, the action of the remaining reflection algebra generators follows by:
\begin{align}
\langle +,\text{\textbf{h}}|\mathcal{D}_{+}(\lambda )& =\sum_{a=1}^{2\mathsf{%
N}}\frac{\sinh (2\lambda +\eta )\sinh (\lambda +\zeta _{a}^{(h_{a})})}{\sinh
(2\zeta _{a}^{(h_{a})}+\eta )\sinh 2\zeta _{a}^{(h_{a})}}\prod_{\substack{ %
b=1  \\ b\neq a\text{ mod}\mathsf{N}}}^{\mathsf{N}}\frac{\cosh 2\lambda
-\cosh 2\zeta _{b}^{(h_{b})}}{\cosh 2\zeta _{a}^{(h_{a})}-\cosh 2\zeta
_{b}^{(h_{b})}}\mathsf{D}_{+}(\zeta _{a}^{(h_{a})})  \notag \\
& \times \langle +,\text{\textbf{h}}|\text{T}_{a}^{\varphi
_{a}}+\det_{q}M(0)\cosh (\lambda +\eta /2)\prod_{b=1}^{\mathsf{N}}\frac{%
\cosh 2\lambda -\cosh 2\zeta _{b}^{(h_{b})}}{\cosh \eta -\cosh 2\zeta
_{b}^{(h_{b})}}\langle +,\text{\textbf{h}}|  \notag \\
& +(-1)^{\mathsf{N}+1}\coth \zeta _{+}\det_{q}M(i\pi /2)\sinh (\lambda +\eta
/2)\prod_{b=1}^{\mathsf{N}}\frac{\cosh 2\lambda -\cosh 2\zeta _{b}^{(h_{b})}%
}{\cosh \eta +\cosh 2\zeta _{b}^{(h_{b})}}\langle +,\text{\textbf{h}}|,
\label{ADMFKL-SOV D+}
\end{align}%
where:%
\begin{equation}
\langle +,h_{1},...,h_{a},...,h_{\mathsf{N}}|\text{T}_{a}^{\pm }=\langle
+,h_{1},...,h_{a}\pm 1,...,h_{\mathsf{N}}|.
\end{equation}%
Indeed, the representation of $\mathcal{A}_{+}(\lambda )$ follows from the
identity (\ref{ADMFKSym-A-D+}) while $\mathcal{C}_{+}(\lambda )$ is uniquely
defined by the quantum determinant relation.\smallskip

\textsf{II)} \underline{Right $\mathcal{B}_{+}(\lambda )$ SOV-representations%
} \ If $\left( \ref{ADMFKE-SOV}\right) $ is satisfied and $b_{+}\left( \lambda
\right) \neq 0$, then the states:%
\begin{equation}
|+,h_{1},...,h_{\mathsf{N}}\rangle \equiv \frac{1}{\text{\textsc{n}}_{+}}%
\prod_{n=1}^{\mathsf{N}}\left( \frac{\mathcal{A}_{+}(\zeta _{n}^{(0)})}{%
k_{n}^{(+)}\mathsf{D}_{+}(-\zeta _{n}^{(1)})}\right) ^{h_{n}}\overline{%
|0\rangle },  \label{ADMFKD-right-eigenstates+}
\end{equation}%
where:%
\begin{equation}
k_{n}^{(+)}=\frac{\sinh (2\xi _{n}-\eta )}{\sinh (2\xi _{n}+\eta )},
\end{equation}%
\ $h_{n}\in \{0,1\},$ $n\in \{1,...,\mathsf{N}\}$, define a $\mathcal{B}%
_{+}(\lambda )$-eigenbasis of $\mathcal{R}_{\mathsf{N}}$:%
\begin{equation}
\mathcal{B}_{+}(\lambda )|+,\text{\textbf{h}}\rangle =|+,\text{\textbf{h}}%
\rangle \text{\textsc{b}}_{+,\text{\textbf{h}}}(\lambda ).
\end{equation}%
On the generic state $|+$, \textbf{h}$\rangle $, the action of the remaining reflection algebra generators follows by:
\begin{align}
\mathcal{A}_{+}(\lambda )|+,\text{\textbf{h}}\rangle & =\sum_{a=1}^{2\mathsf{%
N}}\text{T}_{a}^{\varphi _{a}}|+,\text{\textbf{h}}\rangle \frac{\sinh
(2\lambda +\eta )\sinh (\lambda +\zeta _{a}^{(h_{a})})}{\sinh (2\zeta
_{a}^{(h_{a})}+\eta )\sinh 2\zeta _{a}^{(h_{a})}}\prod_{\substack{ b=1  \\ %
b\neq a\text{ mod}\mathsf{N}}}^{\mathsf{N}}\frac{\cosh 2\lambda -\cosh
2\zeta _{b}^{(h_{b})}}{\cosh 2\zeta _{a}^{(h_{a})}-\cosh 2\zeta
_{b}^{(h_{b})}}\mathsf{A}_{+}(\zeta _{a}^{(h_{a})})  \notag \\
& +|+,\text{\textbf{h}}\rangle \det_{q}M(0)\cosh (\lambda +\eta
/2)\prod_{b=1}^{\mathsf{N}}\frac{\cosh 2\lambda -\cosh 2\zeta _{b}^{(h_{b})}%
}{\cosh \eta -\cosh 2\zeta _{b}^{(h_{b})}}  \notag \\
& +(-1)^{\mathsf{N}+1}|+,\text{\textbf{h}}\rangle \coth \zeta
_{+}\det_{q}M(i\pi /2)\sinh (\lambda +\eta /2)\prod_{b=1}^{\mathsf{N}}\frac{%
\cosh 2\lambda -\cosh 2\zeta _{b}^{(h_{b})}}{\cosh \eta +\cosh 2\zeta
_{b}^{(h_{b})}},  \label{ADMFKR-SOV A+}
\end{align}%
where:%
\begin{equation}
\mathsf{A}_{+}(\zeta _{a}^{(h_{a})})=\left( k_{a}^{(+)}\right) ^{\varphi
_{a}}\mathsf{D}_{+}(\zeta _{a}^{(h_{a})}-2\varphi _{a}\xi _{a}),\text{\ \ \
\ \ T}_{a}^{\pm }|+,h_{1},...,h_{a},...,h_{\mathsf{N}}\rangle
=|+,h_{1},...,h_{a}\pm 1,...,h_{\mathsf{N}}\rangle .
\end{equation}%
Indeed, the representation of $\mathcal{D}_{+}(\lambda )$ follows from the
identity (\ref{ADMFKSym-A-D+}) while $\mathcal{C}_{+}(\lambda )$ is uniquely
defined by the quantum determinant relation.\smallskip
\end{theorem}

\begin{proof}[Proof of \textsf{I)}]
The proof is given along the same line delineated in the previous theorem,
we just need to do the following remarks. First of all let us give the
(boundary-bulk) decomposition of the reflection algebra generator:%
\begin{equation}
\mathcal{B}_{+}(\lambda )=B(\lambda )D(-\lambda )a_{+}(\lambda )+D(\lambda
)D(-\lambda )b_{+}(\lambda )-B(\lambda )B(-\lambda )c_{+}(\lambda
)-D(\lambda )B(-\lambda )d_{+}(\lambda ),  \label{ADMFKBB-dec-B+}
\end{equation}%
in terms of the generators of the Yang-Baxter algebra. Then, the properties (%
\ref{ADMFKLeft-reference}) imply that $\langle 0|$ is a $\mathcal{B}_{+}(\lambda
) $-eigenstate with non-zero eigenvalue:%
\begin{equation}
\text{\textsc{b}}_{+,\text{\textbf{1}}}(\lambda )\equiv \kappa _{+}e^{\tau
_{+}}\frac{\sinh (2\lambda +\eta )}{\sinh \zeta _{+}}d(\lambda )d(-\lambda ).
\end{equation}%
The proof of the point \textsf{I)} is based on the following reflection
algebra commutation relations:%
\begin{eqnarray}
\mathcal{D}_{+}(\lambda _{1})\mathcal{B}_{+}(\lambda _{2}) &=&\frac{\sinh
(\lambda _{1}-\lambda _{2}+\eta )\sinh (\lambda _{2}+\lambda _{1}+\eta )}{%
\sinh (\lambda _{1}-\lambda _{2})\sinh (\lambda _{1}+\lambda _{2})}\mathcal{B%
}_{+}(\lambda _{2})\mathcal{D}_{+}(\lambda _{1})  \notag \\
&&-\frac{\sinh (2\lambda _{2}+\eta )\sinh \eta }{\sinh (\lambda _{1}-\lambda
_{2})\sinh 2\lambda _{2}}\mathcal{B}_{+}(\lambda _{1})\mathcal{D}%
_{+}(\lambda _{2})  \notag \\
&&-\frac{\sinh \eta \sinh (2\lambda _{2}+\eta )}{\sinh (\lambda _{1}+\lambda
_{2})\sinh 2\lambda _{2}}\mathcal{B}_{+}(\lambda _{1})\mathcal{D}%
_{+}(-\lambda _{2}),
\end{eqnarray}%
and on the identities:%
\begin{equation}
\mathcal{U}_{+}(-\eta /2)=\det_{q}M(0)\text{ }I_{0},\text{ \ \ }\mathcal{U}%
_{+}(-\eta /2+i\pi /2)=i\coth \zeta _{+}\det_{q}M(i\pi /2)\text{ }\sigma
_{0}^{z}.  \label{ADMFKU+identities}
\end{equation}%
By using them and the fact that $\mathcal{D}_{+}(\lambda )$ has the
following functional dependence w.r.t. $\lambda $:%
\begin{equation}
\mathcal{D}_{+}(\lambda )=\sum_{a=0}^{2\mathsf{N}+1}e^{\left( 2a-2\mathsf{N}%
+1\right) \lambda }\mathcal{D}_{+,a}
\end{equation}%
we get the following interpolation formula for the action on $\langle +$, \textbf{h}$|$:%
\begin{align}
\langle +,\text{\textbf{h}}|\mathcal{D}_{+}(\lambda )& =\sum_{a=1}^{2\mathsf{%
N}}\prod_{\substack{ b=1  \\ b\neq a}}^{2\mathsf{N}+2}\frac{\sinh (\lambda
-\zeta _{b}^{(h_{b})})}{\sinh (\zeta _{a}^{(h_{a})}-\zeta _{b}^{(h_{b})})}%
\mathsf{D}_{+}(\zeta _{a}^{(h_{a})})\langle +,\text{\textbf{h}}|\text{T}%
_{a}^{\varphi _{a}}  \notag \\
& +\det_{q}M(0)\prod_{\substack{ b=1  \\ b\neq 2\mathsf{N}+1}}^{2\mathsf{N}%
+2}\frac{\sinh (\lambda -\zeta _{b}^{(h_{b})})}{\sinh (\zeta _{2\mathsf{N}%
+1}^{(1)}-\zeta _{b}^{(h_{b})})}\langle +,\text{\textbf{h}}|  \notag \\
& -i\coth \zeta _{+}\det_{q}M(i\pi /2)\prod_{b=1}^{2\mathsf{N}+1}\frac{\sinh
(\lambda -\zeta _{b}^{(h_{b})})}{\sinh (\zeta _{2\mathsf{N}+2}^{(1)}-\zeta
_{b}^{(h_{b})})}\langle +,\text{\textbf{h}}|,
\end{align}%
where we have denoted $\zeta _{2\mathsf{N}+b}(h_{2\mathsf{N}+b})=\zeta _{2%
\mathsf{N}+b}^{(1)}$ for $b=1,2$ and%
\begin{equation}
\zeta _{2\mathsf{N}+1}^{(1)}=-\eta /2,\text{ \ \ \ \ }\zeta _{2\mathsf{N}%
+2}^{(1)}=-\eta /2+i\pi /2.
\end{equation}%
Then, it is a simple exercise to rewrite this in the form $\left( \ref{ADMFKL-SOV
D+}\right) $.
\end{proof}

\begin{proof}[Proof of \textsf{II)}]
Similarly, by using (\ref{ADMFKBB-dec-B+}) and (\ref{ADMFKRight-reference}), it
follows that $\overline{|0\rangle }$ is a $\mathcal{B}_{+}(\lambda )$%
-eigenstate with non-zero eigenvalue:%
\begin{equation}
\text{\textsc{b}}_{+,\text{\textbf{0}}}(\lambda )\equiv \kappa _{+}e^{\tau
_{+}}\frac{\sinh (2\lambda +\eta )}{\sinh \zeta _{+}}a(\lambda )a(-\lambda ).
\end{equation}%
Then the proof of the point \textsf{II)} is based on the following
reflection algebra commutation relations:%
\begin{eqnarray}
\mathcal{B}_{+}(\lambda _{2})\mathcal{A}_{+}(\lambda _{1}) &=&\frac{\sinh
(\lambda _{1}-\lambda _{2}+\eta )\sinh (\lambda _{2}+\lambda _{1}+\eta )}{%
\sinh (\lambda _{1}-\lambda _{2})\sinh (\lambda _{1}+\lambda _{2})}\mathcal{A%
}_{+}(\lambda _{1})\mathcal{B}_{+}(\lambda _{2})  \notag \\
&&+\frac{\sinh \eta \sinh (2\lambda _{2}+\eta )}{\sinh (\lambda _{2}-\lambda
_{1})\sinh 2\lambda _{2}}\mathcal{A}_{+}(\lambda _{2})\mathcal{B}%
_{+}(\lambda _{1})  \notag \\
&&+\frac{\sinh \eta \sinh (2\lambda _{2}+\eta )}{\sinh (\lambda _{1}+\lambda
_{2})\sinh 2\lambda _{2}}\mathcal{A}_{+}(-\lambda _{2})\mathcal{B}%
_{+}(\lambda _{1}).
\end{eqnarray}%
By using them we get the following interpolation formula for the action on $%
|+,\text{\textbf{h}}\rangle $:%
\begin{align}
\mathcal{A}_{+}(\lambda )|+,\text{\textbf{h}}\rangle =& \sum_{a=1}^{2%
\mathsf{N}}\text{T}_{a}^{\varphi _{a}}|+,\text{\textbf{h}}\rangle \prod 
_{\substack{ b=1  \\ b\neq a}}^{2\mathsf{N}+2}\frac{\sinh (\lambda -\zeta
_{b}^{(h_{b})})}{\sinh (\zeta _{a}^{(h_{a})}-\zeta _{b}^{(h_{b})})}\mathsf{A}%
_{+}(\zeta _{a}^{(h_{a})})  \notag \\
& +|+,\text{\textbf{h}}\rangle \det_{q}M(0)\prod_{\substack{ b=1  \\ b\neq 2%
\mathsf{N}+1}}^{2\mathsf{N}+2}\frac{\sinh (\lambda -\zeta _{b}^{(h_{b})})}{%
\sinh (\zeta _{2\mathsf{N}+1}^{(1)}-\zeta _{b}^{(h_{b})})}  \notag \\
& +i|+,\text{\textbf{h}}\rangle \coth \zeta _{+}\det_{q}M(i\pi
/2)\prod_{b=1}^{2\mathsf{N}+1}\frac{\sinh (\lambda -\zeta _{b}^{(h_{b})})}{%
\sinh (\zeta _{2\mathsf{N}+2}^{(1)}-\zeta _{b}^{(h_{b})})},
\end{align}%
which can be rewritten in the form $\left( \ref{ADMFKR-SOV A+}\right) $.
\end{proof}

\section{SOV-decomposition of the identity}

The action of a generic left $\mathcal{B}_{\pm }$-eigenstate on a generic
right $\mathcal{B}_{\pm }$-eigenstate are here compute in this way allowing
to write the decomposition of the identity in the corresponding basis. It is
worth remarking that for the Hermitian conjugation properties of Proposition %
\ref{ADMFKHerm-conj} these results correspond to the computations of
scalar products between $\mathcal{B}_{\pm }$-eigenvectors and $\mathcal{C}%
_{\pm }$-eigenvectors. We show that up to an overall constant these are
completely fixed by the left and right SOV-representations of the
Yang-Baxter algebras when the gauge in the SOV-representations are chosen.

\subsection{Change of basis properties}
Let us present the main properties of the $2^{\mathsf{N}}\times 2^{\mathsf{N}}$ matrices $%
U^{(L,\epsilon )}$ and $U^{(R,\epsilon )}$:%
\begin{equation}
\langle \epsilon ,\text{\textbf{h}}|=\langle \text{\textbf{h}}%
|U^{(L,\epsilon )}=\sum_{i=1}^{2^{\mathsf{N}}}U_{\varkappa \left( \text{%
\textbf{h}}\right) ,i}^{(L,\epsilon )}\langle \varkappa ^{-1}\left( i\right)
|\text{ \ \ and\ \ \ }|\epsilon ,\text{\textbf{h}}\rangle =U^{(R,\epsilon )}|%
\text{\textbf{h}}\rangle =\sum_{i=1}^{2^{\mathsf{N}}}U_{i,\varkappa \left( 
\text{\textbf{h}}\right) }^{(R,\epsilon )}|\varkappa ^{-1}\left( i\right)
\rangle ,
\end{equation}%
which define the change of basis to the SOV-basis starting from the original
spin basis:%
\begin{equation}
\langle \text{\textbf{h}}|\equiv \otimes _{n=1}^{\mathsf{N}}\langle
2h_{n}-1,n|\text{ \ \ \ \ and \ \ \ }|\text{\textbf{h}}\rangle \equiv
\otimes _{n=1}^{\mathsf{N}}|2h_{n}-1,n\rangle ,
\end{equation}%
where $\varkappa $ is the following natural isomorphism between the sets $%
\{0,1\}^{\mathsf{N}}$ and $\{1,...,2^{\mathsf{N}}\}$: 
\begin{equation}
\varkappa :\text{\textbf{h}}\in \{0,1\}^{\mathsf{N}}\rightarrow \varkappa
\left( \text{\textbf{h}}\right) \equiv 1+\sum_{a=1}^{\mathsf{N}%
}2^{(a-1)}h_{a}\in \{1,...,2^{\mathsf{N}}\}.  \label{ADMFKcorrisp}
\end{equation}%
Note that the matrices $U^{(L,\epsilon )}$ and $U^{(R,\epsilon )}$\ are
invertible matrices for the diagonalizability of $\mathcal{B}_{\epsilon
}(\lambda )$:
\begin{equation}
U^{(L,\epsilon )}\mathcal{B}_{\epsilon }(\lambda )=\Delta _{\mathcal{B}%
_{\epsilon }}(\lambda )U^{(L,\epsilon )},\text{ \ \ }\mathcal{B}_{\epsilon
}(\lambda )U^{(R,\epsilon )}=U^{(R,\epsilon )}\Delta _{\mathcal{B}_{\epsilon
}}(\lambda ).
\end{equation}
Here $\Delta _{\mathcal{B}_{\epsilon }}(\lambda )$ is the $2^{\mathsf{N}}\times 2^{\mathsf{N}}$ diagonal matrix whose elements, for the simplicity of the $\mathcal{B}_{\epsilon }$-spectrum, read:
\begin{equation}
\left( \Delta _{\mathcal{B}_{\epsilon }}(\lambda )\right) _{i,j}\equiv
\delta _{i,j}\text{\textsc{b}}_{\epsilon ,\varkappa ^{-1}\left( i\right)
}(\lambda )\text{ \ }\forall i,j\in \{1,...,2^{\mathsf{N}}\}.
\end{equation}%
Moreover, it holds:

\begin{proposition}
The $2^{\mathsf{N}}\times 2^{\mathsf{N}}$ matrix:%
\begin{equation}
M^{(\epsilon )}\equiv U^{(L,\epsilon )}U^{(R,\epsilon )}
\end{equation}%
is diagonal and it is characterized by:%
\begin{equation}
M_{\varkappa \left( \text{\textbf{h}}\right) \varkappa \left( \text{\textbf{h%
}}\right) }^{(\epsilon )}=\langle \epsilon ,\text{\textbf{h}}|\epsilon ,%
\text{\textbf{h}}\rangle =\prod_{1\leq b<a\leq \mathsf{N}}\frac{1}{\eta
_{a}^{(h_{a})}-\eta _{b}^{(h_{b})}}.  \label{ADMFKM_jj}
\end{equation}
\end{proposition}

\begin{proof}
Note that being the action of a left $\mathcal{B}_{\epsilon }$-eigenstate on a
right $\mathcal{B}_{\epsilon }$-eigenstate, corresponding to different $%
\mathcal{B}_{\epsilon }$-eigenvalues, zero this implies that the matrix $%
M^{(\epsilon )}$ is diagonal; then to compute its diagonal elements we
compute the matrix elements $\theta _{a}^{\left( -\right) }\equiv \langle
-,h_{1},...,h_{a}=0,...,h_{\mathsf{N}}|\mathcal{A}_{-}(\xi _{a}+\eta
/2)|-,h_{1},...,h_{a}=1,...,h_{\mathsf{N}}\rangle $, where $a\in \{1,...,%
\mathsf{N}\}$. Using the left action of the operator $\mathcal{A}%
_{-}(\xi _{a}+\eta /2)$ we get:%
\begin{eqnarray}
\theta _{a}^{\left( -\right) } &=&\mathsf{A}_{-}(\eta /2-\xi _{a})\frac{%
\sinh \eta }{\sinh (2\xi _{a}-\eta )}\prod_{\substack{ b=1  \\ b\neq a}}^{%
\mathsf{N}}\frac{\cosh 2\zeta _{a}^{(1)}-\cosh 2\zeta _{b}^{(h_{b})}}{\cosh
2\zeta _{a}^{(0)}-\cosh 2\zeta _{b}^{(h_{b})}}  \notag \\
&&\times \left. \langle -,h_{1},...,h_{a}=1,...,h_{\mathsf{N}%
}|-,h_{1},...,h_{a}=1,...,h_{\mathsf{N}}\rangle \right.
\end{eqnarray}%
while using the decomposition (\ref{ADMFKSym-A-D-}) and the fact that:%
\begin{equation}
\mathcal{D}_{-}(-\xi _{a}-\eta /2)|-,h_{1},...,h_{a}=1,...,h_{\mathsf{N}%
}\rangle =\text{\b{0}}
\end{equation}%
it holds:%
\begin{eqnarray}
\mathcal{A}_{-}(\xi _{a}+\eta /2)|-,h_{1},...,h_{a}=1,...,h_{\mathsf{N}%
}\rangle &=&\frac{k_{a}\sinh \eta }{\sinh (2\xi _{a}+\eta )}\mathsf{A}%
_{-}(\eta /2-\xi _{a})  \notag \\
&\times &|-,h_{1},...,h_{a}=0,...,h_{\mathsf{N}}\rangle ,
\end{eqnarray}%
and then we get:%
\begin{equation}
\theta _{a}^{\left( -\right) }=\frac{\sinh \eta }{\sinh (2\xi _{a}-\eta )}%
\mathsf{A}_{-}(\eta /2-\xi _{a})\langle -,h_{1},...,h_{a}=0,...,h_{\mathsf{N}%
}|-,h_{1},...,h_{a}=0,...,h_{\mathsf{N}}\rangle .
\end{equation}%
so that it holds:%
\begin{equation}
\frac{\langle -,h_{1},...,h_{a}=1,...,h_{\mathsf{N}%
}|-,h_{1},...,h_{a}=1,...,h_{\mathsf{N}}\rangle }{\langle
-,h_{1},...,h_{a}=0,...,h_{\mathsf{N}}|-,h_{1},...,h_{a}=0,...,h_{\mathsf{N}%
}\rangle }=\prod_{\substack{ b=1  \\ b\neq a}}^{\mathsf{N}}\frac{\cosh
2\zeta _{a}^{(0)}-\cosh 2\zeta _{b}^{(h_{b})}}{\cosh 2\zeta _{a}^{(1)}-\cosh
2\zeta _{b}^{(h_{b})}},  \label{ADMFKF1}
\end{equation}%
from which one can prove:%
\begin{equation}
\frac{\langle -,h_{1},...,h_{\mathsf{N}}|-,h_{1},...,h_{\mathsf{N}}\rangle }{%
\langle -,1,...,1|-,1,...,1\rangle }=\prod_{1\leq b<a\leq \mathsf{N}}\frac{%
\eta _{a}^{\left( 1\right) }-\eta _{b}^{\left( 1\right) }}{\eta _{a}^{\left(
h_{a}\right) }-\eta _{b}^{\left( h_{b}\right) }}.  \label{ADMFKF2}
\end{equation}%
This prove the proposition for $\epsilon =-$, being
\begin{equation}
\langle -,1,...,1|-,1,...,1\rangle =\prod_{1\leq b<a\leq \mathsf{N}}\frac{1}{%
\eta _{a}^{\left( 1\right) }-\eta _{b}^{\left( 1\right) }},
\end{equation}
 by our definition of the normalization \textsc{n}$_{-}$. Similarly for $\epsilon =+$, we compute the matrix elements $\theta
_{a}^{\left( +\right) }\equiv \langle +,h_{1},...,h_{a}=1,...,h_{\mathsf{N}}|%
\mathcal{D}_{+}(\xi _{a}-\eta /2)|+,h_{1},...,h_{a}=0,...,h_{\mathsf{N}%
}\rangle $, where $a\in \{1,...,\mathsf{N}\}$. Then using the left action of
the operator $\mathcal{D}_{+}(\xi _{a}-\eta /2)$ we get:%
\begin{eqnarray}
\theta _{a}^{\left( +\right) } &=&-\mathsf{D}_{+}(-\xi _{a}-\eta /2)\frac{%
\sinh \eta }{\sinh (2\xi _{a}+\eta )}\prod_{\substack{ b=1  \\ b\neq a}}^{%
\mathsf{N}}\frac{\cosh 2\zeta _{a}^{(0)}-\cosh 2\zeta _{b}^{(h_{b})}}{\cosh
2\zeta _{a}^{(1)}-\cosh 2\zeta _{b}^{(h_{b})}}  \notag \\
&&\times \left. \langle +,h_{1},...,h_{a}=0,...,h_{\mathsf{N}%
}|+,h_{1},...,h_{a}=0,...,h_{\mathsf{N}}\rangle \right.
\end{eqnarray}%
while using the decomposition (\ref{ADMFKSym-A-D+}) and the fact that:%
\begin{equation}
\mathcal{A}_{+}(-\xi _{a}+\eta /2)|+,h_{1},...,h_{a}=0,...,h_{\mathsf{N}%
}\rangle =\text{\b{0}}
\end{equation}%
it holds:%
\begin{eqnarray}
\mathcal{D}_{+}(\xi _{a}-\eta /2)|+,h_{1},...,h_{a}=0,...,h_{\mathsf{N}%
}\rangle &=&-\frac{k_{a}^{(+)}\sinh \eta }{\sinh (2\xi _{a}-\eta )}\mathsf{D}%
_{+}(-\xi _{a}-\eta /2)  \notag \\
&\times &|+,h_{1},...,h_{a}=1,...,h_{\mathsf{N}}\rangle ,
\end{eqnarray}%
and then we get:%
\begin{equation}
\theta _{a}=-\frac{\sinh \eta }{\sinh (2\xi _{a}+\eta )}\mathsf{D}_{+}(-\xi
_{a}-\eta /2)\langle +,h_{1},...,h_{a}=1,...,h_{\mathsf{N}%
}|+,h_{1},...,h_{a}=1,...,h_{\mathsf{N}}\rangle
\end{equation}%
and so:%
\begin{equation}
\frac{\langle +,h_{1},...,h_{a}=1,...,h_{\mathsf{N}%
}|+,h_{1},...,h_{a}=1,...,h_{\mathsf{N}}\rangle }{\langle
+,h_{1},...,h_{a}=0,...,h_{\mathsf{N}}|+,h_{1},...,h_{a}=0,...,h_{\mathsf{N}%
}\rangle }=\prod_{\substack{ b=1  \\ b\neq a}}^{\mathsf{N}}\frac{\cosh
2\zeta _{a}^{(0)}-\cosh 2\zeta _{b}^{(h_{b})}}{\cosh 2\zeta _{a}^{(1)}-\cosh
2\zeta _{b}^{(h_{b})}},  \label{ADMFKF1+}
\end{equation}%
from which we have:%
\begin{equation}
\frac{\langle +,h_{1},...,h_{\mathsf{N}}|+,h_{1},...,h_{\mathsf{N}}\rangle }{%
\langle +,0,...,0|+,0,...,0\rangle }=\prod_{1\leq b<a\leq \mathsf{N}}\frac{%
\eta _{a}^{\left( 0\right) }-\eta _{b}^{\left( 0\right) }}{\eta _{a}^{\left(
h_{a}\right) }-\eta _{b}^{\left( h_{b}\right) }},
\end{equation}%
which proves the proposition, being
\begin{equation}
\langle +,0,...,0|+,0,...,0\rangle =\prod_{1\leq b<a\leq \mathsf{N}}\frac{1}{%
\eta _{a}^{\left( 0\right) }-\eta _{b}^{\left( 0\right) }}.
\end{equation}
 by our definition of the normalization 
\textsc{n}$_{+}$.
\end{proof}

\subsection{SOV-decomposition of the identity}

The following spectral decomposition of the identity $\mathbb{I}$:%
\begin{equation}
\mathbb{I}\equiv \sum_{i=1}^{2^{\mathsf{N}}}\mu _{i}|\epsilon ,\varkappa
^{-1}\left( i\right) \rangle \langle \epsilon ,\varkappa ^{-1}\left(
i\right) |,
\end{equation}%
can be given in terms of the left and right SOV-basis, where the $\mu
_{i}\equiv \left( \langle \epsilon ,\varkappa ^{-1}\left( i\right) |\epsilon
,\varkappa ^{-1}\left( i\right) \rangle \right) ^{-1}$ is the analogous of
the so-called Sklyanin's measure\footnote{%
Sklyanin's measure has been first introduced by Sklyanin in the quantum Toda
chain \cite{ADMFKSk1}, see also \cite{ADMFKSm98} and \cite{ADMFKBT06} for further
discussions.} in our 6-vertex reflection algebra representations. Now using
the result of the previous section we can explicitly write:%
\begin{equation}
\mathbb{I}\equiv \sum_{h_{1},...,h_{\mathsf{N}}=0}^{1}\prod_{1\leq b<a\leq 
\mathsf{N}}(\eta _{a}^{(h_{a})}-\eta _{a}^{(h_{a})})|\epsilon ,h_{1},...,h_{%
\mathsf{N}}\rangle \langle \epsilon ,h_{1},...,h_{\mathsf{N}}|.
\label{ADMFKDecomp-Id}
\end{equation}

\section{SOV characterization of $\mathcal{T}_\epsilon(\protect\lambda )$-spectrum}

Here we construct  in the $\mathcal{B}_{\pm }$-SOV-representations the spectrum of the transfer matrix $\mathcal{T}_\pm(\lambda )$
for the class of boundary conditions listed in Theorem \ref{ADMFKTh1}. Let us start giving the following
characterization:

\begin{lemma}
Let us denote with $\Sigma _{\mathcal{T}}$ the set of the eigenvalue
functions of the transfer matrix $\mathcal{T}_\epsilon(\lambda )$, then any $\tau_\epsilon
(\lambda )\in \Sigma _{\mathcal{T}}$ is an even function of $\lambda $\ of
the form:%
\begin{align}
\tau_\epsilon(\lambda )& =2\sinh (\lambda -\eta /2)\sinh (\lambda +\eta /2)\coth
\zeta _{-}\coth \zeta _{+}\det_{q}M(i\pi /2)  \notag \\
& +2\cosh (\lambda -\eta /2)\cosh (\lambda +\eta /2)\det_{q}M(0)  \notag \\
& +\sinh (2\lambda -\eta )\sinh (2\lambda +\eta )\sum_{b=1}^{\mathsf{N}+1%
}c_{b}^{\tau_\epsilon}(\cosh 2\lambda )^{b-1}.  \label{ADMFKset-t}
\end{align}
\end{lemma}

\begin{proof}
The transfer matrix $\mathcal{T}_\epsilon(\lambda )$\ is an even function of $\lambda 
$ of maximal degree $\mathsf{N}+2$ in cosh$2\lambda$ so the same is true for the $\tau_\epsilon(\lambda )\in \Sigma _{\mathcal{T}}$ .
Moreover, from the identities (\ref{ADMFKU-identities}) and (\ref{ADMFKU+identities})
after some simple computation the following identities are derived:%
\begin{align}
\mathcal{T}_\epsilon(\pm \eta /2)&=\mathsf{a}_{\epsilon }(\epsilon \eta
/2)\det_{q}M(0)= 2\cosh \eta \det_{q}M(0), \\
\mathcal{T}_\epsilon(\pm (\eta /2-i\pi /2))&=i\mathsf{a}_{\epsilon }(\epsilon \eta
/2-\epsilon i\pi /2)\coth \zeta _{+}\det_{q}M(i\pi /2)  \notag \\
&=-2\cosh \eta \coth\zeta _{-}\coth \zeta _{+}\det_{q}M(i\pi /2),
\end{align}%
both for $\epsilon =+,-$. This identities together with the known functional
form of $\mathcal{T}_\epsilon(\lambda )$ w.r.t. $\lambda $ imply the statement in the
lemma.
\end{proof}

\textbf{Remark: }It is worth writing out explicitly how the triangular
transfer matrices $\mathcal{T}_{\epsilon }(\lambda )$, here considered, are
defined within the original representations of the boundary matrices.

a) In the case $\mathcal{T}_{\epsilon }(\lambda )=\mathcal{T}_{\backslash
}^{\left( \epsilon \right) }(\lambda )$, we have that the boundary matrix $%
K_{-\epsilon }(\lambda )$ is just diagonal, which corresponds to impose the
boundary parameter $k_{-\epsilon }=0$ . Then, the transfer matrix $\mathcal{T%
}_{\epsilon }(\lambda )$ is a degree  $\mathsf{N}+1$ polynomial in cosh\,2$\lambda $ so
that the coefficient $c_{ \mathsf{N}+1}^{\tau_\epsilon}$ in $\left( \ref{ADMFKset-t}%
\right) $ is zero. Moreover, in this case the boundary terms in the
Hamiltonian read:%
\begin{eqnarray}
&&\frac{\sinh \eta }{\sinh \zeta _{\epsilon }}\left[ \sigma _{\epsilon
}^{z}\cosh \zeta _{\epsilon }+2\kappa _{\epsilon }(\sigma _{\epsilon
}^{x}\cosh \tau _{\epsilon }+i\sigma _{\epsilon }^{y}\sinh \tau _{\epsilon })%
\right]   \notag \\
&&+\frac{\sinh \eta }{\sinh \zeta _{-\epsilon }}\sigma _{-\epsilon
}^{z}\cosh \zeta _{-\epsilon }.
\end{eqnarray}

b) In the case $\mathcal{T}_{\epsilon }(\lambda )\neq \mathcal{T}%
_{\backslash }^{\left( \epsilon \right) }(\lambda )$, we have that the
boundary matrix $K_{-\epsilon }(\lambda )$ is properly lower triangular,
which corresponds to impose $k_{-\epsilon }=h_{-\epsilon }e^{\tau
_{-\epsilon }}$ with $h_{-\epsilon }$ finite while sending $\tau _{-\epsilon
}\longrightarrow -\infty $, so that it holds%
\begin{equation}
K_{-\epsilon }(\lambda )=\frac{1}{\sinh \zeta _{-\epsilon }}%
\begin{pmatrix}
\sinh (\zeta _{-\epsilon }+\lambda -\eta /2) & 0 \\ 
h_{-\epsilon }\sinh (2\lambda -\eta ) & \sinh (\zeta _{-\epsilon }-\lambda
+\eta /2)%
\end{pmatrix}%
_{0}.
\end{equation}%
Then, the transfer matrix $\mathcal{T}_{\epsilon }(\lambda )$ is a degree
 $\mathsf{N}+2$ polynomial in cosh\,2$\lambda $ and the leading coefficient is central and
can be easily computed to be:%
\begin{equation}
\lim_{\lambda \rightarrow \pm \infty }\mathcal{T}_{\epsilon }(\lambda
)e^{\mp 2(\mathsf{N}+2)\lambda }=\frac{h_{-\epsilon }\kappa _{\epsilon }\,e^{\tau
_{\epsilon }}}{2^{2\mathsf{N}+2}\,\sinh \zeta _{+}\sinh \zeta _{-}}\,,
\end{equation}%
so that the coefficient in $\left( \ref{ADMFKset-t}\right) $ reads:%
\begin{equation}
c_{\mathsf{N}+1}^{\tau_\epsilon}=\frac{h_{-\epsilon }\kappa _{\epsilon }\,e^{\tau
_{\epsilon }}}{2^{\mathsf{N}}\,\sinh \zeta _{+}\sinh \zeta _{-}},
\end{equation}%
and in this case the boundary terms in the Hamiltonian read:%
\begin{eqnarray}
&&\frac{\sinh \eta }{\sinh \zeta _{\epsilon }}\left[ \sigma _{\epsilon
}^{z}\cosh \zeta _{\epsilon }+2\kappa _{\epsilon }(\sigma _{\epsilon
}^{x}\cosh \tau _{\epsilon }+i\sigma _{\epsilon }^{y}\sinh \tau _{\epsilon })%
\right]   \notag \\
&&+\frac{\sinh \eta }{\sinh \zeta _{-\epsilon }}[\sigma _{-\epsilon
}^{z}\cosh \zeta _{-\epsilon }+h_{-\epsilon }(\sigma _{-\epsilon
}^{x}-i\sigma _{-\epsilon }^{y}{})].
\end{eqnarray}%
where we have used the notations 
\begin{equation}
\sigma _{\epsilon }^{l}=\left\{ 
\begin{array}{c}
\sigma _{1}^{l}\text{ \ for }\epsilon =- \\ 
\sigma _{+}^{l}\text{ \ for }\epsilon =+%
\end{array}%
\right. \text{ for any }l=\{x,y,z\}.
\end{equation}

\subsection{Transfer matrix spectrum in $\mathcal{B}_{-}$-SOV-representations%
}

In this section we characterize the spectrum of the transfer matrix: 
\begin{equation}
\mathcal{T}_{-}(\lambda )\equiv \mathcal{T}_{\setminus }^{(-)}(\lambda
)+c_{+}\left( \lambda \right) \mathcal{B}_{-}(\lambda ),  \label{ADMFKT-}
\end{equation}%
associated to the representations of the reflection algebra under the
following class of boundary parameters:%
\begin{equation}
b_{+}\left( \lambda \right) =0\text{ and }b_{-}\left( \lambda \right) \neq 0.
\label{ADMFKboundary-}
\end{equation}

\begin{theorem}
\label{ADMFKC:T-eigenstates-}If the condition $\left( \ref{ADMFKE-SOV}\right) $ is
satisfied, then $\mathcal{T}_{-}(\lambda )$ has simple spectrum and $\Sigma
_{\mathcal{T}_{-}}$ coincides with the solutions of the discrete system of
equations:%
\begin{equation}
\tau _{-}(\pm \zeta _{a}^{(0)})\tau _{-}(\pm \zeta _{a}^{(1)})=\text{\textsc{%
a}}_{-}(\zeta _{n}^{(1)})\text{\textsc{a}}_{-}(-\zeta _{n}^{(0)}),\text{ \ \ 
}\forall a\in \{1,...,\mathsf{N}\},  \label{ADMFKI-Functional-eq}
\end{equation}%
in the class of functions of the form $(\ref{ADMFKset-t})$, where the coefficient 
\textsc{a}$_{-}(\lambda )$ is defined by: 
\begin{equation}
\text{\textsc{a}}_{-}(\lambda )\equiv \mathsf{a}_{+}(\lambda )\mathsf{A}%
_{-}(\lambda ),
\end{equation}%
and satisfies the quantum determinant condition:%
\begin{equation}
\mathsf{a}_{+}(\lambda )\mathsf{a}_{+}(-\lambda +\eta )\det_{q}\mathcal{U}%
_{-}(\lambda -\eta /2)=\sinh (2\lambda -\eta )\text{\textsc{a}}_{-}(\lambda )%
\text{\textsc{a}}_{-}(-\lambda +\eta ).
\end{equation}

\begin{itemize}
\item[\textsf{I)}] The vector:%
\begin{equation}
|\tau _{-}\rangle =\sum_{h_{1},...,h_{\mathsf{N}}=0}^{1}\prod_{a=1}^{\mathsf{%
N}}Q_{\tau _{-}}(\zeta _{a}^{(h_{a})})\prod_{1\leq b<a\leq \mathsf{N}}(\eta
_{a}^{(h_{a})}-\eta _{b}^{(h_{b})})|-,h_{1},...,h_{\mathsf{N}}\rangle ,
\label{ADMFKeigenT-r-D}
\end{equation}%
defines, uniquely up to an overall normalization, the right $\mathcal{T}_{-}$%
-eigenstate corresponding to $\tau _{-}(\lambda )\in \Sigma _{\mathcal{T}%
_{-}}$. The coefficients in $(\ref{ADMFKeigenT-r-D})$ are characterized by:%
\begin{equation}
Q_{\tau _{-}}(\zeta _{a}^{(1)})/Q_{\tau _{-}}(\zeta _{a}^{(0)})=\tau
_{-}(\zeta _{a}^{(0)})/\text{\textsc{a}}_{-}(-\zeta _{a}^{(0)}).
\label{ADMFKt-Q-relation}
\end{equation}

\item[\textsf{II)}] The covector%
\begin{equation}
\langle \tau _{-}|=\sum_{h_{1},...,h_{\mathsf{N}}=0}^{1}\prod_{a=1}^{\mathsf{%
N}}\bar{Q}_{\tau _{-}}(\zeta _{a}^{(h_{a})})\prod_{1\leq b<a\leq \mathsf{N}%
}(\eta _{a}^{(h_{a})}-\eta _{b}^{(h_{b})})\langle -,h_{1},...,h_{\mathsf{N}%
}|,  \label{ADMFKeigenT-l-D}
\end{equation}%
defines, uniquely up to an overall normalization, the left $\mathcal{T}_{-}$%
-eigenstate corresponding to $\tau _{-}(\lambda )\in \Sigma _{\mathcal{T}%
_{-}}$. The coefficients in $(\ref{ADMFKeigenT-l-D})$ are characterized by:%
\begin{equation}
\bar{Q}_{\tau _{-}}(\zeta _{a}^{(1)})/\bar{Q}_{\tau _{-}}(\zeta
_{a}^{(0)})=\alpha _{a}^{(-)}k_{a}^{(-)}\tau _{-}(\zeta _{a}^{(0)})/\text{%
\textsc{a}}_{-}(\zeta _{a}^{(1)}),  \label{ADMFKt-Qbar-relation}
\end{equation}%
where:%
\begin{equation}
\alpha _{n}^{(-)}=\frac{\mathsf{a}_{+}(\zeta _{a}^{(1)})}{\mathsf{d}%
_{+}(-\zeta _{a}^{(0)})}=\frac{\mathsf{d}_{+}(\zeta _{a}^{(1)})}{\mathsf{a}%
_{+}(-\zeta _{a}^{(0)})}=\frac{\sinh (2\xi _{n}+2\eta )}{k_{n}^{(-)}\sinh
(2\xi _{n}-2\eta )}.
\end{equation}
\end{itemize}
\end{theorem}

\begin{proof}
In the $\mathcal{B}_{-}$-SOV representations the spectral problem for $%
\mathcal{T}_{-}(\lambda )$ is reduced to a discrete system of $2^{\mathsf{N}%
} $ Baxter-like equations 
\begin{equation}
\tau _{-}(\zeta _{n}^{(h_{n})})\Psi _{\tau _{-}}(\text{\textbf{h}})\,=\text{%
\textsc{a}}_{-}(\zeta _{n}^{(h_{n})})\Psi _{\tau _{-}}(\mathsf{T}_{n}^{-}(%
\text{\textbf{h}}))+\text{\textsc{a}}_{-}(-\zeta _{n}^{(h_{n})})\Psi _{\tau
_{-}}(\mathsf{T}_{n}^{+}(\text{\textbf{h}})),  \label{ADMFKSOVBax1}
\end{equation}%
for any$\,n\in \{1,...,\mathsf{N}\}$ and \textbf{h}$\in \{0,1\}^{\mathsf{N}}$%
, in the coefficients (\textit{wave-functions}):%
\begin{equation}
\Psi _{\tau _{-}}(\text{\textbf{h}})\equiv \langle -,h_{1},...,h_{\mathsf{N}%
}|\tau _{-}\rangle ,
\end{equation}%
of $|\tau _{-}\rangle $ the $\mathcal{T}_{-}$-eigenstate associated to $\tau
_{-}(\lambda )\in \Sigma _{\mathcal{T}_{-}}$; here, we have used the
notations:%
\begin{equation}
\mathsf{T}_{n}^{\pm }(\text{\textbf{h}})\equiv (h_{1},\dots ,h_{n}\pm
1,\dots ,h_{\mathsf{N}}).
\end{equation}%
Being:%
\begin{equation}
\text{\textsc{a}}_{-}(\zeta _{n}^{(0)})=\text{\textsc{a}}_{-}(-\zeta
_{n}^{(1)})=0,
\end{equation}%
the previous system of equations $\left( \ref{ADMFKSOVBax1}\right) $ is
equivalent to the following system of homogeneous equations:%
\begin{equation}
\left( 
\begin{array}{cc}
\tau _{-}(\zeta _{n}^{(0)}) & -\text{\textsc{a}}_{-}(-\zeta _{n}^{(0)}) \\ 
-\text{\textsc{a}}_{-}(\zeta _{n}^{(1)}) & \tau _{-}(\zeta _{n}^{(1)})%
\end{array}%
\right) \left( 
\begin{array}{c}
\Psi _{\tau -}(h_{1},...,h_{n}=0,...,h_{1}) \\ 
\Psi _{\tau -}(h_{1},...,h_{n}=1,...,h_{1})%
\end{array}%
\right) =\left( 
\begin{array}{c}
0 \\ 
0%
\end{array}%
\right) ,  \label{ADMFKhomo-system}
\end{equation}%
for any$\,n\in \{1,...,\mathsf{N}\}$ with $h_{m\neq n}\in \{0,1\}$. The
condition $\tau _{-}(\lambda )\in \Sigma _{\mathcal{T}_{-}}$ implies that
the determinants of the $2\times 2$ matrices in $\left( \ref{ADMFKhomo-system}%
\right) $ must be zero for any$\,n\in \{1,...,\mathsf{N}\}$, which is
equivalent to (\ref{ADMFKI-Functional-eq}). Moreover, the rank of the matrices in 
$\left( \ref{ADMFKhomo-system}\right) $ is 1 being%
\begin{equation}
\text{\textsc{a}}_{-}(-\zeta _{n}^{(0)})\neq 0\text{\ \ and \ \textsc{a}}%
_{-}(\zeta _{n}^{(1)})\neq 0,  \label{ADMFKRank1}
\end{equation}%
and then (up to an overall normalization) the solution is unique:%
\begin{equation}
\frac{\Psi _{\tau -}(h_{1},...,h_{n}=1,...,h_{1})}{\Psi _{\tau
-}(h_{1},...,h_{n}=0,...,h_{1})}=\frac{\tau _{-}(\zeta _{a}^{(0)})}{\text{%
\textsc{a}}_{-}(-\zeta _{a}^{(0)})},
\end{equation}%
for any$\,n\in \{1,...,\mathsf{N}\}$ with $h_{m\neq n}\in \{0,1\}$. So fixed 
$\tau _{-}(\lambda )\in \Sigma _{\mathcal{T}_{-}}$ there exists (up to
normalization) one and only one corresponding $\mathcal{T}_{-}$-eigenstate $%
|\tau _{-}\rangle $ with coefficients of the factorized form given in $%
\left( \ref{ADMFKeigenT-r-D}\right) $-$\left( \ref{ADMFKt-Q-relation}\right) $; i.e.
the $\mathcal{T}_{-}$-spectrum is simple.

Vice versa, if $\tau _{-}(\lambda )$ is in the set of functions (\ref{ADMFKset-t}%
) and satisfies (\ref{ADMFKI-Functional-eq}), then the state $|\tau _{-}\rangle $
defined by $\left( \ref{ADMFKeigenT-r-D}\right) $-$\left( \ref{ADMFKt-Q-relation}%
\right) $ satisfies:%
\begin{equation}
\left\langle -,h_{1},...,h_{\mathsf{N}}\right\vert \mathcal{T}_{-}(\zeta
_{n}^{(h_{n})})|\tau _{-}\rangle =\tau _{-}(\zeta _{n}^{(h_{n})})\langle
-,h_{1},...,h_{\mathsf{N}}|\tau _{-}\rangle \text{ \ }\forall n\in \{1,...,%
\mathsf{N}\}
\end{equation}%
for any $\mathcal{B}_{-}$-eigenstate $\left\langle -,h_{1},...,h_{\mathsf{N}%
}\right\vert $ and this implies:%
\begin{equation}
\left\langle -,h_{1},...,h_{\mathsf{N}}\right\vert \mathcal{T}_{-}(\lambda
)|\tau _{-}\rangle =\tau _{-}(\lambda )\langle -,h_{1},...,h_{\mathsf{N}%
}|\tau _{-}\rangle \text{ \ \ }\forall \lambda \in \mathbb{C},
\end{equation}%
i.e. $\tau _{-}(\lambda )\in \Sigma _{\mathcal{T}_{-}}$\ and $|\tau
_{-}\rangle $ is the corresponding $\mathcal{T}_{-}$-eigenstate.

Concerning the left $\mathcal{T}_{-}$-eigenstates the proof is done
as above we have just to remark that in this case the matrix elements:%
\begin{equation}
\langle \tau _{-}|\mathcal{T}_{-}(\zeta _{n}^{(h_{n})})|-,h_{1},...,h_{%
\mathsf{N}}\rangle ,
\end{equation}%
can be computed by using the right $\mathcal{B}_{-}$-representation:%
\begin{equation}
\tau _{-}(\zeta _{n}^{(h_{n})})\bar{\Psi}_{\tau _{-}}(\text{\textbf{h}})\,=%
\text{\textsc{d}}_{-}(\zeta _{n}^{(h_{n})})\bar{\Psi}_{\tau _{-}}(\mathsf{T}%
_{n}^{-}(\text{\textbf{h}}))+\text{\textsc{d}}_{-}(-\zeta _{n}^{(h_{n})})%
\bar{\Psi}_{\tau _{-}}(\mathsf{T}_{n}^{+}(\text{\textbf{h}})),\text{ \ \ }%
\forall n\in \{1,...,\mathsf{N}\}
\end{equation}%
where:%
\begin{equation}
\bar{\Psi}_{\tau _{-}}(\text{\textbf{h}})\equiv \langle \tau
_{-}|-,h_{1},...,h_{\mathsf{N}}\rangle ,
\end{equation}%
and the coefficient \textsc{d}$_{-}(\lambda )$\ reads:%
\begin{equation}
\text{\textsc{d}}_{-}(\lambda )\equiv \mathsf{d}_{+}(\lambda )\mathsf{D}%
_{-}(\lambda ).
\end{equation}
\end{proof}

It is worth pointing out that for the analysis of the continuum limit it is
interesting to get a reformulation of this characterization by functional
equations. The construction of a Baxter Q-operator can play an important
role to achieve this aim. Let us recall that a Q-operator is of a one-parameter
operator family which satisfies properties of the type:%
\begin{equation}
\left[ \mathcal{T}_{-}(\lambda ),Q(\lambda )\right] =0,\text{ }\left[
Q(\lambda ),Q(\mu )\right] =0,\text{ \ }\mathcal{T}_{-}(\lambda )Q(\lambda
)=\alpha \left( \lambda \right) Q(\lambda /q)+\beta \left( \lambda \right)
Q(\lambda q)\,,  \label{ADMFKQ-op-ch}
\end{equation}%
where $\alpha \left( \lambda \right) $ and $\beta \left( \lambda \right) $
are some characteristic functions of the constructed Q-operator. Indeed, if
this functional equation coincides with the discrete system $\left( \ref{ADMFKSOVBax1}\right) $ in the spectrum of the $\mathcal{B}_{\pm }$-zeros we can
use the Q-operator to reformulate by functional equations the SOV spectrum
characterization. A Baxter Q-operator has been constructed in\footnote{%
See also \cite{ADMFKBBOY95} and \cite{ADMFKYB95} for the construction of the Q-operator in the spin 1/2
and higher spin XXZ quantum chain with twisted boundary conditions.} \cite{ADMFKN05,ADMFKMNS06} and \cite{ADMFKYNZ06} for some classes of representations of the
reflection algebra; then, it will be important to make a connection with the
present SOV analysis. Let us observe that for {\it root of unit} $\eta =2i\pi
p^{\prime }/p$ ($p$ and $p^{\prime }\in Z^{\geq 0}$) the existence of
non-trivial solutions of the Baxter equation leads to the functional
equation:%
\begin{equation}
\det_{p}D(\Lambda )=0,\text{ \ \ }\Lambda \in \mathbb{C}
\label{ADMFKfun-eq-T-eigen}
\end{equation}%
where $D(\lambda )$ is the following $p\times p$ matrix:
\begin{equation}
D(\lambda )\equiv 
\begin{pmatrix}
\tau _{-}(\lambda ) & -\beta (\lambda ) & 0 & \cdots & 0 & -\alpha (\lambda )
\\ 
-\alpha (q\lambda ) & \tau _{-}(q\lambda ) & -\beta (q\lambda ) & 0 & \cdots
& 0 \\ 
0 & {\quad }\ddots &  &  &  & \vdots \\ 
\vdots &  & \cdots &  &  & \vdots \\ 
\vdots &  &  & \cdots &  & \vdots \\ 
\vdots &  &  &  & \ddots {\qquad } & 0 \\ 
0 & \ldots & 0 & -\alpha (q^{p-2}\lambda ) & \tau _{-}(q^{p-2}\lambda ) & 
-\beta (q^{p-2}\lambda ) \\ 
-\beta (q^{p-1}\lambda ) & 0 & \ldots & 0 & -\alpha (q^{p-1}\lambda ) & \tau
_{-}(q^{p-1}\lambda )
\end{pmatrix},  \label{ADMFKD-matrix}
\end{equation}
written only in terms of the $\mathcal{T}_{-}$-eigenvalue $\tau _{-}(\lambda )$ and its determinant is clearly a function of $\Lambda =\lambda ^{p}$. Note that the method based on the combined use of the fusion of transfer
matrices \cite{ADMFKKRS,ADMFKKR}\ and the truncation identity for root of unit $\eta$, \cite{ADMFKBR89} and \cite{ADMFKN02, ADMFKN03}, should lead to the same equation $\left( %
\ref{ADMFKfun-eq-T-eigen}\right) $.

\subsection{Transfer matrix spectrum in $\mathcal{B}_{+}$-SOV-representations%
}

In this section we characterize the spectrum of the transfer matrix: 
\begin{equation}
\mathcal{T}_{+}(\lambda )\equiv \mathcal{T}_{\setminus }^{(+)}(\lambda
)+c_{-}\left( \lambda \right) \mathcal{B}_{+}(\lambda ),  \label{ADMFKT+}
\end{equation}%
associated to the representations of the reflection algebra under the
following class of boundary parameters:%
\begin{equation}
b_{-}\left( \lambda \right) =0\text{ and }b_{+}\left( \lambda \right) \neq 0.
\label{ADMFKboundary+}
\end{equation}%
Let us now write the left and right eigenstates of the transfer matrices (%
\ref{ADMFKT+}) in the $\mathcal{B}_{+}$-SOV-representation:

\begin{theorem}
\label{ADMFKC:T-eigenstates+}If $\left( \ref{ADMFKE-SOV}\right) $ is satisfied, then $%
\mathcal{T}_{+}(\lambda )$ has simple spectrum and $\Sigma _{\mathcal{T}%
_{+}} $ coincides with the solutions of the discrete system of equations: 
\begin{equation}
\tau _{+}(\pm \zeta _{a}^{(0)})\tau _{+}(\pm \zeta _{a}^{(1)})=\text{\textsc{%
d}}_{+}(-\zeta _{n}^{(1)})\text{\textsc{d}}_{+}(\zeta _{n}^{(0)}),\text{ \ \ 
}\forall a\in \{1,...,\mathsf{N}\},  \label{ADMFKI-Functional-eq+}
\end{equation}%
in the class of functions of the form (\ref{ADMFKset-t}), where the coefficient 
\textsc{d}$_{+}(\lambda )$ is defined by:%
\begin{equation}
\text{\textsc{d}}_{+}(\lambda )\equiv \mathsf{d}_{-}(\lambda )\mathsf{D}%
_{+}(\lambda ),
\end{equation}%
and satisfies the quantum determinant condition:%
\begin{equation}
\mathsf{d}_{-}(\lambda -\eta /2)\mathsf{d}_{-}(-\lambda -\eta /2)\det_{q}%
\mathcal{U}_{+}(\lambda )=\sinh (2\lambda +2\eta )\text{\textsc{d}}%
_{+}(\lambda -\eta /2)\text{\textsc{d}}_{+}(-\lambda -\eta /2).
\end{equation}

\begin{itemize}
\item[\textsf{I)}] The vector:%
\begin{equation}
|\tau _{+}\rangle =\sum_{h_{1},...,h_{\mathsf{N}}=0}^{1}\prod_{a=1}^{\mathsf{%
N}}Q_{\tau _{+}}(\zeta _{a}^{(h_{a})})\prod_{1\leq b<a\leq \mathsf{N}}(\eta
_{a}^{(h_{a})}-\eta _{b}^{(h_{b})})|+,h_{1},...,h_{\mathsf{N}}\rangle ,
\label{ADMFKeigenT-r-D+}
\end{equation}%
defines, uniquely up to an overall normalization, the right $\mathcal{T}_{+}$%
-eigenstate corresponding to $\tau _{+}(\lambda )\in \Sigma _{\mathcal{T}%
_{+}}$. The coefficients in $(\ref{ADMFKeigenT-r-D+})$ are characterized by:%
\begin{equation}
Q_{\tau _{+}}(\zeta _{a}^{(1)})/Q_{\tau _{+}}(\zeta _{a}^{(0)})=\tau
_{+}(\zeta _{a}^{(0)})/\text{\textsc{d}}_{+}(\zeta _{a}^{(0)}).
\label{ADMFKt-Q-relation+}
\end{equation}

\item[\textsf{II)}] The covector:%
\begin{equation}
\langle \tau _{+}|=\sum_{h_{1},...,h_{\mathsf{N}}=0}^{1}\prod_{a=1}^{\mathsf{%
N}}\bar{Q}_{\tau _{+}}(\zeta _{a}^{(h_{a})})\prod_{1\leq b<a\leq \mathsf{N}%
}(\eta _{a}^{(h_{a})}-\eta _{b}^{(h_{b})})\langle +,h_{1},...,h_{\mathsf{N}%
}|,  \label{ADMFKeigenT-l-D+}
\end{equation}%
defines, uniquely up to an overall normalization, the left $\mathcal{T}_{+}$%
-eigenstate corresponding to $\tau _{+}(\lambda )\in \Sigma _{\mathcal{T}%
_{+}}$. The coefficients in $(\ref{ADMFKeigenT-l-D+})$ are characterized by:%
\begin{equation}
\bar{Q}_{\tau _{+}}(\zeta _{a}^{(1)})/\bar{Q}_{\tau _{+}}(\zeta
_{a}^{(0)})=\tau _{+}(\zeta _{a}^{(0)})/\left( \alpha _{a}^{(+)}k_{a}^{(+)}%
\text{\textsc{d}}_{+}(-\zeta _{a}^{(1)})\right) ,  \label{ADMFKt-Qbar-relation+}
\end{equation}%
where:%
\begin{equation}
\alpha _{a}^{(+)}=\frac{\mathsf{a}_{-}(\zeta _{a}^{(0)})}{\mathsf{d}%
_{-}(-\zeta _{a}^{(1)})}=\frac{\mathsf{d}_{-}(\zeta _{a}^{(0)})}{\mathsf{a}%
_{-}(-\zeta _{a}^{(1)})}=\frac{\sinh (2\xi _{n}-2\eta )}{k_{n}^{(+)}\sinh
(2\xi _{n}+2\eta )}.
\end{equation}
\end{itemize}
\end{theorem}

\begin{proof}
Taken a $\mathcal{T}_{+}$-eigenstate $|\tau _{+}\rangle $ corresponding to
the eigenvalue $\tau _{+}(\lambda )\in \Sigma _{\mathcal{T}_{+}}$, it has
coefficients:%
\begin{equation}
\Psi _{\tau _{+}}(\text{\textbf{h}})\equiv \langle +,h_{1},...,h_{\mathsf{N}%
}|\tau _{+}\rangle ,
\end{equation}%
in the SOV-basis, which satisfies the following discrete system of
Baxter-like equations:%
\begin{equation}
\tau _{+}(\zeta _{n}^{(h_{n})})\Psi _{\tau _{+}}(\text{\textbf{h}})\,=\text{%
\textsc{d}}_{+}(\zeta _{n}^{(h_{n})})\Psi _{\tau _{+}}(\mathsf{T}_{n}^{+}(%
\text{\textbf{h}}))+\text{\textsc{d}}_{+}(-\zeta _{n}^{(h_{n})})\Psi _{\tau
_{+}}(\mathsf{T}_{n}^{-}(\text{\textbf{h}})),
\end{equation}%
for any$\,n\in \{1,...,\mathsf{N}\}$ and \textbf{h}$\in \{0,1\}^{\mathsf{N}}$%
. We can rewrite this as the following system of homogeneous equations:%
\begin{equation}
\left( 
\begin{array}{cc}
\tau _{+}(\zeta _{n}^{(0)}) & -\text{\textsc{d}}_{+}(\zeta _{n}^{(0)}) \\ 
-\text{\textsc{d}}_{+}(-\zeta _{n}^{(1)}) & \tau _{+}(\zeta _{n}^{(1)})%
\end{array}%
\right) \left( 
\begin{array}{c}
\Psi _{\tau _{+}}(h_{1},...,h_{n}=0,...,h_{1}) \\ 
\Psi _{\tau _{+}}(h_{1},...,h_{n}=1,...,h_{1})%
\end{array}%
\right) =\left( 
\begin{array}{c}
0 \\ 
0%
\end{array}%
\right) ,
\end{equation}%
being:%
\begin{equation}
\text{\textsc{d}}_{+}(-\zeta _{n}^{(0)})=\text{\textsc{d}}_{+}(\zeta
_{n}^{(1)})=0,\text{ \ \textsc{d}}_{+}(\zeta _{n}^{(0)})\neq 0\text{\ \ and
\ \textsc{d}}_{+}(-\zeta _{n}^{(1)})\neq 0.
\end{equation}%
From which the condition $\tau _{+}(\lambda )\in \Sigma _{\mathcal{T}_{+}}$
directly implies (\ref{ADMFKI-Functional-eq}) and moreover it has to hold:%
\begin{equation}
\frac{\Psi _{\tau _{+}}(h_{1},...,h_{n}=1,...,h_{1})}{\Psi _{\tau
_{+}}(h_{1},...,h_{n}=0,...,h_{1})}=\frac{\tau _{+}(\zeta _{a}^{(0)})}{\text{%
\textsc{d}}_{+}(\zeta _{a}^{(0)})},
\end{equation}%
for any$\,n\in \{1,...,\mathsf{N}\}$ with $h_{m\neq n}\in \{0,1\}$. This
fixes the factorized form given in $\left( \ref{ADMFKeigenT-r-D+}\right) $-$%
\left( \ref{ADMFKt-Q-relation+}\right) $ for the $\mathcal{T}_{+}$-eigenstate $%
|\tau _{+}\rangle $ and implies the simplicity of the $\mathcal{T}_{+}$%
-spectrum. Taken a $\tau _{+}(\lambda )$ solution of (\ref{ADMFKI-Functional-eq})
in the class of function (\ref{ADMFKset-t}) and constructed $|\tau _{+}\rangle $
by $\left( \ref{ADMFKeigenT-r-D+}\right) $-$\left( \ref{ADMFKt-Q-relation+}\right) $,
then the proof that $\tau _{+}(\lambda )\in \Sigma _{\mathcal{T}_{+}}$\ and $%
|\tau _{+}\rangle $ is the corresponding $\mathcal{T}_{+}$-eigenstate can be
given following the same steps presented in Theorem \ref{ADMFKC:T-eigenstates-}.
Concerning the left $\mathcal{T}_{+}$-eigenstates the construction is done
as above we have just to remark that in this case the matrix elements:%
\begin{equation}
\langle \tau _{+}|\mathcal{T}_{+}(\zeta _{n}^{(h_{n})})|+,h_{1},...,h_{%
\mathsf{N}}\rangle ,
\end{equation}%
can be computed by using the right $\mathcal{B}_{+}$-representation:%
\begin{equation}
\tau _{+}(\zeta _{n}^{(h_{n})})\bar{\Psi}_{\tau _{+}}(\text{\textbf{h}})\,=%
\text{\textsc{a}}_{+}(\zeta _{n}^{(h_{n})})\bar{\Psi}_{\tau _{+}}(\mathsf{T}%
_{n}^{+}(\text{\textbf{h}}))+\text{\textsc{a}}_{+}(-\zeta _{n}^{(h_{n})})%
\bar{\Psi}_{\tau _{+}}(\mathsf{T}_{n}^{-}(\text{\textbf{h}})),\text{ \ \ }%
\forall n\in \{1,...,\mathsf{N}\}
\end{equation}%
where:%
\begin{equation}
\bar{\Psi}_{\tau _{+}}(\text{\textbf{h}})\equiv \langle \tau
_{+}|-,h_{1},...,h_{\mathsf{N}}\rangle ,
\end{equation}%
and the coefficient \textsc{a}$_{+}(\lambda )$\ reads:%
\begin{equation}
\text{\textsc{a}}_{+}(\lambda )\equiv \mathsf{a}_{-}(\lambda )\mathsf{A}%
_{+}(\lambda ).
\end{equation}
\end{proof}

\section{Scalar Products}

The presentation will be done simultaneously for $\mathcal{B}_{\epsilon }$%
-SOV-representations with $\epsilon =+,-$.

\begin{proposition}
Let $\langle \alpha _{\epsilon }|$ an
arbitrary covector and let $|\beta _{\epsilon }\rangle $ be an arbitrary vector of separate forms:
\begin{align}
\langle \alpha _{\epsilon }|& =\sum_{h_{1},...,h_{\mathsf{N}%
}=0}^{1}\prod_{a=1}^{\mathsf{N}}\alpha _{\epsilon ,a}(\zeta
_{a}^{(h_{a})})\prod_{1\leq b<a\leq \mathsf{N}}(\eta _{a}^{(h_{a})}-\eta
_{b}^{(h_{b})})\langle \epsilon ,h_{1},...,h_{\mathsf{N}}|,
\label{ADMFKFact-left-SOV} \\
|\beta _{\epsilon }\rangle & =\sum_{h_{1},...,h_{\mathsf{N}%
}=0}^{1}\prod_{a=1}^{\mathsf{N}}\beta _{\epsilon ,a}(\zeta
_{a}^{(h_{a})})\prod_{1\leq b<a\leq \mathsf{N}}(\eta _{a}^{(h_{a})}-\eta
_{b}^{(h_{b})})|\epsilon ,h_{1},...,h_{\mathsf{N}}\rangle ,
\label{ADMFKFact-right-SOV}
\end{align}%
in the $\mathcal{B}_{\epsilon }$-eigenbasis, then the action of $\langle
\alpha _{\epsilon }|$ on $|\beta _{\epsilon }\rangle $ reads:%
\begin{equation}
\langle \alpha _{\epsilon }|\beta _{\epsilon }\rangle =\det_{\mathsf{N}}||%
\mathcal{M}_{a,b}^{\left( \alpha _{\epsilon },\beta _{\epsilon }\right) }||%
\text{ \ \ with \ }\mathcal{M}_{a,b}^{\left( \alpha _{\epsilon },\beta
_{\epsilon }\right) }\equiv \sum_{h=0}^{1}\alpha _{\epsilon ,a}(\zeta
_{a}^{(h)})\beta _{\epsilon ,a}(\zeta _{a}^{(h)})(\eta _{a}^{(h)})^{(b-1)}.
\label{ADMFKScalar-p1}
\end{equation}%
Moreover, if $\tau _{\epsilon }(\lambda )\neq \tau _{\epsilon }^{\prime }(\lambda )\in
\Sigma _{\mathcal{T}_{\epsilon }}$, the action of the $\mathcal{T}%
_{\epsilon }$-eigencovector $\langle \tau _{\epsilon }|$\ on the $\mathcal{T}%
_{\epsilon }$-eigenvector $|\tau _{\epsilon }^{\prime }\rangle $\ is zero,
and in particular\ it holds:%
\begin{equation}
\sum_{b=1}^{\mathsf{N}}\mathcal{M}_{a,b}^{\left( \tau _{\epsilon },\tau
_{\epsilon }^{\prime }\right) }c_{b}^{\left( \tau _{\epsilon },\tau
_{\epsilon }^{\prime }\right) }=0\text{ \ \ \ \ }\forall a\in \{1,...,%
\mathsf{N}\},  \label{ADMFKzero-eigenvector}
\end{equation}%
where the $c_{b}^{\left( \tau _{\epsilon },\tau _{\epsilon }^{\prime
}\right) }$ are defined by:%
\begin{equation}
\tau _{\epsilon }(\lambda )-\tau _{\epsilon }^{\prime }(\lambda )\equiv
\sinh (2\lambda -\eta )\sinh (2\lambda +\eta )\sum_{b=1}^{\mathsf{N}%
}c_{b}^{\left( \tau _{\epsilon },\tau _{\epsilon }^{\prime }\right) }(\cosh
2\lambda )^{b-1}.
\end{equation}
\end{proposition}

\begin{proof}
The formula $\left( \ref{ADMFKM_jj}\right) $ and the SOV-decomposition of the
states $\langle \alpha _{\epsilon }|$ and $|\beta _{\epsilon }\rangle$ imples:%
\begin{equation}
\langle \alpha _{\epsilon }|\beta _{\epsilon }\rangle =\sum_{h_{1},...,h_{%
\mathsf{N}}=0}^{1}V(\eta _{1}^{(h_{1})},...,\eta _{\mathsf{N}}^{(h_{\mathsf{N%
}})})\prod_{a=1}^{\mathsf{N}}\alpha _{\epsilon,a}(\zeta _{a}^{(h_{a})})\beta
_{\epsilon,a}(\zeta _{a}^{(h_{a})}),
\end{equation}%
where $V(x_{1},...,x_{\mathsf{N}})\equiv \prod_{1\leq b<a\leq \mathsf{N}%
}(x_{a}-x_{b})$ is the Vandermonde determinant which for the multilinearity
of the determinant implies $\left( \ref{ADMFKScalar-p1}\right) $.

The $\mathcal{T}_{\epsilon }$-eigenstates $\langle \tau _{\epsilon }|$ and $%
|\tau _{\epsilon }^{\prime }\rangle $ are left and right separate states
then the action $\langle \tau _{\epsilon }|\tau _{\epsilon }^{\prime
}\rangle $ is also given by $\left( \ref{ADMFKScalar-p1}\right) $ and to prove $%
\langle \tau _{\epsilon }|\tau _{\epsilon }^{\prime }\rangle =0$ for $\tau
_{\epsilon }(\lambda )\neq \tau _{\epsilon }^{\prime }(\lambda )\in \Sigma _{%
\mathcal{T}_{\epsilon }}$ we have just to prove $\left( \ref{ADMFKzero-eigenvector}\right) $. It is simple to remark that:%
\begin{equation}
\sum_{b=1}^{\mathsf{N}}\mathcal{M} _{a,b}^{\left( \tau _{\epsilon },\tau _{\epsilon
}^{\prime }\right) }c_{b}^{\left( \tau _{\epsilon },\tau _{\epsilon
}^{\prime }\right) }=\sum_{h=0}^{1}\frac{Q_{\tau _{\epsilon }^{\prime
}}(\zeta _{a}^{(h)})\bar{Q}_{\tau _{\epsilon }}(\zeta _{a}^{(h)})(\tau
_{\epsilon }(\zeta _{a}^{(h)})-\tau _{\epsilon }^{\prime }(\zeta _{a}^{(h)}))%
}{\sinh (2\zeta _{a}^{(h)}-\eta )\sinh (2\zeta _{a}^{(h)}+\eta )},
\label{ADMFKzero-eigenvector-1}
\end{equation}%
so in the case $\epsilon =-$, we can use the equations (\ref{ADMFKt-Q-relation})
and (\ref{ADMFKt-Qbar-relation}) to rewrite:%
\begin{align}
Q_{\tau _{-}^{\prime }}(\zeta _{a}^{(1)})\bar{Q}_{\tau _{-}}(\zeta
_{a}^{(1)})(\tau _{-}(\zeta _{a}^{(1)})-\tau _{-}^{\prime }(\zeta
_{a}^{(1)}))& =k_{a}\alpha _{a}\text{\textsc{a}}_{-}(-\zeta
_{a}^{(0)})Q_{\tau _{-}^{\prime }}(\zeta _{a}^{(1)})\bar{Q}_{\tau
_{-}}(\zeta _{a}^{(0)})  \notag \\
& -\text{\textsc{a}}_{-}(\zeta _{a}^{(1)})Q_{\tau _{-}^{\prime }}(\zeta
_{a}^{(0)})\bar{Q}_{\tau _{-}}(\zeta _{a}^{(1)}),
\end{align}%
and%
\begin{align}
Q_{\tau _{-}^{\prime }}(\zeta _{a}^{(0)})\bar{Q}_{\tau _{-}}(\zeta
_{a}^{(0)})(\tau _{-}(\zeta _{a}^{(0)})-\tau _{-}^{\prime }(\zeta
_{a}^{(0)}))& =\left( k_{a}\alpha _{a}\right) ^{-1}[\text{\textsc{a}}%
_{-}(\zeta _{a}^{(1)})Q_{\tau _{-}^{\prime }}(\zeta _{a}^{(0)})\bar{Q}_{\tau
_{-}}(\zeta _{a}^{(1)})  \notag \\
& -k_{a}\alpha _{a}\text{\textsc{a}}_{-}(-\zeta _{a}^{(0)})Q_{\tau
_{-}^{\prime }}(\zeta _{a}^{(1)})\bar{Q}_{\tau _{-}}(\zeta _{a}^{(0)})].
\end{align}%
Then by substituting them in (\ref{ADMFKzero-eigenvector-1}) we get (\ref{ADMFKzero-eigenvector}). Similarly, in the case $\epsilon =+$, we can use the
equations (\ref{ADMFKt-Q-relation+}) and (\ref{ADMFKt-Qbar-relation+}) to rewrite:%
\begin{align}
Q_{\tau _{+}^{\prime }}(\zeta _{a}^{(1)})\bar{Q}_{\tau _{+}}(\zeta
_{a}^{(1)})(\tau _{+}(\zeta _{a}^{(1)})-\tau _{+}^{\prime }(\zeta
_{a}^{(1)}))& =\left( k_{a}^{(+)}\alpha _{a}^{(+)}\right) ^{-1}[\text{%
\textsc{d}}_{+}(\zeta _{a}^{(0)})Q_{\tau _{+}^{\prime }}(\zeta _{a}^{(1)})%
\bar{Q}_{\tau _{+}}(\zeta _{a}^{(0)})  \notag \\
& -k_{a}^{(+)}\alpha _{a}^{(+)}\text{\textsc{d}}_{+}(-\zeta
_{a}^{(1)})Q_{\tau _{+}^{\prime }}(\zeta _{a}^{(0)})\bar{Q}_{\tau
_{+}}(\zeta _{a}^{(1)})],
\end{align}%
and%
\begin{align}
Q_{\tau _{+}^{\prime }}(\zeta _{a}^{(0)})\bar{Q}_{\tau _{+}}(\zeta
_{a}^{(0)})(\tau _{+}(\zeta _{a}^{(0)})-\tau _{+}^{\prime }(\zeta
_{a}^{(0)}))& =k_{a}^{(+)}\alpha _{a}^{(+)}\text{\textsc{d}}_{+}(-\zeta
_{a}^{(1)})Q_{\tau _{+}^{\prime }}(\zeta _{a}^{(0)})\bar{Q}_{\tau
_{+}}(\zeta _{a}^{(1)})  \notag \\
& -\text{\textsc{d}}_{+}(\zeta _{a}^{(0)})Q_{\tau _{+}^{\prime }}(\zeta
_{a}^{(1)})\bar{Q}_{\tau _{+}}(\zeta _{a}^{(0)}).
\end{align}%
Then by substituting them in (\ref{ADMFKzero-eigenvector-1}) we get (\ref{ADMFKzero-eigenvector}).
\end{proof}

It is worth remarking that the vector $\left( \langle \epsilon ,\alpha
|\right) ^{\dagger }\in \mathcal{R}_{\mathsf{N}}$ is of separate form in the 
$\mathsf{C}_{\epsilon }$-eigenbasis thanks to the Hermitian conjugation
properties of the Yang-Baxter generators. Then the previous result describes
also the scalar products for these states. The determinant formulae obtained
here can then be considered as the SOV analogous of the Slavnov's scalar
product formula \cite{ADMFKSlav89,ADMFKSlav97}-\cite{ADMFKKitMT99} which holds in the
framework of the algebraic Bethe ansatz.

\section{Boundary reconstructions of strings of local operators}

Here, we present identities between couples of (boundary) generators of the
reflection algebra and (bulk) generators of the Yang-Baxter algebra which
can be used to reconstruct any local operator in terms of the boundary
operators. This program is explained and developed for simple strings of
local operators likes:%
\begin{equation}
\sigma _{1}^{\pm }\cdots \sigma _{n}^{\pm },\text{ \ \ \ }\sigma _{n}^{\pm
}\cdots \sigma _{\mathsf{N}}^{\pm }.
\end{equation}

\subsection{Mixed bulk and $\mathcal{U}_{\pm }$-boundary reconstructions}

Let us start recalling the bulk reconstruction formulae:

\begin{proposition}[\protect\cite{ADMFKKitMT99}]
\label{ADMFKRecBulk}Let $x_{n}\in \,$End$($R$_{n})$ be the generic local operator in the local quantum
space R$_{n}$, then it admits the following
reconstruction in terms of the generators of the Yang-Baxter algebra: 
\begin{align}
x_{n}& \equiv \prod_{a=1}^{n-1}T(\zeta _{a}^{(1)})\frac{T(\zeta _{n}^{(1)})}{%
\det_{q}M(\xi _{n})}tr_{0}(M_{0}(\zeta _{n}^{(0)})\sigma
_{0}^{y}x_{0}^{t_{0}}\sigma _{0}^{y})\prod_{a=1}^{n-1}T^{-1}(\zeta
_{a}^{(1)}) \\
& =\prod_{a=n+1}^{\mathsf{N}}T^{-1}(\zeta _{a}^{(1)})tr_{0}(M_{0}(\zeta
_{n}^{(0)})\sigma _{0}^{y}x_{0}^{t_{0}}\sigma _{0}^{y})\frac{T(\zeta
_{n}^{(1)})}{\det_{q}M(\xi _{n})}\prod_{a=n+1}^{\mathsf{N}}T(\zeta
_{a}^{(1)}) \\
& =\prod_{a=1}^{n-1}T(\zeta _{a}^{(1)})tr_{0}(M_{0}(\zeta _{n}^{(1)})x_{0})%
\frac{T(\zeta _{n}^{(0)})}{\det_{q}M(\xi _{n})}\prod_{a=1}^{n-1}T^{-1}(\zeta
_{a}^{(1)}) \\
& =\prod_{a=n+1}^{\mathsf{N}}T^{-1}(\zeta _{a}^{(1)})\frac{T(\zeta
_{n}^{(0)})}{\det_{q}M(\xi _{n})}tr_{0}(M_{0}(\zeta
_{n}^{(1)})x_{0})\prod_{a=n+1}^{\mathsf{N}}T(\zeta _{a}^{(1)}),
\end{align}%
where we have used the identities:%
\begin{equation}
T(\zeta _{a}^{(1)})T(\zeta _{n}^{(0)})=\det_{q}M(\xi _{n}).
\end{equation}
\end{proposition}

In \cite{ADMFKWang00} a reconstruction of local operators which use the
reconstruction of the propagator by the (bulk) transfer matrix $T(\zeta
_{a}^{(h_{a})})$ of the Yang-Baxter algebra and the elements of the
(boundary) reflection algebra $\mathcal{U}_{+}\left( \lambda \right) $ has
been derived. In this subsection we reproduce this result and provide other
three equivalent reconstructions:

\begin{proposition}
Let $x_{n}\in \,$End$($R$_{n})$ be the generic local operator in the local quantum
space R$_{n}$, then
it admits the following reconstructions in terms of the generators of the
(boundary) reflection algebra $\mathcal{U}_{+}\left( \lambda \right) $:%
\begin{align}
x_{n}& =\prod_{a=1}^{n-1}T(\zeta _{a}^{(1)})tr_{0}(\mathcal{U}_{+}(\zeta
_{n}^{(1)})x_{0})\frac{\mathcal{\bar{T}}_{+}(\zeta _{n}^{(0)})}{\det_{q}%
\mathcal{\bar{U}}_{+}(\xi _{n})}\prod_{a=1}^{n-1}T^{-1}(\zeta _{a}^{(1)})
\label{ADMFKRec-3} \\
& =\prod_{a=1}^{n-1}T(\zeta _{a}^{(1)})\frac{\mathcal{\bar{T}}_{+}(\zeta
_{n}^{(1)})}{\det_{q}\mathcal{\bar{U}}_{+}(\xi _{n})}tr_{0}(\mathcal{U}%
_{+}(-\zeta _{n}^{(0)})x_{0})\prod_{a=1}^{n-1}T^{-1}(\zeta _{a}^{(1)}),
\label{ADMFKRec-4}
\end{align}%
and the following ones in terms of the generators of the (boundary)
reflection algebra $\mathcal{U}_{-}\left( \lambda \right) $:%
\begin{align}
x_{n}& =\prod_{a=n+1}^{\mathsf{N}}T^{-1}(\zeta _{a}^{(1)})tr_{0}(\mathcal{U}%
_{-}(\zeta _{n}^{(0)})\sigma _{0}^{y}x_{0}^{t_{0}}\sigma _{0}^{y})\frac{%
\mathcal{\bar{T}}_{-}(\zeta _{n}^{(1)})}{\det_{q}\mathcal{\bar{U}}_{-}(\xi
_{n})}\prod_{a=n+1}^{\mathsf{N}}T(\zeta _{a}^{(1)})  \label{ADMFKRec-1} \\
& =\prod_{a=n+1}^{\mathsf{N}}T^{-1}(\zeta _{a}^{(1)})\frac{\mathcal{\bar{T}}%
_{-}(\zeta _{n}^{(0)})}{\det_{q}\mathcal{\bar{U}}_{-}(\xi _{n})}tr_{0}(%
\mathcal{U}_{-}(-\zeta _{n}^{(1)})\sigma _{0}^{y}x_{0}^{t_{0}}\sigma
_{0}^{y})\prod_{a=n+1}^{\mathsf{N}}T(\zeta _{a}^{(1)}).  \label{ADMFKRec-2}
\end{align}%
Here, we have denoted\footnote{%
Note that\begin{eqnarray*}
\mathcal{\bar{T}}_{-}(\lambda ) &\equiv &tr_{0}(\mathcal{U}_{-}(\lambda
)K_{+}(\lambda )|_{\zeta _{+}=i\pi /2})=\cosh (\lambda +\eta /2)(\mathcal{A}%
_{-}(\lambda )+\mathcal{D}_{-}(\lambda )), \\
\mathcal{\bar{T}}_{+}(\lambda ) &\equiv &tr_{0}(\mathcal{U}_{+}(\lambda
)K_{-}(\lambda )|_{\zeta _{-}=i\pi /2})=\cosh (\lambda -\eta /2)(\mathcal{A}%
_{+}(\lambda )+\mathcal{D}_{+}(\lambda )).
\end{eqnarray*}%
}:%
\begin{equation}
\mathcal{\bar{T}}_{\pm }(\lambda )=\bar{a}_{\mp }(\lambda )\mathcal{A}_{\pm
}(\lambda )+\bar{a}_{\mp }(-\lambda )\mathcal{A}_{\pm }(-\lambda ),
\end{equation}%
and%
\begin{eqnarray}
\bar{a}_{\pm }(\lambda ) &\equiv &\mathsf{a}_{\pm }(\lambda )|_{\zeta
_{+}=i\pi /2}=\mathsf{d}_{\pm }(\lambda )|_{\zeta _{+}=i\pi /2}=\cosh
(\lambda \mp \eta /2)\frac{\sinh (2\lambda \pm \eta )}{\sinh 2\lambda }, \\
\det_{q}\mathcal{\bar{U}}_{-}(\lambda ) &\equiv &\text{\textsc{a}}%
_{-}|_{\zeta _{+}=i\pi /2}(\lambda +\eta /2)\text{\textsc{a}}_{-}|_{\zeta
_{+}=i\pi /2}(-\lambda +\eta /2), \\
\det_{q}\mathcal{\bar{U}}_{+}(\lambda ) &\equiv &\text{\textsc{d}}%
_{+}|_{\zeta _{-}=i\pi /2}(-\lambda -\eta /2)\text{\textsc{d}}_{+}|_{\zeta
_{-}=i\pi /2}(\lambda -\eta /2).
\end{eqnarray}
\end{proposition}

\begin{proof}
The first reconstruction in terms of $\mathcal{U}_{+}\left( \lambda \right) $
is the result proven in Proposition 1 of \cite{ADMFKWang00}; similarly it is
possible to prove our second reconstruction in terms of $\mathcal{U}%
_{+}\left( \lambda \right) $. Let us prove here the reconstructions in terms
of $\mathcal{U}_{-}\left( \lambda \right) $; by definition of propagator
operator it holds:%
\begin{align}
& \prod_{a=1}^{n}T(\zeta _{a}^{(1)})tr_{0}(\mathcal{U}_{-}(\lambda )\sigma
_{0}^{y}x_{0}^{t_{0}}\sigma _{0}^{y})\left. =\right. (-1)^{\mathsf{N}%
}tr_{0}(L_{0,n}(\lambda )\ldots L_{01}(\lambda )L_{0\mathsf{N}}(\lambda
)\ldots L_{0n+1}(\lambda )K_{-}(\lambda )\sigma _{0}^{y}  \notag \\
& \hspace{5.5cm}L_{0,n+1}^{t_{0}}(-\lambda )\ldots L_{0\mathsf{N}%
}^{t_{0}}(-\lambda )L_{01}^{t_{0}}(-\lambda )\ldots L_{0n}^{t_{0}}(-\lambda
)x_{0}^{t_{0}}\sigma _{0}^{y})\prod_{a=1}^{n}T(\zeta _{a}^{(1)})
\label{ADMFKRec-a}
\end{align}%
where $L_{0a}(\lambda )\equiv R_{0a}(\lambda -\xi _{a}-\eta /2)$ for any $%
a\in \{1,...,\mathsf{N}\}$. Now using the identities: 
\begin{equation}
L_{0n}(\zeta _{n}^{(1)})=P_{0n}\sinh \eta ,\text{ \ }L_{0n}(\zeta
_{n}^{(0)})=-\sinh \eta \sigma _{0}^{y}P_{0n}^{t_{0}}\sigma _{0}^{y},
\end{equation}%
where $P_{0n}$ is the permutation operator between the 2-dimensional spaces R%
$_{0}$ and R$_{n}$. Then, the r.h.s of (\ref{ADMFKRec-a}) computed in $\lambda
=-\zeta _{n}^{(1)}$ reads:%
\begin{equation}
tr_{0}(\mathcal{U}_{-}(-\zeta _{n}^{(1)}))x_{n}\prod_{a=1}^{n}T(\zeta
_{a}^{(1)}),
\end{equation}%
being:%
\begin{equation}
L_{0n}^{t_{0}}(\zeta _{n}^{(1)})x_{0}^{t_{0}}\sigma _{0}^{y}=\sinh \eta 
\left[ x_{0}P_{0n}\right] ^{t_{0}}\sigma _{0}^{y}=\sinh \eta
P_{0n}^{t_{0}}\sigma _{0}^{y}x_{n}.
\end{equation}%
Similarly, the r.h.s of (\ref{ADMFKRec-a}) computed in $\lambda =-\zeta
_{n}^{(1)} $ reads:%
\begin{equation}
x_{n}tr_{0}(\mathcal{U}_{-}(-\zeta _{n}^{(1)}))\prod_{a=1}^{n}T(\zeta
_{a}^{(1)}),
\end{equation}%
being:%
\begin{equation}
x_{0}^{t_{0}}\sigma _{0}^{y}L_{0n}(\zeta _{n}^{(0)})=-\sinh \eta \left[
P_{0n}x_{0}\right] ^{t_{0}}\sigma _{0}^{y}=x_{n}\sigma _{0}^{y}L_{0n}(\zeta
_{n}^{(0)}).
\end{equation}%
By using these formulae the reconstructions (\ref{ADMFKRec-1}) and (\ref{ADMFKRec-2})
simply follows.
\end{proof}

The previous proposition naturally implies the following:

\begin{corollary}
The following annihilations identities hold:\newline
\textsf{I) }For the generators of the reflection algebra $\mathcal{U}%
_{-}\left( \lambda \right) $:%
\begin{eqnarray}
\mathcal{A}_{-}(\zeta _{n}^{(0)})\mathcal{C}_{-}(\pm \zeta _{n}^{(1)}) &=&%
\mathcal{A}_{-}(\zeta _{n}^{(0)})\mathcal{D}_{-}(-\zeta _{n}^{(1)})=0, \\
\mathcal{A}_{-}(-\zeta _{n}^{(1)})\mathcal{C}_{-}(\pm \zeta _{n}^{(0)}) &=&%
\mathcal{A}_{-}(-\zeta _{n}^{(1)})\mathcal{D}_{-}(\zeta _{n}^{(0)})=0,
\end{eqnarray}%
\begin{eqnarray}
\mathcal{D}_{-}(\zeta _{n}^{(0)})\mathcal{B}_{-}(\pm \zeta _{n}^{(1)}) &=&%
\mathcal{D}_{-}(\zeta _{n}^{(0)})\mathcal{A}_{-}(-\zeta _{n}^{(1)})=0, \\
\mathcal{D}_{-}(-\zeta _{n}^{(1)})\mathcal{B}_{-}(\pm \zeta _{n}^{(0)}) &=&%
\mathcal{D}_{-}(-\zeta _{n}^{(1)})\mathcal{A}_{-}(\zeta _{n}^{(0)})=0,
\end{eqnarray}%
\begin{equation}
\mathcal{B}_{-}(\pm \zeta _{n}^{(0)})\mathcal{B}_{-}(\pm \zeta _{n}^{(1)})=%
\mathcal{B}_{-}(\pm \zeta _{n}^{(0)})\mathcal{A}_{-}(-\zeta _{n}^{(1)})=%
\mathcal{B}_{-}(\pm \zeta _{n}^{(1)})\mathcal{A}_{-}(\zeta _{n}^{(0)})=0,
\label{ADMFKAnnih-BB-BA_-}
\end{equation}%
and%
\begin{equation}
\mathcal{C}_{-}(\pm \zeta _{n}^{(0)})\mathcal{C}_{-}(\pm \zeta _{n}^{(1)})=%
\mathcal{C}_{-}(\pm \zeta _{n}^{(0)})\mathcal{D}_{-}(-\zeta _{n}^{(1)})=%
\mathcal{C}_{-}(\pm \zeta _{n}^{(1)})\mathcal{D}_{-}(\zeta _{n}^{(0)})=0.
\end{equation}%
\newline
\textsf{II) }For the generators of the reflection algebra $\mathcal{U}%
_{+}\left( \lambda \right) $:%
\begin{eqnarray}
\mathcal{A}_{+}(-\zeta _{n}^{(0)})\mathcal{B}_{+}(\pm \zeta _{n}^{(0)}) &=&%
\mathcal{A}_{+}(-\zeta _{n}^{(0)})\mathcal{D}_{+}(\zeta _{n}^{(1)})=0, \\
\mathcal{A}_{+}(\zeta _{n}^{(1)})\mathcal{B}_{+}(\pm \zeta _{n}^{(0)}) &=&%
\mathcal{A}_{+}(\zeta _{n}^{(1)})\mathcal{D}_{+}(-\zeta _{n}^{(0)})=0,
\end{eqnarray}%
\begin{eqnarray}
\mathcal{D}_{+}(-\zeta _{n}^{(0)})\mathcal{C}_{+}(\pm \zeta _{n}^{(1)}) &=&%
\mathcal{D}_{+}(-\zeta _{n}^{(0)})\mathcal{A}_{+}(\zeta _{n}^{(1)})=0, \\
\mathcal{D}_{+}(\zeta _{n}^{(1)})\mathcal{C}_{+}(\pm \zeta _{n}^{(0)}) &=&%
\mathcal{D}_{+}(\zeta _{n}^{(1)})\mathcal{A}_{+}(-\zeta _{n}^{(0)})=0,
\end{eqnarray}%
\begin{equation}
\mathcal{B}_{+}(\pm \zeta _{n}^{(0)})\mathcal{B}_{+}(\pm \zeta _{n}^{(1)})=%
\mathcal{B}_{+}(\pm \zeta _{n}^{(0)})\mathcal{D}_{+}(\zeta _{n}^{(1)})=%
\mathcal{B}_{+}(\pm \zeta _{n}^{(1)})\mathcal{D}_{+}(-\zeta _{n}^{(0)})=0,
\end{equation}%
and%
\begin{equation}
\mathcal{C}_{+}(\pm \zeta _{n}^{(0)})\mathcal{C}_{+}(\pm \zeta _{n}^{(1)})=%
\mathcal{C}_{+}(\pm \zeta _{n}^{(0)})\mathcal{A}_{+}(\zeta _{n}^{(1)})=%
\mathcal{C}_{+}(\pm \zeta _{n}^{(1)})\mathcal{A}_{+}(-\zeta _{n}^{(0)})=0.
\end{equation}
\end{corollary}

\begin{proof}
The proof of these annihilation identities can be done along the same line
used for the bulk case. Let us sketch it; by using the reconstruction
formulae of the previous proposition and the identities:%
\begin{equation}
\mathcal{\bar{T}}_{\pm }(\zeta _{n}^{(1)})\mathcal{\bar{T}}_{\pm }(\zeta
_{n}^{(0)})=\det_{q}\mathcal{\bar{U}}_{\pm }(\xi _{n}),
\end{equation}
it also holds:%
\begin{eqnarray}
&&\prod_{a=1}^{n-1}T^{-1}(\zeta
_{a}^{(1)})x_{n}^{(3)}\prod_{a=1}^{n-1}T(\zeta _{a}^{(1)})\left. =\right. 
\frac{tr_{0}(\mathcal{U}_{+}(\zeta _{n}^{(1)})x_{0}^{(1)})tr_{0}(\mathcal{U}%
_{+}(-\zeta _{n}^{(0)})x_{0}^{(2)})}{\det_{q}\mathcal{\bar{U}}_{+}(\xi _{n})}
\\
&&\text{ \ \ \ \ \ \ \ \ \ \ \ \ \ \ \ \ \ }\left. =\right. \frac{\mathcal{%
\bar{T}}_{+}(\zeta _{n}^{(1)})tr_{0}(\mathcal{U}_{+}(-\zeta
_{n}^{(0)})x_{0}^{(1)})tr_{0}(\mathcal{U}_{+}(\zeta _{n}^{(1)})x_{0}^{(2)})%
\mathcal{\bar{T}}_{+}(\zeta _{n}^{(0)})}{\left( \det_{q}\mathcal{\bar{U}}%
_{+}(\xi _{n})\right) ^{2}},
\end{eqnarray}%
and:%
\begin{align}
& \prod_{a=n+1}^{\mathsf{N}}T(\zeta _{a}^{(1)})y_{n}^{(3)}\prod_{a=n+1}^{%
\mathsf{N}}T^{-1}(\zeta _{a}^{(1)})\left. =\right. \frac{tr_{0}(\mathcal{U}%
_{-}(\zeta _{n}^{(0)})y_{0}^{(1)})tr_{0}(\mathcal{U}_{-}(-\zeta
_{n}^{(1)})y_{0}^{(2)})}{\det_{q}\mathcal{\bar{U}}_{-}(\xi _{n})} \\
& \text{ \ \ \ \ \ \ \ \ \ \ \ \ \ \ \ \ \ \ \ }\left. =\right. \frac{%
\mathcal{\bar{T}}_{-}(\zeta _{n}^{(0)})tr_{0}(\mathcal{U}_{-}(-\zeta
_{n}^{(1)})y_{0}^{(1)})tr_{0}(\mathcal{U}_{-}(\zeta _{n}^{(0)})y_{0}^{(2)})%
\mathcal{\bar{T}}_{-}(\zeta _{n}^{(0)})}{\left( \det_{q}\mathcal{\bar{U}}%
_{-}(\xi _{n})\right) ^{2}}.
\end{align}%
where we have defined:%
\begin{equation}
x_{n}^{(3)}\equiv x_{n}^{(1)}x_{n}^{(2)},\text{ \ \ \ \ \ \ \ }%
y_{n}^{(3)}\equiv \sigma _{0}^{y}\left( y_{n}^{(2)}y_{n}^{(1)}\right)
^{t_{0}}\sigma _{0}^{y},
\end{equation}
we can use now these formulae to derive all the annihilation formulae.
\end{proof}

\subsection{$\mathcal{U}_{\pm }$-boundary reconstruction}

Let us remark that the bulk reconstruction formulae of Proposition \ref{ADMFKRecBulk}, the corresponding annihilation identities\footnote{%
See Lemma 5.1 of \cite{ADMFKKKMNST07}.} and the Yang-Baxter commutation relation
implies that the generic string of local operators of the form:%
\begin{equation}
\prod_{a=1}^{n}x_{a}  \label{ADMFKP-locals}
\end{equation}%
can indeed be represented as linear combinations of the following strings of
bulk operators:%
\begin{equation}
\left( tr_{0}(M_{0}(\zeta _{1}^{(1)})y_{0}^{(1)})\frac{T(\zeta _{1}^{(0)})}{%
\det_{q}M(\xi _{1})}\right) \left( tr_{0}(M_{0}(\zeta _{2}^{(1)})y_{0}^{(2)})\frac{%
T(\zeta _{2}^{(0)})}{\det_{q}M(\xi _{2})}\right) \cdots \left(
tr_{0}(M_{0}(\zeta _{n}^{(1)})y_{0}^{(n)})\frac{T(\zeta _{n}^{(0)})}{\det_{q}M(\xi
_{n})}\right) ,
\end{equation}%
then the following identities between bulk and $\mathcal{U}_{+}$-boundary
generators:%
\begin{equation}
tr_{0}(M_{0}(\zeta _{n}^{(1)})y_{0})\frac{T(\zeta _{n}^{(0)})}{\det_{q}M(\xi
_{n})}=tr_{0}(\mathcal{U}_{+}(\zeta _{n}^{(1)})y_{0})\frac{\mathcal{\bar{T}}%
_{+}(\zeta _{n}^{(0)})}{\det_{q}\mathcal{\bar{U}}_{+}(\xi _{n})},
\end{equation}%
imply that we can also write (\ref{ADMFKP-locals}) by a linear combinations of
the following strings of $\mathcal{U}_{+}$-boundary generators:%
\begin{equation}
\left( tr_{0}(\mathcal{U}_{+}(\zeta _{1}^{(1)})y_{0}^{(1)})\frac{\mathcal{%
\bar{T}}_{+}(\zeta _{1}^{(0)})}{\det_{q}\mathcal{\bar{U}}_{+}(\xi _{1})}%
\right) \left( tr_{0}(\mathcal{U}_{+}(\zeta _{2}^{(1)})y_{0}^{(2)})\frac{%
\mathcal{\bar{T}}_{+}(\zeta _{2}^{(0)})}{\det_{q}\mathcal{\bar{U}}_{+}(\xi
_{2})}\right) \cdots \left( tr_{0}(\mathcal{U}_{+}(\zeta
_{n}^{(1)})y_{0}^{(n)})\frac{\mathcal{\bar{T}}_{+}(\zeta _{n}^{(0)})}{%
\det_{q}\mathcal{\bar{U}}_{+}(\xi _{n})}\right) .
\end{equation}%
Similarly, the identities between bulk and $\mathcal{U}_{-}$-boundary
generators:%
\begin{equation}
tr_{0}(M_{0}(\zeta _{n}^{(0)})y_{0})\frac{T(\zeta _{n}^{(1)})}{\det_{q}M(\xi
_{n})}=tr_{0}(\mathcal{U}_{-}(\zeta _{n}^{(0)})y_{0})\frac{\mathcal{\bar{T}}%
_{-}(\zeta _{n}^{(1)})}{\det_{q}\mathcal{\bar{U}}_{-}(\xi _{n})}
\label{ADMFKId-Bulk-Boundary_-}
\end{equation}%
imply that the generic string of local operators of the form:%
\begin{equation}
\prod_{a=n}^{\mathsf{N}}x_{a}
\end{equation}%
can indeed be represented as linear combinations of the following strings of 
$\mathcal{U}_{-}$-boundary generators:%
\begin{equation}
\left( tr_{0}(\mathcal{U}_{-}(\zeta _{\mathsf{N}}^{(0)})y_{0}^{(\mathsf{N})})%
\frac{\mathcal{\bar{T}}_{-}(\zeta _{\mathsf{N}}^{(1)})}{\det_{q}\mathcal{%
\bar{U}}_{-}(\xi _{\mathsf{N}})}\right) \left( tr_{0}(\mathcal{U}_{-}(\zeta
_{\mathsf{N}-1}^{(0)})y_{0}^{(\mathsf{N}-1)})\frac{\mathcal{\bar{T}}%
_{-}(\zeta _{\mathsf{N}-1}^{(1)})}{\det_{q}\mathcal{\bar{U}}_{-}(\xi _{%
\mathsf{N}-1})}\right) \cdots \left( tr_{0}(\mathcal{U}_{-}(\zeta
_{n}^{(0)})y_{0}^{(n)})\frac{\mathcal{\bar{T}}_{-}(\zeta _{n}^{(1)})}{%
\det_{q}\mathcal{\bar{U}}_{-}(\xi _{n})}\right).
\end{equation}%
Here we show explicitly as these reconstructions work for some special
string of local operators:

\begin{proposition}
Let us consider the open XXZ spin chain with general $K_{-}$ and diagonal or
triangular $K_{+}$, then the following boundary reconstruction holds:%
\begin{eqnarray}
\sigma _{n}^{-}\cdots \sigma _{\mathsf{N}}^{-} &=&(-1)^{\mathsf{N}%
+1-n}\prod_{a=n}^{\mathsf{N}}\frac{\bar{a}_{+}(\zeta _{a}^{(1)})}{\mathsf{a}%
_{+}(\zeta _{a}^{(1)})}\prod_{n\leq a<b\leq \mathsf{N}}\frac{\sinh (\xi
_{a}+\xi _{b}-\eta )}{\sinh (\xi _{a}+\xi _{b})}  \notag \\
&&\times \mathcal{B}_{-}(\zeta _{\mathsf{N}}^{(0)})\cdots \mathcal{B}%
_{-}(\zeta _{n}^{(0)})\frac{\mathcal{T}_{-}(\zeta _{\mathsf{N}}^{(1)})}{%
\det_{q}\mathcal{\bar{U}}_{-}(\xi _{\mathsf{N}})}\cdots \frac{\mathcal{T}%
_{-}(\zeta _{n}^{(1)})}{\det_{q}\mathcal{\bar{U}}_{-}(\xi _{n})}.
\label{ADMFKsigma-rec-B-}
\end{eqnarray}%
Let us consider the open XXZ spin chain with general $K_{+}$ and diagonal or
triangular $K_{-}$, then the following boundary reconstruction holds:%
\begin{eqnarray}
\sigma _{1}^{-}\cdots \sigma _{n}^{-} &=&\prod_{a=1}^{n}\frac{\bar{d}%
_{-}(\zeta _{a}^{(0)})}{\mathsf{d}_{-}(\zeta _{a}^{(0)})}\prod_{1\leq
a<b\leq n}\frac{\sinh (\xi _{a}+\xi _{b}+\eta )}{\sinh (\xi _{a}+\xi _{b})} 
\notag \\
&&\times \mathcal{B}_{+}(\zeta _{1}^{(1)})\cdots \mathcal{B}_{+}(\zeta
_{n}^{(1)})\frac{\mathcal{T}_{+}(\zeta _{n}^{(0)})}{\det_{q}\mathcal{\bar{U}}%
_{+}(\xi _{n})}\cdots \frac{\mathcal{T}_{+}(\zeta _{1}^{(0)})}{\det_{q}%
\mathcal{\bar{U}}_{+}(\xi _{1})}.  \label{ADMFKsigma-rec-B+}
\end{eqnarray}
\end{proposition}

\begin{proof}
Let us prove explicitly the first reconstruction as for the second one can
proceed similarly. First of all by using the bulk reconstruction of
Proposition \ref{ADMFKRecBulk} one gets:%
\begin{equation}
\sigma _{n}^{-}\cdots \sigma _{\mathsf{N}}^{-}=(-1)^{\mathsf{N}+1-n}B(\zeta
_{\mathsf{N}}^{(0)})\cdots B(\zeta _{n}^{(0)})\frac{T(\zeta _{\mathsf{N}%
}^{(1)})}{\det_{q}M_{0}(\xi _{\mathsf{N}})}\cdots \frac{T(\zeta _{n}^{(1)})}{%
\det_{q}M_{0}(\xi _{n})}.
\end{equation}%
Now by using the bulk annihilation identities $B(\zeta _{\mathsf{N}%
}^{(0)})D(\zeta _{\mathsf{N}}^{(1)})=0$, we get the set of identities:%
\begin{align}
B(\zeta _{\mathsf{N}}^{(0)})\cdots B(\zeta _{n}^{(0)})T(\zeta _{\mathsf{N}%
}^{(1)})& =B(\zeta _{\mathsf{N}}^{(0)})\cdots B(\zeta _{n}^{(0)})A(\zeta _{%
\mathsf{N}}^{(1)})\notag \\
& =\frac{\sinh (\xi _{n}-\xi _{\mathsf{N}}-\eta )}{\sinh (\xi _{n}-\xi _{%
\mathsf{N}})}B(\zeta _{\mathsf{N}}^{(0)})\cdots B(\zeta _{n+1}^{(1)})A(\zeta
_{\mathsf{N}}^{(1)})B(\zeta _{n}^{(0)})\notag \\
& =\prod_{a=n}^{\mathsf{N}-1}\frac{\sinh (\xi _{a}-\xi _{\mathsf{N}}-\eta )}{%
\sinh (\xi _{a}-\xi _{\mathsf{N}})}B(\zeta _{\mathsf{N}}^{(0)})A(\zeta _{%
\mathsf{N}}^{(1)})B(\zeta _{\mathsf{N}}^{(1)})\cdots B(\zeta _{n}^{(0)}) \notag\\
& =\prod_{a=n}^{\mathsf{N}-1}\frac{\sinh (\xi _{a}-\xi _{\mathsf{N}}-\eta )}{%
\sinh (\xi _{a}-\xi _{\mathsf{N}})}B(\zeta _{\mathsf{N}}^{(0)})T(\zeta _{%
\mathsf{N}}^{(1)})B(\zeta _{\mathsf{N}}^{(1)})\cdots B(\zeta _{n}^{(0)})
\end{align}%
where to commute $B(\zeta _{n}^{(0)})$ and $A(\zeta _{\mathsf{N}}^{(1)})$ we
have used the Yang-Baxter commutation relation:%
\begin{equation}
B(\zeta _{n}^{(0)})A(\zeta _{\mathsf{N}}^{(1)})=\frac{\sinh (\xi _{n}-\xi _{%
\mathsf{N}}-\eta )}{\sinh (\xi _{n}-\xi _{\mathsf{N}})}A(\zeta _{\mathsf{N}%
}^{(1)})B(\zeta _{n}^{(0)})+\frac{\sinh \eta }{\sinh (\xi _{n}-\xi _{\mathsf{%
N}})}B(\zeta _{\mathsf{N}}^{(1)})A(\zeta _{n}^{(0)})
\end{equation}%
which with the bulk annihilation identity $B(\zeta _{n}^{(0)})B(\zeta
_{n}^{(1)})=0$ imply the above second identity, the third one is obtained by
reiterating the commutation and finally the forth one by using once again
the annihilation identity $B(\zeta _{\mathsf{N}}^{(0)})D(\zeta _{\mathsf{N}%
}^{(1)})=0$. Repeating the same procedure to commute the others transfer
matrices through the product of the $B$-operators we get:%
\begin{equation}
\sigma _{\mathsf{N}}^{-}\cdots \sigma _{n}^{-}=(-1)^{\mathsf{N}%
+1-n}\prod_{n\leq a<b\leq \mathsf{N}}\frac{\sinh (\xi _{a}-\xi _{b}-\eta )}{%
\sinh (\xi _{a}-\xi _{b})}\frac{B(\zeta _{\mathsf{N}}^{(0)})T(\zeta _{%
\mathsf{N}}^{(1)})}{\det_{q}M_{0}(\xi _{\mathsf{N}})}\cdots \frac{B(\zeta
_{n}^{(0)})T(\zeta _{n}^{(1)})}{\det_{q}M_{0}(\xi _{n})},
\end{equation}%
and then by using the boundary-bulk identities $\left( \ref{ADMFKId-Bulk-Boundary_-}\right) $ we can rewrite it in the form:%
\begin{equation}
\sigma _{\mathsf{N}}^{-}\cdots \sigma _{n}^{-}=(-1)^{\mathsf{N}%
+1-n}\prod_{n\leq a<b\leq \mathsf{N}}\frac{\sinh (\xi _{a}-\xi _{b}-\eta )}{%
\sinh (\xi _{a}-\xi _{b})}\frac{\mathcal{B}_{-}(\zeta _{\mathsf{N}}^{(0)})%
\mathcal{\bar{T}}_{-}(\zeta _{\mathsf{N}}^{(1)})}{\det_{q}\mathcal{\bar{U}}%
_{-}(\xi _{\mathsf{N}})}\cdots \frac{\mathcal{B}_{-}(\zeta _{n}^{(0)})%
\mathcal{\bar{T}}_{-}(\zeta _{n}^{(1)})}{\det_{q}\mathcal{\bar{U}}_{-}(\xi
_{n})}.  \label{ADMFKInterm-1}
\end{equation}%
We can now use the boundary annihilation identities $\left( \ref{ADMFKAnnih-BB-BA_-}\right) $ and the reflection algebra commutation relations to
move all the transfer matrices $\mathcal{\bar{T}}_{-}(\zeta _{a}^{(1)})$ to
the right and transform them into $\mathcal{T}_{-}(\zeta _{a}^{(1)})$. More
in details, from the annihilation identities and the definition of the
transfer matrices it holds:%
\begin{equation}
\mathcal{B}_{-}(\zeta _{n}^{(0)})\mathcal{\bar{T}}_{-}(\zeta _{n}^{(1)})=%
\frac{\bar{a}_{+}(\zeta _{n}^{(1)})}{\mathsf{a}_{+}(\zeta _{n}^{(1)})}%
\mathcal{B}_{-}(\zeta _{n}^{(0)})\mathcal{T}_{-}(\zeta _{n}^{(1)}),
\end{equation}%
and so the first transfer matrix on the right in $\left( \ref{ADMFKInterm-1}%
\right) $ can be rewritten in the desired form. Now let us apply the
procedure to the product $\mathcal{B}_{-}(\zeta _{n+1}^{(0)})\mathcal{\bar{T}%
}_{-}(\zeta _{n+1}^{(1)})\mathcal{B}_{-}(\zeta _{n}^{(0)})$ to rewrite it in
the desired form $\mathcal{B}_{-}(\zeta _{n+1}^{(0)})\mathcal{B}_{-}(\zeta
_{n}^{(0)})\mathcal{T}_{-}(\zeta _{n+1}^{(1)})$, the annihilation identity $%
\left( \ref{ADMFKAnnih-BB-BA_-}\right) $ implies:%
\begin{eqnarray}
\mathcal{B}_{-}(\zeta _{n+1}^{(0)})\mathcal{\bar{T}}_{-}(\zeta _{n+1}^{(1)})%
\mathcal{B}_{-}(\zeta _{n}^{(0)}) &=&\bar{a}_{+}(\zeta _{n+1}^{(1)})\mathcal{%
B}_{-}(\zeta _{n+1}^{(0)})\mathcal{A}_{-}(\zeta _{n+1}^{(1)})\mathcal{B}%
_{-}(\zeta _{n}^{(0)}) \\
&=&\bar{a}_{+}(\zeta _{n+1}^{(1)})\frac{\sinh (\xi _{n}-\xi _{n+1})}{\sinh
(\xi _{n}-\xi _{n+1}-\eta )}\frac{\sinh (\xi _{n}+\xi _{n+1}-\eta )}{\sinh
(\xi _{n}+\xi _{n+1})}  \notag \\
&&\times \mathcal{B}_{-}(\zeta _{n+1}^{(0)})\mathcal{B}_{-}(\zeta _{n}^{(0)})%
\mathcal{A}_{-}(\zeta _{n+1}^{(1)}).  \label{ADMFKInterm-2}
\end{eqnarray}%
In the second identity we have used the reflection algebra commutation
relation:%
\begin{eqnarray}
\mathcal{A}_{-}(\zeta _{n+1}^{(1)})\mathcal{B}_{-}(\zeta _{n}^{(0)}) &=&%
\frac{\sinh (\xi _{n}-\xi _{n+1})}{\sinh (\xi _{n}-\xi _{n+1}-\eta )}\frac{%
\sinh (\xi _{n}+\xi _{n+1}-\eta )}{\sinh (\xi _{n}+\xi _{n+1})}\mathcal{B}%
_{-}(\zeta _{n}^{(0)})\mathcal{A}_{-}(\zeta _{n+1}^{(1)})  \notag \\
&&-\frac{\sinh (2\xi _{n}-2\eta )\sinh \eta }{\sinh (\xi _{n}-\xi
_{n+1}-\eta )\sinh (2\xi _{n}-\eta )}\mathcal{B}_{-}(\zeta _{n+1}^{(1)})%
\mathcal{A}_{-}(\zeta _{n}^{(0)})  \notag \\
&&-\frac{\sinh \eta }{\sinh (\xi _{n}+\xi _{n+1})\sinh 2\lambda _{1}}%
\mathcal{B}_{-}(\zeta _{n+1}^{(1)})\mathcal{\tilde{D}}_{-}(\zeta _{n}^{(0)})
\end{eqnarray}%
and the fact that due to the presence of $\mathcal{B}_{-}(\zeta _{n+1}^{(0)})$ the second and third terms on the
right of this formula give zero when inserted in $\left( \ref{ADMFKInterm-1}%
\right) $. Now by using the commutativity of the $\mathcal{B}_{-}$%
-generators and once again the annihilation identities it holds:%
\begin{equation}
\mathcal{B}_{-}(\zeta _{n+1}^{(0)})\mathcal{B}_{-}(\zeta _{n}^{(0)})\mathcal{%
A}_{-}(\zeta _{n+1}^{(1)})=\frac{\mathcal{B}_{-}(\zeta _{n+1}^{(0)})\mathcal{%
B}_{-}(\zeta _{n}^{(0)})\mathcal{T}_{-}(\zeta _{n+1}^{(1)})}{\mathsf{a}%
_{+}(\zeta _{n}^{(1)})},
\end{equation}%
so that we have accomplished our task for $\mathsf{N}=n+1$ having the
product $\mathcal{T}_{-}(\zeta _{n+1}^{(1)})\mathcal{T}_{-}(\zeta
_{n}^{(1)}) $ to the right of $\left( \ref{ADMFKInterm-1}\right) $. Reiterating
this procedure for the remaining transfer matrices $\mathcal{\bar{T}}%
_{-}(\zeta _{a}^{(1)})$ for $a\in \{n+2,...,\mathsf{N}\}$ we obtain our
result.
\end{proof}

\section{Matrix elements}

\subsection{From $\mathcal{B}_{-}$-SOV representation}

Here we consider the transfer matrices (\ref{ADMFKT-}), then the
following proposition holds:

\begin{proposition}
Let $\langle \tau _{-}|$ and $|\tau _{-}^{\prime }\rangle $
be a generic couple of left and right $\mathcal{T}_{-}$-eigenstates, then we
have:%
\begin{align}
\langle \tau _{-}|\sigma _{n}^{-}\cdots \sigma _{\mathsf{N}}^{-}|\tau
_{-}^{\prime }\rangle & =\frac{-\left( \frac{\kappa _{-}e^{\tau _{-}}\sinh
\eta }{2^{\mathsf{N}}\sinh \zeta _{-}}\right) ^{\mathsf{N}-n}}{V(\eta
_{n}^{(0)},...,\eta _{\mathsf{N}}^{(0)})}\prod_{n\leq a<b\leq \mathsf{N}}%
\frac{\sinh (\xi _{a}+\xi _{b}-\eta )}{\sinh (\xi _{a}+\xi _{b})}\det_{2%
\mathsf{N}-n}||\Sigma _{a,b}^{(-,n,\tau _{-},\tau _{-}^{\prime })}||  \notag
\\
& \times \prod_{a=n}^{\mathsf{N}}\frac{\bar{a}_{+}(\zeta _{a}^{(1)})\tau
_{-}^{\prime }(\zeta _{a}^{(1)})\bar{Q}_{\tau _{-}}(\zeta _{a}^{(1)})Q_{\tau
_{-}^{\prime }}(\zeta _{a}^{(1)})\sinh 2\xi _{a}}{\mathsf{a}_{+}(\zeta
_{a}^{(1)})\det_{q}\mathcal{\bar{U}}_{-}(\xi _{a})},
\end{align}%
where $||\Sigma _{a,b}^{(-,n,\tau _{-},\tau _{-}^{\prime })}||$ is the $(2\mathsf{N}-n)\times (2\mathsf{N}-n)$ matrix of elements:%
\begin{align}
\Sigma _{a,b}^{(-,n,\tau _{-},\tau _{-}^{\prime })}& \equiv \mathcal{M}
_{a,b}^{(\tau _{-},\tau _{-}^{\prime })}\text{ \ \ \ for }a\in \{1,...,n-1\},%
\text{ }b\in \{1,...,2\mathsf{N}-n\}, \\
\Sigma _{a,b}^{(-,n,\tau _{-},\tau _{-}^{\prime })}& \equiv \left( \eta
_{a}^{(0)}\right) ^{(b-1)}\text{ \ \ \ for }a\in \{n,...,\mathsf{N}\},\text{ 
}b\in \{1,...,2\mathsf{N}-n\}, \\
\Sigma _{a,b}^{(-,n,\tau _{-},\tau _{-}^{\prime })}& \equiv \left( \eta
_{a}^{(1)}\right) ^{(b-1)}\text{ \ \ \ for }a\in \{\mathsf{N}+1,...,2\mathsf{%
N}-n\},\text{ }b\in \{1,...,2\mathsf{N}-n\}.
\end{align}
\end{proposition}

\begin{proof}
Here we use the reconstruction (\ref{ADMFKsigma-rec-B-}) for $\sigma
_{n}^{-}\cdots \sigma _{\mathsf{N}}^{-}$, then we can act with the product
of transfer matrices on the right $\mathcal{T}_{-}$-eigenstate $|\tau
_{-}^{\prime }\rangle $ and we are left with the following computation:%
\begin{equation}
\mathcal{B}_{-}(\zeta _{\mathsf{N}}^{(0)})\cdots \mathcal{B}_{-}(\zeta
_{n}^{(0)})|\tau _{-}^{\prime }\rangle .
\end{equation}%
From the decomposition of $|\tau _{-}^{\prime }\rangle $ in the $\mathcal{B}%
_{-}$-eigenstates, we get: 
\begin{align}
\mathcal{B}_{-}(\zeta _{\mathsf{N}}^{(0)})\cdots \mathcal{B}_{-}(\zeta
_{n}^{(0)})|\tau _{-}^{\prime }\rangle & =\prod_{a=n}^{\mathsf{N}}Q_{\tau
_{-}^{\prime }}(\zeta
_{a}^{(1)})\sum_{h_{1},...,h_{n-1}=0}^{1}\prod_{a=1}^{n-1}Q_{\tau
_{-}^{\prime }}(\zeta _{a}^{(h_{a})})\prod_{a=n}^{\mathsf{N}}\text{\textsc{b}%
}_{-,\{h_{1},...,h_{n-1},1,...,1\}}(\zeta _{a}^{(0)})  \notag \\
& V(\eta _{1}^{(h_{1})},...,\eta _{n-1}^{(h_{n-1})},\eta _{n}^{(1)},...,\eta
_{\mathsf{N}}^{(1)})|h_{1},...,h_{n-1},1,...,1\rangle .
\end{align}%
Let us rewrite the $\mathcal{B}_{-}$-eigenvalues in terms of the $\eta
_{a}^{(h_{a})}$:%
\begin{equation}
\text{\textsc{b}}_{-,\text{\textbf{h}}}(\lambda )\equiv \frac{\left(
-1\right) ^{\mathsf{N}}\kappa _{-}e^{\tau _{-}}\sinh (2\lambda -\eta )}{2^{%
\mathsf{N}}\sinh \zeta _{-}}\prod_{a=1}^{\mathsf{N}}\left( \cosh 2\lambda
-\eta _{a}^{(h_{a})}\right) ,
\end{equation}%
and then we have:%
\begin{align}
& \prod_{a=n}^{\mathsf{N}}\text{\textsc{b}}_{-,\{h_{1},...,h_{n-1},1,...,1%
\}}(\zeta _{a}^{(0)})V(\eta _{1}^{(h_{1})},...,\eta _{n-1}^{(h_{n-1})},\eta
_{n}^{(1)},...,\eta _{\mathsf{N}}^{(1)})  \notag \\
& =\frac{\left( \frac{\left( -1\right) ^{\mathsf{N}-n}\kappa _{-}e^{\tau
_{-}}\sinh \eta }{2^{\mathsf{N}}\sinh \zeta _{-}}\right) ^{\mathsf{N}-n}}{%
V(\eta _{n}^{(0)},...,\eta _{\mathsf{N}}^{(0)})}\prod_{a=n}^{\mathsf{N}%
}\sinh 2(\xi _{a}-\eta )  \notag \\
& \times V(\eta _{1}^{(h_{1})},...,\eta _{n-1}^{(h_{n-1})},\eta
_{n}^{(0)},...,\eta _{\mathsf{N}}^{(0)},\eta _{n}^{(1)},...,\eta _{\mathsf{N}%
}^{(1)}).
\end{align}%
Using this last formula and taking the scalar product we obtain our result.
\end{proof}

\subsection{From $\mathcal{B}_{+}$-SOV representation}

Here we consider the transfer matrices $(\ref{ADMFKT+})$, then the
following proposition holds:

\begin{proposition}
Let us consider the open XXZ with transfer matrix $\mathcal{T}_{+}(\lambda )$
$(\ref{ADMFKT+})$ and let $\langle \tau _{+}|$ and $|\tau _{+}^{\prime }\rangle $
be a generic couple of left and right $\mathcal{T}_{+}$-eigenstates, then we
have:%
\begin{align}
\langle \tau _{+}|\sigma _{1}^{-}\cdots \sigma _{n}^{-}|\tau _{+}^{\prime
}\rangle & =\frac{\left( \frac{\left( -1\right) ^{\mathsf{N}}\kappa
_{+}e^{\tau _{+}}\sinh \eta }{2^{\mathsf{N}}\sinh \zeta _{+}}\right) ^{n}}{%
V(\eta _{n}^{(1)},...,\eta _{\mathsf{N}}^{(1)})}\prod_{1\leq a<b\leq n}\frac{%
\sinh (\xi _{a}+\xi _{b}+\eta )}{\sinh (\xi _{a}+\xi _{b})}\det_{\mathsf{N}%
+n}||\Sigma _{a,b}^{(+,n,\tau _{+},\tau _{+}^{\prime })}||  \notag \\
& \times \prod_{a=1}^{n}\frac{\bar{d}_{-}(\zeta _{a}^{(0)})\tau _{+}^{\prime
}(\zeta _{a}^{(0)})\bar{Q}_{\tau _{+}}(\zeta _{a}^{(0)})Q_{\tau _{+}^{\prime
}}(\zeta _{a}^{(0)})\sinh 2\xi _{a}}{\mathsf{d}_{-}(\zeta _{a}^{(0)})\det_{q}%
\mathcal{\bar{U}}_{+}(\xi _{a})},
\end{align}%
where $||\Sigma _{a,b}^{(+.n,\tau _{+},\tau _{+}^{\prime })}||$ is the $(\mathsf{N}+n)\times (\mathsf{N}+n)$ matrix of elements:%
\begin{align}
\Sigma _{a,b}^{(+,n,\tau +,\tau _{+}^{\prime })}& \equiv \left( \eta
_{a}^{(1)}\right) ^{(b-1)}\text{ \ \ \ for }a\in \{1,...,n\},\text{ }b\in
\{1,...,\mathsf{N}+n\}, \\
\Sigma _{a,b}^{(+,n,\tau +,\tau _{+}^{\prime })}& \equiv \left( \eta
_{a-n}^{(0)}\right) ^{(b-1)}\text{ \ \ \ for }a\in \{n+1,...,2n\},\text{ }%
b\in \{1,...,\mathsf{N}+n\}, \\
\Sigma _{a,b}^{(+,n,\tau +,\tau _{+}^{\prime })}& \equiv \mathcal{M}
_{a-n,b}^{(\tau _{+},\tau _{+}^{\prime })}\text{ \ \ \ for }a\in \{2n+1,...,%
\mathsf{N}+n\},\text{ }b\in \{1,...,\mathsf{N}+n\}.
\end{align}
\end{proposition}

\begin{proof}
Here we use the reconstruction (\ref{ADMFKsigma-rec-B+}) for $\sigma
_{1}^{-}\cdots \sigma _{n}^{-}$, then we can act with the product of
transfer matrices on the right $\mathcal{T}_{+}$-eigenstate $|\tau
_{+}^{\prime }\rangle $ and we are left with the following computation:%
\begin{equation}
\mathcal{B}_{+}(\zeta _{\mathsf{N}}^{(0)})\cdots \mathcal{B}_{+}(\zeta
_{n}^{(0)})|\tau _{+}^{\prime }\rangle .
\end{equation}%
From the decomposition of $|\tau _{+}^{\prime }\rangle $ in the $\mathcal{B}%
_{+}$-eigenstates, we get: 
\begin{align}
\mathcal{B}_{+}(\zeta _{\mathsf{N}}^{(1)})\cdots \mathcal{B}_{+}(\zeta
_{n}^{(1)})|\tau _{+}^{\prime }\rangle & =\prod_{a=1}^{n}Q_{\tau
_{+}^{\prime }}(\zeta _{a}^{(1)})\sum_{h_{n+1},...,h_{\mathsf{N}%
}=0}^{1}\prod_{a=n+1}^{\mathsf{N}}Q_{\tau _{+}^{\prime }}(\zeta
_{a}^{(h_{a})})\prod_{a=1}^{n}\text{\textsc{b}}_{+,\{0,...,0,h_{n+1},...,h_{%
\mathsf{N}}\}}(\zeta _{a}^{(1)})  \notag \\
& V(\eta _{1}^{(0)},...,\eta _{n}^{(0)},\eta _{n+1}^{(h_{n+1})},...,\eta _{%
\mathsf{N}}^{(h_{\mathsf{N}})})|0,...,0,h_{n+1},...,h_{\mathsf{N}}\rangle .
\end{align}%
Let us rewrite the $\mathcal{B}_{+}$-eigenvalues in terms of the $\eta
_{a}^{(h_{a})}$:%
\begin{equation}
\text{\textsc{b}}_{+,\text{\textbf{h}}}(\lambda )\equiv \frac{\kappa
_{+}e^{\tau _{+}}\sinh (2\lambda +\eta )\left( -1\right) ^{\mathsf{N}}}{2^{%
\mathsf{N}}\sinh \zeta _{-}}\prod_{a=1}^{\mathsf{N}}\left( \cosh 2\lambda
-\eta _{a}^{(h_{a})}\right) ,
\end{equation}%
and then we have:%
\begin{align}
& \prod_{a=n}^{\mathsf{N}}\text{\textsc{b}}_{+,\{0,...,0,h_{n+1},...,h_{%
\mathsf{N}}\}}(\zeta _{a}^{(1)})V(\eta _{1}^{(0)},...,\eta _{n}^{(0)},\eta
_{n+11}^{(h_{n+1})},...,\eta _{\mathsf{N}}^{(h_{\mathsf{N}})})  \notag \\
& =\frac{\left( \frac{\left( -1\right) ^{\mathsf{N}}\kappa _{+}e^{\tau
_{+}}\sinh \eta }{2^{\mathsf{N}}\sinh \zeta _{+}}\right) ^{n}}{V(\eta
_{1}^{(1)},...,\eta _{n}^{(1)})}\prod_{a=n}^{\mathsf{N}}\sinh 2\xi _{a}\sinh
2(\xi _{a}+\eta )  \notag \\
& \times V(\eta _{1}^{(1)},...,\eta _{n}^{(1)},\eta _{1}^{(0)},...,\eta
_{n}^{(0)},\eta _{n+1}^{(h_{n+1})},...,\eta _{\mathsf{N}}^{(h_{\mathsf{N}%
})}).
\end{align}%
Using this last formula and taking the scalar product we obtain our result.
\end{proof}

\section{Conclusion and outlook}
We have analyzed the integrable quantum models associated to the transfer matrices corresponding to one general non-diagonal and one diagonal or triangular
boundary matrices. For these integrable quantum models, defining the open spin 1/2 XXZ quantum chain in the same
class of non-diagonal boundary matrices for the homogeneous limit, we have obtained the complete SOV-characterization of the transfer matrix eigenvalues and eigenstates, the proof of the simplicity of the spectrum and determinant formulae of $\mathsf{N}\times\mathsf{N}$ matrices for the scalar products of separate states. Finally, matrix elements of a class of quasi-local operators have been computed on the transfer matrix eigenstates in determinant form by using the reconstruction of these operators by the Sklyanin's quantum separate variables. The relevance of these findings in the framework of the non-equilibrium
systems like the partial asymmetric simple exclusion processes (PASEP) will
be described in \cite{ADMFKN12-4} and for the most general
symmetric simple exclusion processes in \cite{ADMFKN12-5}, where moreover further matrix elements of quasi-local operators will be computed.

In the literature of quantum integrable models there exist different applications of separation of variable methods for computing the matrix elements of local operators. An important example is presented in the Smirnov's paper \cite{ADMFKSm98}, where determinant formulae for the matrix elements of a conjectured basis of local operators have been derived in Sklyanin's SOV framework for the quantum integrable Toda chain \cite{ADMFKSk1}. There is a strong analogy among Smirnov's formulae, those that we have here derived and more in general those which appear in the series of papers \cite{ADMFKN12-0}, \cite{ADMFKN12-1,ADMFKN12-3} and \cite{ADMFKGMN12-SG,ADMFKGMN12-T2}. The main differences in all these formulae are due to model dependent features, like the nature of the spectrum of the quantum separate variables. In fact, it is worth citing also the results of the papers \cite{ADMFKBBS96,ADMFKBBS97} on the form factors of the restricted sine-Gordon at the reflectionless points in the S-matrix formulation\footnote{See \cite{ADMFKA.Zam77}-\cite{ADMFKM92} and references
therein.}. The form factors there derived\footnote{Note that recently in \cite{ADMFKJMS11-03} these results have been connected to the important achievements obtained  in \cite{ADMFKJMS11-03}-\cite{ADMFKJMS11-02} where a fermionic basis of
quasi-local operators has been introduced in the infinite volume limit of the XXZ spin 1/2 chain.} can be represented once again as determinants and the connection with SOV emerges on the basis of the semi-classical analysis of \cite{ADMFKBBS96}, there also used as a tool to overcome the problem\footnote{Let us recall that this is a longstanding problem in the S-matrix formulation. The description of massive IQFTs as (superrenormalizable) perturbations of conformal field theories \cite{ADMFKVi70}-\cite{ADMFKDFMS97} by relevant local fields \cite{ADMFKZam88}-\cite{ADMFKGM96} has been at the origin of the attempt of classifying the local field content of massive theories (the set of the solutions to the form factor equations \cite{ADMFKKW78,ADMFKSm92}) by that of the corresponding ultraviolet conformal field theories. Several results are known which confirm this characterization, see for example \cite{ADMFKCM90}-\cite{ADMFKJMT03} and the series of works \cite{ADMFKDN05-1}-\cite{ADMFKDN08}.} of the local fields identification.

Let us comment that in this paper we have followed a different approach for the reconstruction
of local operators w.r.t that used in \cite{ADMFKKKMNST07,ADMFKKKMNST08}. A part the different
framework, ABA in \cite{ADMFKKKMNST07,ADMFKKKMNST08} and SOV in this paper, we have decided to reconstruct
local operators directly by the quantum separate variables of the 6-vertex
reflection algebra and not in terms of those of the 6-vertex Yang-Baxter
algebra. The main motivation to do so is related to the increased complexity
of the functional relations among the generators of these two algebras in
the general non-diagonal cases which make more complicated compute the
action of the quantum separate variables of the 6-vertex Yang-Baxter algebra
on the eigenstates of the 6-vertex reflection algebra transfer matrices. Our
current approach is of course more natural as the action of the quantum
separate variables of the reflection algebra on the corresponding transfer
matrix eigenstates has a simpler form. However, it is worth commenting that
this reconstruction program is not yet completed as we have so far
constructed explicitly only some classes of quasi-local operators by our approach but
we believe possible to use this type of reconstruction to compute all matrix
elements of local operators. One important motivation is to derive form
factors of local operators expressed by determinant formulae as it was obtained in \cite{ADMFKN12-0} for the 6-vertex transfer matrix with antiperiodic boundary conditions.
Indeed, the knowledge of the form factors of local operators is an important
step toward the complete solution of the quantum model as the form factors
represent an efficient numerical tools for the computation of two point correlation functions. In fact, we can rewrite correlation functions in spectral series of form factors and then we can try to use the
same approach developed in \cite{ADMFKCM05} in the ABA framework and used in the
series of works\footnote{The dynamical structure factors, important physical observables measurable by neutron scattering experiments \cite{ADMFKBloch36}-\cite{ADMFKBalescu75}, were computed by this method.} \cite{ADMFKCHM05}-\cite{ADMFKCCS07}. This is a concrete project as also in our SOV
framework will be possible to have representations for the scalar products and
complete characterization of the transfer matrix spectrum in terms of solutions of a system of Bethe equations.

Finally, let us comment that the analysis developed in this paper define the
required setup to extend the results on the spectrum characterization and
the scalar product formulae in the SOV framework to the most general
non-diagonal spin-1/2 open XXZ and XYZ quantum chains. Indeed, the so-called
gauge transformations\footnote{%
These gauge transformations has been first introduced by Baxter \cite{ADMFKBa72-1}-\cite{ADMFKBa72-3}.} can be used also in the reflection algebra framework to reduce the spectral problem to one analyzable by SOV. More in
details, both the transfer matrices of 8-vertex and 6-vertex reflection
algebras associated to the most general integrable boundaries matrices can
be reduced by gauge transformations to those of a dynamical 6-vertex
reflection algebra of elliptic and trigonometric type, respectively, with
one triangular boundary matrix. The implementation of the SOV analysis for
the spectral problem of these dynamical 6-vertex systems is currently under
study in collaboration with N. Kitanine and it consists in the
generalization to the dynamical case of the SOV results derived in this
paper for the standard reflection algebra. It is then worth mentioning that
in \cite{ADMFKN12-3, ADMFK?NT12} it will be shown as the SOV results from the
spectrum up to the form factors of local operators can be extended
from the standard Yang-Baxter algebra to the dynamical one.
\bigskip

{\bf Acknowledgments}\, The author would like to thank N. Kitanine, K. K. Kozlowski, J. M. Maillet, B. M.  McCoy, V. Terras for their interest in this work. The author is supported by National Science Foundation grants PHY-0969739 and gratefully acknowledge the YITP Institute of Stony Brook, where he had the opportunity to develop his research programs and the privilege to have stimulating discussions with B. M. McCoy. Finally, the author would like to thank the Theoretical Physics Group of the Laboratory of Physics at ENS-Lyon and the Mathematical Physics Group at IMB of the Dijon University for their hospitality under the support ANR-10-BLAN-0120-04-DIADEMS.

\begin{small}

\end{small}

\end{document}